\newif\iffull
\fulltrue
\newif\ifcmts
\cmtstrue

\newif\ifsodafinal
\sodafinalfalse

\ifsodafinal
\documentclass[twoside,leqno,twocolumn]{article}
\usepackage{ltexpprt}
\def\url#1{\texttt{#1}}
\usepackage{url}
\else

  \iffull
  \documentclass[11pt,a4paper]{article}
   \usepackage{amsthm}
\usepackage{ifthen}
\newboolean{MareksFormat}
\InputIfFileExists{MareksSwitch}{}{}
\ifthenelse{\boolean{MareksFormat}}{
 \usepackage{geometry}
\geometry{papersize={135mm,180mm},top=2.5mm,bottom=2.5mm,left=1.5mm,right=1.5mm}
}{ \usepackage[a4paper,margin=1in]{geometry}}

  \else
  \documentclass[11pt,a4paper]{article}
  \usepackage[top=1in,bottom=1in,left=1in,right=1in]{geometry}
  \usepackage{amsthm}
  \fi
  \usepackage[colorlinks=true,citecolor=blue,urlcolor=blue,linkcolor=blue,bookmarksopen=false,unicode=true]{hyperref}
\fi

\usepackage{amssymb,amsmath}

\usepackage{bbm}
\usepackage{color}
\usepackage{graphicx}

\iffull
\usepackage{tikz}
\usetikzlibrary{matrix,arrows}
\usepackage[all]{xy}
\fi

\ifsodafinal
\usepackage{tikz}
\usetikzlibrary{matrix,arrows}
\usepackage[all]{xy}

\else

\theoremstyle{plain}
\newtheorem{theorem}{Theorem}[section]
\newtheorem{lemma}[theorem]{Lemma}
\newtheorem{cor}[theorem]{Corollary}
\newtheorem{fact}[theorem]{Fact}
\newtheorem{proposition}[theorem]{Proposition}

\theoremstyle{definition}

\newtheorem{Remarks}[theorem]{Remarks}
\newtheorem{Remark}[theorem]{Remark}

\fi

\DeclareMathOperator{\im}{im} 

\newcommand\bmi{\boxminus}
\newcommand{\+}{\boxplus}

\newcommand{\kk}{k}

\DeclareMathOperator{\SM}{SMap}
\DeclareMathOperator{\M}{Map}

\newcommand{\thedim}{d}

\newcommand{\bzero}{\boldsymbol{0}}  
\newcommand{\oo}{o}                   
\renewcommand\:{\colon}

\def\EML{\mathrm{EML}}
\def\AW{\mathrm{AW}}
\def\SHI{\mathrm{SHI}}
\newcommand{\nonadd}{a}
\newcommand{\Nonadd}{A}
\DeclareMathOperator{\sig}{sig}
\DeclareMathOperator{\coker}{coker}
\newcommand{\susp}[1]{S #1}

\newcommand{\Z}{\mathbbm{Z}}

\newcommand{\R}{\mathbbm{R}}

\newcommand{\RR}{\mathcal{R}}

\newcommand{\BB}{\mathcal{B}}
\renewcommand{\AA}{\mathcal{A}}
\newcommand{\CC}{\mathcal{C}}
\newcommand{\rez}{\xi}

\newcommand\makevec[1]{{\bf #1}}
\def \aa {\makevec{a}}
\def \bb {\makevec{b}}
\def \cc {\makevec{c}}
\def \dd {\makevec{d}}
\def \xx {\makevec{x}}
\def \yy {\makevec{y}}
\def \ww {\makevec{w}}

\newcommand\bsigma{\boldsymbol{\sigma}}
\newcommand\btau{\boldsymbol{\tau}}

\DeclareMathOperator{\id}{id}   
\newcommand\NULLHOA{{\tt NullHom}}
\DeclareMathOperator{\ind}{ind}

\long\def\onefigure#1#2{
\begin{figure*}[tbp]
\begin{center}
#1
\end{center}
\caption{#2}
\end{figure*}
}
\def\immediateFigure#1{%
\smallskip\begin{center}#1\end{center}\smallskip }
\newcommand{\labfig}[2]  
{\onefigure{\mbox{\includegraphics{#1}}}{\label{f:#1} #2} }

\newcommand{\labfigw}[3]  
{\onefigure{\mbox{\includegraphics[width=#2]{#1}}}{\label{f:#1} #3}}

\newcommand{\immfig}[1]  
{\immediateFigure{\mbox{\includegraphics{#1}}}}

\newcommand{\immfigw}[2] 
{\immediateFigure{\mbox{\includegraphics[width=#2]{#1}}}}

\def\indef#1{\emph{#1}}

\newcommand{\heading}[1]{\vspace{1ex}\par\noindent{\bf\boldmath #1}}
\newcommand{\subh}[1]{\vspace{1ex}\par\noindent{\emph{#1}}}

\ifcmts
\newcommand{\marrow}{\marginpar{\boldmath$\longleftarrow$}}
\newcommand{\uli}[1]{\ifhmode\newline\fi\marrow \textsf{*** (ULI: ) #1}}
\newcommand{\jirka}[1]{\ifhmode\newline\fi\marrow \textsf{*** (JIRKA: ) #1\newline}}
\newcommand{\marek}[1]{\ifhmode\newline\fi\marrow \textsf{*** (MAREK: ) #1\newline}}
\newcommand{\martin}[1]{\ifhmode\newline\fi\marrow \textsf{*** (MARTIN: ) #1\newline}}

\else
\newcommand{\marrow}{}
\newcommand{\uli}[1]{}
\newcommand{\jirka}[1]{}
\newcommand{\marek}[1]{}
\newcommand{\martin}[1]{}

\fi

\def\kamsymb{{\rm b}}
\def\istsymb{{\rm c}}
\def\ethsymb{{\rm d}}
\def\grasymb{{\rm *}}
\def\massymb{{\rm a}}
\def\grensymb{{\rm e}}

\title{Computing all maps into a sphere$^\grasymb$
}
\author{Martin \v{C}adek$^\massymb$
\and Marek Kr\v{c}\'al$^{\kamsymb,\istsymb}$ \and
Ji\v{r}\'{\i} Matou\v{s}ek$^{\kamsymb,\ethsymb}$
\and Francis Sergeraert$^\grensymb$ \and Luk\'a\v{s}
Vok\v{r}\'{\i}nek$^\massymb$
 \and Uli Wagner$^{\istsymb}$}

\begin{document}

\maketitle

{\renewcommand\thefootnote{\grasymb}
\footnotetext{
The research by M.\,\v{C}. and L.\,V.~was supported by a Czech Ministry of Education grant
(MSM 0021622409).
The research by M.\,K. was supported by project GAUK 49209.
The research by J.\,M.\ and M.\,K.\ was also supported
by project 1M0545 by the Ministry of Education of the Czech Republic and by Center of Excellence -- Inst.\ for Theor.\
       Comput.\ Sci., Prague (project P202/12/G061 of GA~\v{C}R).
The research by J.\,M. was also supported by the  ERC Advanced Grant No.~267165.
The research by U.\,W. was supported
 by the Swiss National Science Foundation (SNF Projects 200021-125309, 200020-138230, and PP00P2-138948).
 }
}

{\renewcommand\thefootnote{\massymb}
\footnotetext{Department of Mathematics and Statistics,
 Masaryk University, Kotl\'a\v{r}sk\'a~2, 611~37~Brno,
 Czech Republic}
}

{\renewcommand\thefootnote{\kamsymb}
\footnotetext{Department of Applied Mathematics,
Charles University, Malostransk\'{e} n\'{a}m.~25,
118~00~~Praha~1,  Czech Republic}
}
{\renewcommand\thefootnote{\istsymb}
\footnotetext{IST Austria, Am Campus 1, 3400~Klosterneuburg, Austria,
}
}

{\renewcommand\thefootnote{\ethsymb}
\footnotetext{Department of Computer Science,
ETH Z\"urich, 8092~Z\"urich, Switzerland}
}
{\renewcommand\thefootnote{\grensymb}
\footnotetext{Institut Fourier,
BP~74,
38402~St Martin, d'H\`eres Cedex,
France}
}


\begin{abstract}
Given topological spaces $X,Y$, a fundamental problem of algebraic
topology is understanding the structure of all continuous maps $X\to Y$.
We consider a computational version, where $X,Y$ are
given as finite simplicial complexes, and the goal is to compute
$[X,Y]$, i.e., all homotopy classes of such maps.

We solve this problem in the \emph{stable range}, where
for some $d\ge 2$, we have $\dim X\le 2d-2$ and  
$Y$ is \emph{$(d-1)$-connected}; in particular,  $Y$ can be the 
$d$-dimensional sphere $S^d$. The algorithm combines classical tools 
and ideas from homotopy theory (
\emph{obstruction theory}, \emph{Postnikov systems},
 and \emph{simplicial sets}) with algorithmic tools from effective algebraic 
topology (\emph{locally effective simplicial sets} 
and \emph{objects with effective homology}).

In contrast, $[X,Y]$ is known to be uncomputable for general $X,Y$,
since for $X=S^1$ it includes a well known undecidable
problem: testing triviality of the fundamental group of~$Y$. 

In follow-up papers, the algorithm is shown to run in polynomial time 
for $d$ fixed, and extended to other problems, such as the
\emph{extension problem}, where we are given a subspace $A\subset X$
and a map $A\to Y$ and ask whether it extends to a map $X\to Y$,
or computing the \emph{$\Z_2$-index}---everything in the stable
range. Outside the stable range, the extension problem
is undecidable.




\end{abstract}

\section{Introduction}

Among results concerning computations in topology, probably the most
famous ones are negative.
For example, there is no algorithm to decide
whether the fundamental group $\pi_1(Y)$ of a given space $Y$ is trivial,
i.e., whether every loop in $Y$ can be continuously contracted to 
a point.\footnote{This follows by a standard reduction,
see, e.g., \cite{Stillwell}, from a result of Adjan and Rabin 
on unsolvability of the triviality problem of a group given in terms
of generators and relation; see, e.g.,~\cite{Soare:ComputabilityDifferentialGeometry-2004}.}

Here we obtain a positive result for a closely related and fairly
general problem, homotopy classification of maps;\footnote{The definition
of homotopy and other basic topological notions will be recalled later.}
namely, we describe an algorithm that works in the so-called \emph{stable range}.


\heading{Computational topology. }
This paper falls into the broader area of \emph{computational topology}, which has been a rapidly developing 
discipline in recent years---see, for instance, the textbooks \cite{EdelsbrunnerHarer:ComputationalTopology-2010,Zomorodian:TopologyComputing-2005,Matveev:AlgorithmicTopology-2007}. 

Our focus is somewhat different from the main current trends 
in the field, where, on the one hand, computational questions
are intensively studied in dimensions $2$ and $3$ (e.g., concerning 
\emph{graphs on surfaces}, \emph{knots} or \emph{$3$-manifolds}\footnote{A seminal early result in the latter direction is Haken's famous algorithm for recognizing the unknot \cite{Haken:TheorieNormalflaechen-1961}.}), and, on the 
other hand, for arbitrary dimensions mainly \emph{homology} computations 
are investigated. 

Homology has been considered an inherently computational 
tool since its inception and there are many software packages that contain practical implementations, e.g., {\tt polymake}~\cite{polymake}. 
Thus, algorithmic solvability of homological questions is
usually obvious, and the challenge may be, e.g., designing
very fast algorithms to deal with large inputs. Moreover,
 lot of research has been 
devoted to developing extensions such as \emph{persistent homology} 
\cite{EdelsbrunnerHarer:PersistentTopologySurvey-2008}, motivated by applications  like data analysis~\cite{Carlsson:TopologyData-2009}.

In contrast, \emph{homotopy-theoretic problems}, as those studied here, are 
generally considered much less tractable than homological ones and
the first question to tackle is usually the existence of any algorithm at all
(indeed, many of them are algorithmically unsolvable,
as the example of triviality of the fundamental group illustrates). 
Such problems lie at the core of algebraic topology and have been 
thoroughly studied from a topological perspective since  the 1940s. 
A significant effort has also
been devoted to \emph{computer-assisted} concrete calculations, most notably
of higher \emph{homotopy groups of spheres}; see, e.g.,
\cite{Kochman:Stable-homotopy-groups-of-spheres-1990}. 

\heading{Effective algebraic topology. }
In the 1990s, three independent groups of researchers proposed general 
frameworks to make various more advanced methods of algebraic topology 
\emph{effective} (algorithmic):
Sch\"on \cite{Schoen-effectivetop}, Smith
\cite{smith-mstructures}, and Sergeraert, Rubio, Dousson, and Romero
 (e.g., \cite{Sergeraert:ComputabilityProblemAlgebraicTopology-1994,RubioSergeraert:ConstructiveAlgebraicTopology-2002,RomeroRubioSergeraert,SergRub-homtypes};
also see \cite{SergerGenova} for an exposition).
These frameworks yielded general \emph{computability} results 
for homotopy-theoretic questions (including
new algorithms for the computation of higher homotopy groups \cite{Real96}),
and in the case of Sergeraert and co-workers, a \emph{practical implementation}
as well.

The problems considered by us
were not addressed in those papers, but we rely on the work of Sergeraert
 et al., and in particular on their framework of 
\emph{objects with effective homology}, for implementing certain 
operations in our algorithm 
(see Sections~\ref{s:outline} and~\ref{sec:prelim}).

We should also mention that our perspective is somewhat different
from the previous work in effective algebraic topology, closer to the view of
theoretical computer science; although in the present paper
we provide only computability results, subsequent work
 also addresses the computational complexity of the considered problems.
We consider this research area fascinating,
and one of our hopes is that our work  
may help to bridge the cultural gap between algebraic topology and theoretical 
computer science.

\heading{The problem: homotopy classification of maps. }  
A central theme in algebraic topology is to understand,
for given topological spaces $X$ and $Y$, the set $[X,Y]$ of homotopy
classes of maps\footnote{%
In this paper, all maps between topological spaces are assumed
to be continuous. Two maps $f,g\colon X \to Y$ are said to be \emph{homotopic},
denoted $f \sim g$, if there is a map $F \colon X \times [0,1] \to Y$
such that $F(\cdot,0)=f$ and $F(\cdot,1)=g$. The equivalence class of $f$
of this relation is denoted $[f]$ and called the \emph{homotopy class} of $f$.}
from $X$ to $Y$.

Many of the celebrated results throughout the history of topology
can be cast as information about $[X,Y]$ for particular spaces 
$X$ and~$Y$.  An early example is a famous theorem
of Hopf from the 1930s, asserting that the homotopy class of a map
$f\colon S^n \to S^n$, where $S^n$ is the $n$-dimensional sphere,
is completely determined by an integer called the \emph{degree} of $f$,
thus giving a one-to-one correspondence $[S^n,S^n] \cong \Z$.
Another great discovery of Hopf, 
with ramifications in modern physics and elsewhere, was a
map $S^3\to S^2$, now called by his name, that is not homotopic
to a constant map.

These two early results concern \emph{higher homotopy
groups}: for our purposes, the \emph{$k$th homotopy group} $\pi_k(Y)$, $k\ge 2$, of
a space $Y$ can be identified with the set $[S^k,Y]$ equipped with
a suitable group operation.\footnote{Formally, the $k$th homotopy 
group $\pi_k(Y)$ of a space $Y$, $k\geq 1$, is defined as the set of all homotopy classes
of \emph{pointed} maps $f\:S^k\to Y$, i.e., maps $f$ that send a distinguished point $s_0\in S^k$ 
to a distinguished point $y_0\in Y$ (and the homotopies $F$ also satisfy $F(s_0,t)=y_0$ for all $t\in[0,1]$). 
Strictly speaking, one should write $\pi_k(Y,y_0)$ but for a path-connected $Y$, the choice of
$y_0$ does not matter. Furthermore, $\pi_k(Y)$ is trivial (has only one element) iff $[S^k,Y]$ is trivial, i.e., if every map $S^k\to Y$ is homotopic to a constant map. Moreover, if $\pi_1(Y)$ is trivial, then for
$k\geq 2$, the pointedness of the maps does not matter and one can identify $\pi_k(Y)$ with $[S^k,Y]$.
Each $\pi_k(Y)$ is a group, which for $k\ge 2$ is Abelian, but the definition of the group operation is not 
important for us at the moment.}  In particular, a very important special case are the higher \emph{homotopy groups of spheres} 
$\pi_k(S^n)$, whose computation has been one of the important challenges propelling research in algebraic topology, with only  
partial results so far despite an enormous effort (see, e.g., \cite{Ravenel,Kochman}).

\heading{The extension problem.} A problem closely related to computing $[X,Y]$ is the \emph{extension problem}: given a subspace $A\subset X$ and a map $f\:A\to Y$, can it be extended to a map $X\to Y$? For example, the famous \emph{Brouwer fixed-point theorem} can be re-stated
as non-extendability of the identity map $S^n\to S^n$ to the ball $D^{n+1}$. A number of topological concepts, which may seem quite advanced and esoteric to a newcomer in algebraic topology, e.g.\ \emph{Steenrod squares}, have a natural motivation in trying to solve the extension problem step by step.

\heading{Early results.} 
Earlier developments around the extension problems are described
 in Steenrod's paper \cite{Steenrod:CohomologyOperationsObstructionsExtendingContinuousFunctions-1972}\iffull \ (based on
a 1957 lecture series)\fi, which we can recommend, for readers with a moderate topological
background, as  an exceptionally clear and accessible, albeit somewhat outdated,
introduction to this area. In particular, in that paper, Steenrod asks for an effective
procedure for (some aspects of) the extension problem.

There has been a tremendous amount of work in homotopy theory
since the 1950s, with a wealth of new concepts and results,
some of them opening completely new areas. 
However, as far as we could find out,
the \emph{algorithmic part} of the program discussed in
 \cite{Steenrod:CohomologyOperationsObstructionsExtendingContinuousFunctions-1972} 
 has not been explicitly carried out until now.

As far as we know, the only algorithmic paper addressing the general problem of computing of $[X,Y]$ is that by Brown
\cite{Brown} from 1957. Brown showed that $[X,Y]$ is computable under the assumption that
$Y$ is 1-connected\footnote{A space $Y$ is said to be \emph{$k$-connected}
if every map $S^i\to Y$ can be extended to $D^{i+1}$,
the ball bounded by the spheres $S^i$, for $i=0,1,\ldots,k$. Equivalently,
$Y$ is path-connected and the first $k$ homotopy groups $\pi_i(Y)$, $i\leq k$, are trivial.}
 and all the higher homotopy groups $\pi_k(Y)$, $2\le k\le\dim X$,
are \emph{finite}. The latter assumption is rather strong\footnote{Steenrod~\cite{Steenrod:CohomologyOperationsObstructionsExtendingContinuousFunctions-1972} calls this restriction  
``most severe,'' and conjectures that it ``should ultimately be unnecessary.''};
  in particular, Brown's algorithm is not applicable
for $Y=S^d$ since $\pi_d(S^d)\cong \Z$. 

In the same paper, Brown also gave an algorithm for computing
$\pi_k(Y)$, $k\ge 2$, for every 1-connected~$Y$. 
To do this, he overcame the restriction on finite
homotopy groups mentioned above, and also discussed in 
Section~\ref{s:outline} below, by a somewhat ad-hoc
method, which does not seem to generalize to the $[X,Y]$ setting.

On the negative side, it is undecidable whether $[S^1,Y]$ is trivial
(since this is equivalent to the triviality of $\pi_1(Y)$). By an equally
classical result of Boone and of Novikov
\cite{Boone:SimpleUnsolvableProblemsGroupTheory1-1954,Boone:SimpleUnsolvableProblemsGroupTheory2-1954,Boone:SimpleUnsolvableProblemsGroupTheory3-1955,Novikov:UndecidabilityWordProblem-1955}
it is undecidable whether a given map $S^1\to Y$ can be extended to 
a map $D^2\to Y$, even
if $Y$ is a finite $2$-dimensional simplicial complex. Thus, both 
the computation $[X,Y]$ and the extension problem are algorithmically 
unsolvable without additional assumptions on $Y$.
These are the only previous undecidability results in this context 
 known to us; more recent results, obtained as a follow-up
of the present paper, will be mentioned later.
For a number of more loosely related undecidability
results we refer to
\cite{Soare:ComputabilityDifferentialGeometry-2004,
NabutovskyWeinberger:AlgorithmicAspectsHomeomorphismProblem-1999,
NabutovskyWeinberger:AlgorithmicUnsolvabilityTrivialityProblemMultidimensionalKnots-1996} and the references therein.

\heading{New results. } 
In this paper we prove the computability of $[X,Y]$ under 
a fairly general condition on $X$ and $Y$. Namely,
we assume that, for some integer $d\ge 2$, we have $\dim X\le 2d-2$,
while $Y$ is $(d-1)$-connected.
A particularly important example of a $(d-1)$-connected space,
often encountered in applications, is the sphere $S^d$.
We also assume that $X$ and $Y$ are given as
finite simplicial complexes or, more generally, as finite \emph{simplicial sets} (a more flexible generalization
of simplicial complexes; see Section~\ref{sec:prelim}).

An immediate problem with computing the set $[X,Y]$ of all homotopy classes
of continuous maps 
is that it may be infinite. However, it is known
that under the just mentioned conditions on $X$ and $Y$,
$[X,Y]$ can be endowed with a structure of a finitely generated
Abelian group.\footnote{In particular, the groups $[X,S^d]$ are known
as the \emph{cohomotopy groups} of $X$; see \cite{HuBook}.}
Our algorithm computes the isomorphism type of this Abelian group.


\begin{theorem}\label{t:main}
Let $d\ge 2$. There is an algorithm that, given finite simplicial
complexes (or finite simplicial sets) $X,Y$, where
 $\dim X\leq 2d-2$ and  $Y$ is $(d-1)$-connected, 
computes the isomorphism type of the Abelian group $[X,Y]$, i.e., expresses it as
a direct product of cyclic groups. 

Moreover, given a simplicial map $f\:X\to Y$, the
element of the computed direct product corresponding to $[f]$
can also be computed. Consequently,
it is possible to test homotopy of simplicial maps $X\to Y$.
\end{theorem} 

We remark that the algorithm does not need any certificate
of the $1$-connectedness of $Y$, but if $Y$ is not $1$-connected,
the result may be wrong.
\smallskip

In the remainder of the introduction, we discuss related results, applications,  general motivation for our work, and directions for future research. 
In Section~\ref{s:outline}, we will present an outline of the methods and of the algorithm.
In Sections~\ref{s:abelops}--\ref{s:implementing-grp-op}, we will introduce and discuss
the necessary preliminaries, and then we
present the algorithm in detail in Section~\ref{sec:main-algo}.

\heading{Follow-up work. } 
We briefly summarize a number of strengthenings and extensions
of Theorem~\ref{t:main}, as well as complementary hardness results,
obtained since the original  submission of this paper. They
will appear in a series of follow-up papers.

\subh{Running time.} 
In the papers \cite{polypost,pKZ1}
it is shown that, for every fixed $d$, the algorithm 
as in Theorem~\ref{t:main} can be implemented
so that its running time is bounded by a \emph{polynomial} 
in the size of $X$ and $Y$.\footnote{Here, for simplicity, we can define the size of a finite simplicial complex $X$ as the number of its simplices; for a simplicial set, we count only nondegenerate simplices. 
It is not hard to see that if the dimension of $X$ is bounded by a constant, 
then $X$ can be encoded by a string of bits of length polynomial in the number of (nondegenerate) simplices; also see the discussion in \cite{polypost}.} 
The nontrivial part of this polynomiality result
is a subroutine for computing \emph{Postnikov systems},
which we use as a black box here---see Section~\ref{s:outline}. 
For the rest of the algorithm, verifying polynomiality is straightforward, see \cite{Krcal-thesis};
except for some brief remarks, we will not consider this issue here, in order
to avoid distraction from the main topic.


\subh{The extension problem.} In \cite[Theorem~1.4]{polypost}, it is shown
 that the methods of the present paper also yield 
an algorithm for the extension problem 
as defined above. The extension problem can actually
be solved even for $\dim X\leq 2d-1$, as opposed to $2d-2$ in Theorem~\ref{t:main}
(still asumming that $Y$ is $(d-1)$-connected).
Again, the running time is polynomial for $d$ fixed.

\subh{Hardness outside the stable range. } The dimension and connectivity assumptions in Theorem~\ref{t:main} turned out to be essential and almost sharp,
in the following sense:
In \cite{ext-hard}, it is shown that, for every $d\ge 2$,
the extension problem  is undecidable for $\dim X=2d$ and
$(d-1)$-connected $Y$. Similar arguments show that for
$\dim X=2d$ and $(d-1)$-connected $Y$, deciding whether every map $X\to Y$ is
homotopic to a constant map (i.e., $|[X,Y]|=1$) is NP-hard 
and no algorithm is known for it \cite[Theorem~2.1.2]{Krcal-thesis}.

\subh{Dependence on $d$. }
The running-time of the algorithm in Theorem~\ref{t:main}
can be made polynomial for every fixed $d$, as was mentioned
above, but it depends on $d$ at least exponentially.
We consider it unlikely that the problem can be solved by an algorithm whose running time also depends polynomially on $d$. One heuristic
reason supporting this belief
is that Theorem~\ref{t:main} includes the computation of the \emph{stable} homotopy groups $\pi_{d+k}(S^d)$, $k\le d-2$. These are considered mathematically
very difficult objects, and a polynomial-time algorithm for computing
them would be quite surprising.
 Another reason is that the related problem of 
computing the higher homotopy groups $\pi_k(Y)$ of a $1$-connected simplicial
complex $Y$ was shown to be \#P-hard if $k$, encoded in unary, is a part
of input \cite{Anick-homotopyhard,ext-hard}, and it is W[1]-hard
w.r.t.~the parameter $k$ \cite{Mat-homotopyW1}, even for $Y$ of dimension~$4$. Still, it would
be interesting to have more concrete hardness results for the setting
of Theorem~\ref{t:main} with variable~$d$.


\subh{Lifting-extension and the equivariant setting. } 
In~\cite{aslep,Vokrinek}, 
the ideas and methods of the present paper are further developed and generalized to more general \emph{lifting-extension problems} 
and to the \emph{equivariant} setting, where a fixed finite group $G$
acts freely on both $X$ and $Y$, and the considered continuous maps are 
also required to be \emph{equivariant}, i.e.,  
to commute with the actions of $G$. 
The basic and important special case with $G=\Z_2$ will be 
discussed in more detail below.

\subh{Homotopy testing. } By Theorem~\ref{t:main}, it is possible to
test homotopy of two simplicial maps $X\to Y$ in the stable range.
It turns out that for this task, unlike for the extension problem,
 the restriction to the stable range is unnecessary: it suffices to assume that
 $Y$ is $1$-connected~\cite{VokriFil-homotopic}.

\heading{Applications, motivation, and future work.}
We consider the fundamental nature of the
 algorithmic problem of computing $[X,Y]$ a sufficient
motivation of our research.
However, we also hope that work in this area will bring
various  connections and applications, also in other fields,
possibly including practically usable software, e.g., for aiding
research in topology. Here we mention two applications that
have already been worked out in detail.

\subh{Robust roots. }
A nice concrete application comes from the  
so-called ROB-SAT problem---robust satisfiability of systems of equations 
The problem is given by a rational value $\alpha>0$ and a piecewise linear
 function $f\:K\to \R^d$ defined by rational values on the vertices of a 
simplicial complex $K$.  The question is whether an arbitrary continuous 
$g\:K\to\R^d$ that is at most $\alpha$-far from $f$ (i.e., \(\|f-g\|_\infty\leq \alpha\)) has a root. In a slightly different and more special form,
this problem was investigated by Franek et al.~\cite{Franek-al},
and later Franek and Kr\v{c}\'{a}l 
\cite{FranekKrcal:RobustSatisfiability} exhibited
 a computational equivalence of ROB-SAT and the extension problem for 
maps into the sphere \(S^{d-1}\). 
The algorithm for the extrendability problem based
on the present paper then yields an algorithmic solution when $\dim K\leq 2d-3$.

\subh{$\Z_2$-index and embeddability.}
An important motivation for the research leading to
the present paper was the computation
of the  \emph{$\Z_2$-index}  (or \emph{genus}) $\ind(X)$ of a
$\Z_2$-space $X$,\footnote{A \emph{$\Z_2$-space} is a topological space $X$
with an action of the group $\Z_2$; the action is described by a homeomorphism $\nu\colon X\to X$ with $\nu\circ\nu=\id_X$. A primary example is a sphere $S^d$ with the antipodal action $x\mapsto-x$. An \emph{equivariant  map} between $\Z_2$-spaces is a continuous map that commutes with the $\Z_2$ actions.} i.e., the smallest $d$
such that $X$ can be equivariantly mapped into~$S^d$. For example, the classical Borsuk--Ulam theorem can be stated in the form $\ind(S^d)\ge d$. 
Generalizing the results in the present paper, \cite{aslep} provided
an algorithm that decides whether $\ind(X)\le d$, provided that $d\geq 2$ 
and $\dim(X)\le 2d-1$; for fixed $d$ the running time is polynomial in the size of~$X$.

The computation of $\ind(X)$ arises, among others, in the problem
of \emph{embeddability} of topological spaces, which is a classical 
and much studied area; see, e.g., the survey by Skopenkov 
\cite{skopenkov-survey}. One of the basic questions here is, 
given a $k$-dimensional finite simplicial complex $K$, 
can it be (topologically) embedded
in $\R^d$? The famous \emph{Haefliger--Weber theorem}
from the 1960s asserts that, in the \emph{metastable range
of dimensions}, i.e., for $k\le \frac23d-1$, embeddability of $K$ in $\R^d$
is equivalent to $\ind(K_\Delta^2)\le d-1$,
where $K_\Delta^2$, the \emph{deleted product} of $K$,
 is a certain $\Z_2$-space constructed from $K$ in a simple manner.
 Thus, in this range, the embedding  problem is, computationally,
a special case of $\Z_2$-index computation. A systematic study of
algorithmic aspects of the embedding problem was initiated in 
\cite{MatousekTancerWagner:HardnessEmbeddings-2011}, and
the metastable range was left as one of the main open problems there
(now resolved as a consequence of \cite{aslep}).

The $\Z_2$-index also appears as
a fundamental quantity in combinatorial applications
of topology. For example, the celebrated result of Lov\'asz
on Kneser's conjecture can be re-stated as
$\chi(G)\ge \ind(B(G))+2$, where $\chi(G)$ is the chromatic
number of a graph $G$, and $B(G)$ is a certain simplicial complex
constructed from $G$ (see, e.g., \cite{Matousek:BorsukUlam-2003}).
We find it striking that prior to \cite{aslep}, \emph{nothing} seems to 
have been known about the computability of such an interesting quantity
as $\ind(B(G))$.

\subh{Explicit maps? }
Our  algorithm for Theorem~\ref{t:main}
works with certain implicit representations
of the elements of $[X,Y]$; it can output a set of generators
of the group in this representation, and it contains a subroutine
implementing the group operation.

It would be interesting to know whether these implicit representations
can be converted into actual maps $X\to Y$ (given, say, as simplicial maps from a sufficiently fine subdivision of $X$ into $Y$) in an effective way. 
Given an implicit representation of a homotopy class $\kappa\in [X,Y]$,
we can compute an explicit map $X\to Y$ in $\kappa$ by a brute force search:
go through finer and finer subdivisions $X'$ of $X$ and through all possible
simplicial maps $X'\to Y$ until a simplicial map in $\kappa$ is found.
Membership in $\kappa$ can be tested using Theorem~\ref{t:main};
this may not be entirely obvious, but we do not give the details here,
since this is only a side-remark.
However, currently we have no upper bound on how fine subdivision
may be required.

This would also be of interest in certain applications such as the 
embeddability problem---whenever we want
to construct an  embedding explicitly, instead of just deciding embeddability. 

Various measures of complexity of embeddings
have been studied in the literature, and very recently, Freedman and
Krushkal \cite{free-kru} obtained bounds
for the \emph{subdivision complexity} of an embedding
$K\to\R^{d}$. Here $d$ and $k=\dim K$ are considered fixed,
and the question is, what is the smallest $f(n)$ 
 such that every $k$-dimensional complex $K$ with $n$ simplices
that is embeddable in $\R^d$ has a subdivision $L$ with at most $f(n)$
simplices that admits a linear embedding in $\R^{d}$
(i.e., an embedding that is an affine map on each simplex of $L$)?
 Freedman and Krushkal essentially solved the case
with $d=2k$ (here the embeddability can be decided
in polynomial time---this is covered by \cite{aslep} but
this particular case goes back to a classical work of Van Kampen from the 1930s;
see \cite{MatousekTancerWagner:HardnessEmbeddings-2011}). 
The  subdivision complexity for the
other cases in the metastable range, i.e., for $k\le \frac23 d-1$, is
 wide open at present, and
obtaining explicit maps $X\to Y$ in the setting
of Theorem~\ref{t:main} might be a key step in its resolution.

\section{An outline of the methods and of the algorithm}
\label{s:outline}

Here we present an overview of the algorithm 
and sketch the main ideas and tools. Everything from this section
will be presented again in the rest of the paper.
Some topological notions are left undefined here and will be 
introduced in later sections.

\heading{The geometric intuition: obstruction theory. }
Conceptually, the basis of the algorithm is classical \emph{obstruction theory} \cite{Eilenberg:CohomologyContinuousMappings-1940}.
For a first encounter, it is probably easier to consider
a version of obstruction theory which proceeds by constructing maps $X\to Y$
 inductively on the $i$-dimensional \emph{skeleta}\footnote{The $i$-skeleton
of a simplicial complex $X$ consists of all simplices
of $X$ of dimension at most~$i$.}
 of $X$, extending them one dimension at a time.
(For the actual algorithm, we use a different, ``dual'' version
of obstruction theory, where we lift maps from $X$ through
stages of a so-called \emph{Postnikov system} of $Y$.)

In a nutshell, at each stage, the extendability of a map
from the $(i-1)$-skeleton to the $i$-skeleton
is characterized by vanishing of a certain
\emph{obstruction}, which can, more or less by known techniques,
be evaluated algorithmically.

Textbook expositions may give the impression that obstruction
theory is a general algorithmic tool for testing the
extendability of maps \iffull (this is actually what some of the
topologists we consulted seemed to assume)\fi. However, the
extension at each step is generally not unique, and
extendability at subsequent steps may depend, in a nontrivial
way, on the choices made earlier. Thus, in principle, one
needs to search an infinitely branching tree of extensions. Brown's
result mentioned earlier, on computing $[X,Y]$ with
the $\pi_k(Y)$'s finite, is based on a complete search of
this tree, where the assumptions on $Y$ guarantee the branching to be finite.

In our setting, we make essential use of the group structure on the set $[X,Y]$ (mentioned in Theorem~\ref{t:main}), 
as well as on some related ones, to produce a finite encoding of the set of all possible extensions at a given stage.

\heading{Semi-effective and fully effective Abelian groups.} The description of our algorithm has several levels.
On the top level, we work with Abelian groups whose elements are homotopy classes of maps. On a lower level,
the group operations and other primitives are implemented by computations with \emph{concrete representatives} of the
homotopy classes; interestingly, on the level of representatives, the operations are generally non-associative. 

We need to be careful in distinguishing ``how explicitly'' the relevant groups are available to us. Specifically, we distinguish
between \emph{semi-effective} and \emph{fully effective} Abelian groups: For the former, we have a suitable way of representing
the elements on a computer and we can compute the various group operations (addition, inverse) on the level of representatives. For the latter, we additionally have a list of generators and relations and we can express a given element in terms of the generators (see Section~\ref{s:abelops}
for a detailed discussion). A homomorphism $f$ between two semi-effective Abelian groups is called \emph{locally effective} if there is an
algorithm that, given a representative of an element $a$, computes a representative of $f(a)$.

\heading{Simplicial sets and objects with effective homology.}
All topological spaces in the algorithm are represented as \emph{simplicial sets}, which  
will be discussed in more detail in Section~\ref{s:simplsets}. Suffice it here to say that a simplicial 
set is a purely combinatorial description of how to build a space from simple building blocks (\emph{simplices}),  
similar to a simplicial complex, but allowing more general ways of gluing simplices together along their faces, which makes many constructions much simpler
and more conceptual.

For the purposes of our exposition we will occasionally talk about 
topological spaces specified in other ways, most notably, as  
\emph{CW-complexes}---e.g., in Sections~\ref{s:postni} 
and~\ref{s:hgroup}. However, we stress that in the algorithm, 
all spaces are represented as simplicial sets.

A finite simplicial set can be encoded explicitly  on a computer 
by a finite bit string, which describes a list of all (nondegenerate) simplices
and the way of gluing them together. 
However, the algorithm also uses a number of \emph{infinite} 
simplicial sets in its computation,
such as simplicial \emph{Eilenberg--MacLane spaces} discussed below.
For these, it is not possible to store the list of all nondegenerate
simplices.

Instead,  we use a general framework
developed by Sergeraert et al.\ (as surveyed, e.g., in \cite{SergerGenova}),
in which a possibly infinite simplicial set is represented by 
a \emph{black box} or \emph{oracle} (we speak of a
 \emph{locally effective} simplicial set).
This means that we have a specified encoding of the simplices of the simplical
set and a collection of algorithms for performing certain operations, such as 
computing a specific face of a given simplex. Similarly, a simplicial map 
between locally effective simplicial sets is \emph{locally effective} 
if there is an algorithm that evaluates it on any given simplex of 
the domain; i.e., given the encoding of an input
simplex, it produces the encoding of the image simplex.

To perform global computations with a given locally effective simplicial set,
e.g.,\ compute its homology and cohomology groups of any given dimension,
the black box representation of these locally 
effective simplicial sets is augmented with 
additional data structures and one
speaks about \emph{simplicial sets with effective homology}.
Sergeraert et al.\ then provide algorithms that 
construct basic topological spaces, 
such as finite simplicial sets or Eilenberg--MacLane spaces,
as simplicial sets with effective homology.
More crucially, the auxiliary data structures of a simplicial set
with effective homology are designed so that if we perform
various topological operations, such as the Cartesian product, 
the \emph{bar construction}, 
the \emph{total space of a fibration}, etc.,
the result is again a simplicial set with effective homology.

\heading{Postnikov systems.}
The target space $Y$ in Theorem~\ref{t:main}
enters the computation in the form of a
\emph{Postnikov system}. 
Roughly speaking, a Postnikov system of a space $Y$ is a
way of building $Y$ from ``canonical pieces'', called \emph{Eilenberg--MacLane
spaces}, whose homotopy structure is the simplest possible,
namely, they have a single non-trivial homotopy group.
The Eilenberg--MacLane spaces occurring in the algorithm will be denoted by
$K_{i}$ and $L_i$, and they depend only on the homotopy groups of~$Y$.

A Postnikov system has \emph{stages} $P_0,P_1,\ldots$,
where $P_i$ reflects the homotopy properties of $Y$ up to
dimension $i$; in particular, $\pi_j(P_i)\cong\pi_j(Y)$
for all $j\le i$, while $\pi_j(P_i)=0$ for $j>i$. The
isomorphisms of the homotopy groups for $j\le i$
are induced by maps $\varphi_i\:Y\to P_i$, which are also a part
of the Postnikov system. Crucially, these maps also induce
bijections $[X,Y]\to [X,P_i]$ whenever $\dim X\leq i$;
in words, homotopy classes of maps $X \to Y$ from any space $X$ of dimension at most $i$
are in bijective correspondence with homotopy classes
of maps $X\to P_i$.

The last component of a Postnikov system are mappings
$\kk_0,\kk_1,\ldots$, where $\kk_{i-1}\colon P_{i-1}\to K_{i+1}$ is called the 
$(i-1)$st \emph{Postnikov class}. Together with the group $\pi_{i}(Y)$, it 
describes how $P_{i}$ is obtained from $P_{i-1}$.

If $Y$ is $(d-1)$-connected, then for $i\le 2d-2$, the Postnikov stage
$P_i$ can be equipped with an $H$-group structure, which is, roughly speaking,
an Abelian group structure ``up to homotopy'' (this is where the connectivity
assumption enters the picture). This $H$-group structure on $P_i$ induces,
in a canonical way, an Abelian group structure on $[X,P_i]$, for every
space $X$, with no restriction on $\dim X$. 

Now assuming $\dim X\le 2d-2$, we have the bijection 
$[X,Y]\to [X,P_{2d-2}]$ as mentioned above, and this can serve as
the definition of the Abelian group structure on $[X,Y]$ used
in Theorem~\ref{t:main}. Therefore, instead of computing 
$[X,Y]$ directly, we actually compute $[X,P_{2d-2}]$, which
yields an isomorphic Abelian group. (However, the elements of $[X,P_{2d-2}]$
are not so easily related to continuous maps $X\to Y$;
this is the cause of the open problem, mentioned in the introduction,
of effectively finding actual maps $X\to Y$ as representatives of the generators.)

Thus, to prove Theorem~\ref{t:main}, we first compute 
the stages $P_0,\ldots,P_{2d-2}$ of a Postnikov system of $Y$, and then,
by induction on $i$, we determine $[X,P_i]$, $i\leq 2d-2$. We return
the description of $[X,P_{2d-2}]$ as an Abelian group.

For the inductive computation of $[X,P_i]$ we 
 do not need any dimension restriction on $X$ anymore, which
is important, because the induction will also involve computing, e.g., 
$[SX,P_{i-1}]$,
where $SX$ is another simplicial set, the  \emph{suspension} of $X$, 
with dimension one larger than that of~$X$.

The stages $P_i$ of the Postnikov system are built as simplicial
sets with a particular property (they are \emph{Kan} simplicial 
sets\footnote{The term \emph{Kan complex} is also 
commonly used in the literature.}),
which ensures that every \emph{continuous} map $X\to P_i$ is homotopic 
to a \emph{simplicial} map. In this way, instead of
the continuous maps $X\to Y$, which are problematic to represent,
we deal only with simplicial maps $X\to P_i$ in the algorithm,
which are discrete, and even finitely representable, objects.


\heading{Outline of the algorithm.} 
\begin{enumerate}
\item As a preprocessing step, we compute, using the algorithm from \cite{polypost}, a suitable representation of the first $2d-2$ 
stages of a Postnikov system for $Y$. We refer to Section~\ref{s:postni} for the full specification of the output provided by this computation;
in particular, we thus obtain the isomorphism types of the first $2d-2$ homotopy groups $\pi_i=\pi_i(Y)$ of $Y$, the Postnikov stages $P_i$ and the Eilenberg--MacLane spaces $L_i$ and $K_{i+1}$, $i\leq 2d-2$, as locally effective simplicial sets, and various maps between these spaces, e.g., the Postnikov classes $\kk_{i-1}\colon P_{i-1}\to K_{i+1}$, as locally effective simplicial maps. 
\item Given a finite simplicial set $X$, the main algorithm computes $[X,P_i]$ as a fully effective Abelian group 
by induction on $i$, $i\leq 2d-2$, and $[X,P_{2d-2}]$ is the desired output.

The principal steps are as follows:
\begin{itemize}
\item  We construct locally effective simplicial maps $\+_i\colon P_i\times P_i \to P_i$ and $\bmi_i\: P_i \to P_i$, $i\leq 2d-2$
(Section~\ref{s:implementing-grp-op}). 
These induce a binary operation $\+_{i\ast}$ and a unary operation $\bmi_{i\ast}$ on $\SM(X,P_i)$ that correspond to the the group operations in 
$[X,P_i]$ on the level of representatives.
This yields, in the terminology of Section~\ref{s:abelops},
a semi-effective representation for $[X,P_i]$. 
\item It remains to convert this semi-effective representation into a fully effective one; this is carried out in detail in Section~\ref{sec:main-algo}.
For this step, we use that $[X,L_i]$ and $[X,K_{i+1}]$ are straightforward to compute as fully effective Abelian groups since, by basic properties of Eilenberg--MacLane spaces, they are canonically isomorphic to certain cohomology groups of $X$. Moreover, we assume that, inductively, we have already computed $[SX,P_{i-1}]$ and $[X,P_{i-1}]$ as fully effective Abelian groups, where $SX$ is the suspension of $X$ mentioned above. 

These four Abelian groups, together with $[X,P_i]$, fit into an \emph{exact sequence} of Abelian groups (see Equation 
(\ref{theexact}) in Section~\ref{s:ind-step}), 
and this is then used to compute the desired fully effective representation 
of $[X,P_i]$---see Section~\ref{sec:main-algo}. Roughly speaking, what happens
here is that, among the maps $X\to P_{i-1}$, we ``filter out'' those that
can be lifted to maps $X\to P_i$ (this corresponds to evaluating
an appropriate obstruction, as was mentioned at the beginning of this section),
for each map that can be lifted we determine all possible liftings,
and finally, we test which of the lifted maps are homotopic.
Since there are infinitely many homotopy classes of maps involved in 
these operations, we have to work globally, with generators and relations
in the appropriate Abelian groups of homotopy classes. 
\end{itemize}
\end{enumerate}

\heading{Remarks. }
\subh{Evaluating Postnikov classes.}
For $Y$ fixed, the subroutines for evaluating the Postnikov classes $\kk_i$, $i\le 2d-2$, could be hard-wired once and
for all. In some particular cases, they are given by known
explicit formulas. In particular, for $Y=S^d$, $\kk_d$ corresponds to
the famous \emph{Steenrod square} \cite{Steenr47,Steenrod:CohomologyOperationsObstructionsExtendingContinuousFunctions-1972}\iffull\ (more
 precisely, to the reduction from integral cohomology to mod
2 cohomology followed by the Steenrod square
$\text{Sq}^2$)\fi, and $\kk_{d+1} $ to
\emph{Adem's secondary cohomology operation}.
However, in the general
case, the only way of evaluating the $\kk_i$ we are aware of is
using  \emph{simplicial sets with effective homology} mentioned earlier.
In this context, our
result can also be regarded as an algorithmization of certain
\emph{higher cohomology operations} (see, e.g.,
\cite{MosherTangora:CohomologyOperations-1968}), although our
development of the required topological underpinning is
somewhat different and, in a way, simpler.\footnote{Let us
also mention the paper by
Gonz\'{a}les-D\'{\i}az and Real \cite{GonzalezDiazReal:ComputationCohomologyOperations-2003},
which provides algorithms for calculating certain primary and secondary
cohomology operations on a finite simplicial complex
 (including the Steenrod square $\text{Sq}^2$ and Adem's
secondary cohomology operation). But both their goal and
approach are different from ours. The algorithms in
\cite{GonzalezDiazReal:ComputationCohomologyOperations-2003}
are based on explicit combinatorial formulas for these
operations on the cochain level. The goal is to speed up
the ``obvious'' way of computing
the image of a given cohomology class under the considered
operation. In our setting, we have no general explicit formulas
available, and we can  work only with the cohomology
classes ``locally,'' since they are usually defined on
infinite simplicial sets. That is, a cohomology class
is represented by a cocycle, and that cocycle is given as
an algorithm that can compute the value of the cocycle on
any given simplex.}


\subh{Avoiding iterated suspensions. } In order to compute $[X,P_i]$, our algorithm recursively computes all suspensions $[\susp X, P_j]$, $d\le j\le i-1$. In a straightforward implementation of the algorithm, for computing $[\susp X,P_{i-1}]$ we should also recursively compute 
$[\susp\susp X,P_{i-2}]$ etc., forming essentially a complete binary tree of recursive calls. We remark that by a slightly more complicated implementation of the algorithm, this tree of recursive calls can be truncated, since we do not really need the complete information about $[\susp X,P_{i-1}]$ to compute $[X,P_i]$. Essentially, we need only a system of generators of $[\susp X,P_{i-1}]$ and not the relations; 
see Remark~\ref{rem:only-need-generators}. We stress, however, that this is merely a way to speed up the algorithm, and only by a constant factor if $d$ is fixed.

\subh{A remark on methods. } From a topological point of view, the tools and ideas that we use and combine to establish Theorem~\ref{t:main}
have been essentially known. 

On the one hand, there is an enormous topological
literature with many beautiful ideas; indeed, in our
experience, a problem with algorithmization may sometimes be
an \emph{abundance} of topological results, and the need to
sort them out. On the other hand, the classical computational
tools have been mostly designed for the ``paper-and-pencil''
model of calculation, where a calculating mathematician can,
e.g., easily switch between different representations of an
object or fill in some missing information by clever ad-hoc
reasoning. Adapting the various methods to machine
calculation sometimes needs a different approach; for instance, a
recursive formulation may be preferable to an explicit, but
cumbersome, formula (see, for example,
\cite{RubioSergeraert:ConstructiveAlgebraicTopology-2002,
SergerCouples} for an explanation of algorithmic difficulties
with \emph{spectral sequences}, a basic and powerful
computational tool in topology).

We see our main contribution as that of synthesis: identifying suitable methods, 
putting them all together, and organizing the result in a hopefully accessible way, so that it can be built on in the 
future. 

Some technical steps are apparently new; in this direction,
our main technical contribution is probably a suitable implementation
of the group operation on $P_i$ (Section~\ref{s:implementing-grp-op}) 
and recursive testing of nullhomotopy (Section~\ref{s:homotopytest}).
The former was generalized and, in a sense, simplified in \cite{aslep},
and the latter was extended to a more general situation 
in~\cite{VokriFil-homotopic}.

\section{Operations with Abelian groups}\label{s:abelops}

On the top level, our algorithm works with finitely generated
Abelian groups. \iffull The structure of such groups is simple
(they are direct sums of cyclic groups) and well known,
but  we will need to deal with certain subtleties
in their algorithmic representations.

\fi
In our setting, an Abelian group $A$ is represented by a set $\AA$,
whose elements are called \indef{representatives}; we also assume
that the representatives can be stored in a computer.
For $\alpha\in\AA$, let $[\alpha]$ denote the element of $A$
represented by $\alpha$. The representation is generally non-unique;
we may have $[\alpha]=[\beta]$ for $\alpha\ne\beta$.

We call $A$ represented in this way \indef{semi-effective}
if algorithms for the following three tasks are available:
\iffull
\begin{enumerate}
\item[(SE1)]
 Provide an element $\oo\in\AA$ representing the neutral element $0\in A$.
\item[(SE2)] Given $\alpha,\beta\in \AA$, compute an element
$\alpha\+\beta\in \AA$
with $[\alpha\+\beta]=[\alpha]+[\beta]$
 (where $+$ is the group operation in $A$).
\item[(SE3)]
 Given $\alpha\in\AA$, compute an element $\bmi \alpha \in\AA$
with $[\bmi \alpha]=-[\alpha]$.
\end{enumerate}
We stress that as a binary operation on $\AA$, $\+$ is not
necessarily a group operation; e.g., we may have $\alpha\+(\beta\+\gamma)\ne
(\alpha\+\beta)\+\gamma$, although of course, $[\alpha\+(\beta\+\gamma)]=
[(\alpha\+\beta)\+\gamma]$.
\fi
\iffull

For a semi-effective Abelian group, we are generally unable
to decide, for $\alpha,\beta\in\AA$, whether $[\alpha]=[\beta]$
(and, in particular, to certify that some element is nonzero).

Even if such an \emph{equality test} is available,
we still cannot infer much global
information about the structure of $A$. For example,
without additional information
we cannot certify that $A$ it is infinite cyclic---it could always
be large but finite cyclic, no matter how many operations
and tests we perform.

We now introduce a much stronger notion, with all the structural
information explicitly available. We call a semi-effective Abelian
group $A$ \indef{fully effective}
 if it is finitely generated
and we have an explicit expression
of $A$ as a direct sum of cyclic groups.
More precisely,
we assume that the following are explicitly available:
\begin{enumerate}
\item[(FE1)] A list of generators $a_1,\ldots,a_k$ of $A$
(given by representatives $\alpha_1,\ldots,\alpha_k\in \AA$)
and a list $(q_1,\ldots,q_k)$,
$q_i\in \{2,3,4,\ldots\}\cup \{\infty\}$,
such that each $a_i$ generates a cyclic subgroup of $A$ of order $q_i$,
$i=1,2,\ldots,k$, and $A$ is the direct sum of these subgroups.
\item[(FE2)] An algorithm that, given $\alpha\in\AA$,
computes a representation of $[\alpha]$ in terms of the generators;
that is, it returns $(z_1,\ldots,z_k)\in\Z^k$
such that $[\alpha]=\sum_{i=1}^kz_i a_i$.
\end{enumerate}
\fi

\iffull
First we observe that, for full effectivity, it is enough to have
$A$ given by arbitrary generators and relations. That is, we consider
a semi-effective $A$  together with a list
$b_1,\ldots,b_n$ of generators of $A$ (again explicitly
given by representatives)
and an $m\times n$ integer matrix $U$ specifying
a complete set of relations for the $b_i$; i.e.,
$\sum_{i=1}^n z_i b_i=0$ holds iff $(z_1,\ldots,z_n)$
is an integer linear combination of the rows of~$U$. Moreover,
we have an algorithm as in (FE2) that allows us to express
a given element $a$ as a linear combination of
$b_1,\ldots,b_n$ (here the expression may not be unique).

\begin{lemma}\label{l:rels}
 A semi-effective $A$
with a list of generators and
relations as above can be converted to a fully effective Abelian group.
\end{lemma}
\begin{proof} This amounts to a computation of a Smith normal form,
a standard step in computing integral homology groups, for example
(see \cite{Storjohann:NearOptimalAlgorithmsSmithNormalForm-1996}
for an efficient algorithm and references).

Concretely, the Smith normal form algorithm applied on $U$ yields
an expression $D=SUT$ with $D$ diagonal and $S,T$ square and invertible
(everything over $\Z$). Letting $\bb=(b_1,\ldots,b_n)$ be
the (column) vector of the given generators,
we define another vector $\aa=(a_1,\ldots,a_n)$
of generators by $\aa:=T^{-1}\bb$.
Then $D\aa=0$ gives a complete set of relations for the $a_i$
(since $DT^{-1}=SU$ and the row spaces of $SU$ and of $U$ are the same).
Omitting the generators $a_i$ such that $|d_{ii}|=1$ yields
a list of generators as in (FE1).
\end{proof}

In the remainder of this section, the special form of the
generators as in (FE1) will bring no advantage---on the
contrary, it would make the notation more cumbersome.
We thus assume that, for the considered fully effective
Abelian groups, we have a list of generators
and an arbitrary integer matrix specifying a complete set of relations
among the generators.
\fi

\iffull \heading{Locally effective mappings. }\fi
Let $X,Y$ be sets. We call a mapping $\varphi\:X\to Y$
\emph{locally effective} if there is an algorithm that, given an arbitrary
$x\in X$, computes $\varphi(x)$.
\iffull

Next, for
\fi
semi-effective Abelian groups $A,B$,
with sets $\AA,\BB$ of representatives,
respectively, we call a mapping $f\:A\to B$
 \indef{locally effective} if there is a locally effective
mapping $\varphi\:\AA\to\BB$ such that
$[\varphi(\alpha)]=f([\alpha])$ for all $\alpha\in \AA$.
In particular, we speak of a \indef{locally effective homomorphism}
if $f$ is a group homomorphism.

\begin{lemma}[Kernel]\label{l:ker}
Let $f\:A\to B$ be a locally effective homomorphism
of fully effective Abelian groups.
Then $\ker(f)=\{a\in A:f(a)=0\}$ can be represented as fully effective.
\end{lemma}

\begin{proof} This essentially amounts to solving a homogeneous
system of linear equations over the integers.

Let $a_1,\ldots,a_m$ be a list of generators
of $A$ and $U$ a matrix specifying a complete set
of relations among them, and similarly for $B$, $b_1,\ldots,b_n$,
and~$V$.
For every $i=1,2,\ldots,m$, we express $f(a_i)=\sum_{j=1}^n z_{ij} b_j$;
then the $m\times n$ matrix $Z=(z_{ji})$ represents $f$ in the sense
that, for $a=\sum_{i=1}^m x_i a_i$, we have $f(a)=\sum_{j=1}^n y_jb_j$
with $\yy=\xx Z$, where $\xx=(x_1,\ldots,x_m)$ and $\yy=(y_1,\ldots,y_n)$
are regarded as \emph{row} vectors.

Since $V$ is the matrix of relations in $B$,
$\sum_{j=1}^n y_jb_j$ equals $0$ in $B$ iff
$\yy=\ww V$ for an integer (row) vector $\ww$.
So $\ker f=\{\sum_{i}x_i a_i: \xx\in\Z^m,
\xx Z=\ww V\mbox{ for some }\ww\in\Z^n\}$.

Given a system of homogeneous linear equations over $\Z$,
we can use the Smith normal form to find
a system of generators for the set of all solutions
(see, e.g., \cite[Chapter~5]{Schrij86}).
In our case, dealing with the system $\xx Z=\ww V$, we can
thus compute integer vectors $\xx^{(1)},\ldots,\xx^{(\ell)}$
such that the elements $a'_k:=\sum_{i=1}^m x_i^{(k)} a_i$, $k=1,2,\ldots,\ell$,
generate $\ker f$. By similar (and routine) considerations, which we omit,
we can then compute a complete set of relations for the
generators $a'_k$, and finally we apply Lemma~\ref{l:rels}.
\end{proof}

The next operation is the dual of taking a kernel, namely,
factoring a given Abelian group by the image of a locally effective
homomorphism. For technical reasons, when applying this lemma later on,
we will need the resulting factor group to be equipped with an additional algorithm
that returns a ``witness'' for an element being zero.

\begin{lemma}[Cokernel]\label{l:coker}
Let $A,B$ be fully effective Abelian groups with sets of representatives
$\AA,\BB$, respectively, and let $f\:A\to B$ be a locally effective
homomorphism. Then we can obtain a fully effective representation
of the factor group $C:=\coker(f)=B/\im(f)$, again with the set $\BB$
of representatives. Moreover, there is an algorithm that,
given a representative $\beta\in\BB$, tests whether
$\beta$ represents $0$ in $C$,
and if yes, returns a representative $\alpha\in\AA$ such
that $[f(\alpha)]=[\beta]$ in $B$.
\end{lemma}

\begin{Remark}
\label{rem:only-need-generators}
As will become apparent from the proof, the assumption that
$A$ is fully effective is not really necessary. Indeed, all that is needed is that $A$
be semi-effective and that we have an explicit list of (representatives of) generators
for $A$. In order to avoid burdening the reader with yet another piece of of terminology,
however, we refrain from defining a special name for such representations.
\end{Remark}

\begin{proof}[Proof of Lemma~\ref{l:coker}]
As a semi-effective representation for $C$, we we simply reuse the one we already
have for $B$. That is, we reuse $\mathcal{B}$ (and the same algorithms for (SE1--3)) to represent the elements of $C$ as well. To distinguish clearly between elements in $B$ and in $C$, for $\beta\in \mathcal{B}$, we use the notation $b=[\beta]$ in $B$ and $\overline{b}=\overline{[\beta]}$ for the corresponding element $b+\im(f)$ in $C$.

For a fully effective representation of $C$, we need the following, by Lemma~\ref{l:rels}:
first, a complete set of generators for $C$ (given by representatives); second, an algorithm as in (FE2)
that expresses an arbitrary element of $C$ (given as $\beta\in B$) as a linear combination
of the generators; and, third, a complete set of relations among the generators.

For the first two tasks, we again reuse the solutions provided by the representation for $B$.
Suppose $b_1,\ldots, b_n$ (represented by $\beta_1,\ldots,\beta_n$) generate $B$. Then $\overline{b}_1,\ldots, \overline{b}_n$ (with the same representatives) generate $C$. Moreover, by assumption, we
have an algorithm that, given $\beta\in \mathcal{B}$, computes integers $z_i$ such that $[\beta]=z_1 b_1+\ldots z_nb_n$ in $B$; then  $\overline{[\beta]}=z_1 \overline{b}_1+\ldots +z_n\overline{b}_n$ in $C$.

A complete set of relations among the the generators of $C$ is obtained as follows. Let the matrix $V$ specify a complete set of relations among the generators $b_j$ of $B$, let $a_1,\ldots, a_m$ be a complete list of generators for $A$, and let $Z$ be an integer matrix representing the homomorphism $f$ with respect to the generators $a_1,\ldots,a_m$ and $b_1,\ldots,b_n$ as in the proof of Lemma~\ref{l:ker}. Then
$$U:=\left(\begin{array}{c}
Z\\ V
\end{array}\right)$$
specifies a complete set of relations among the $\overline{b}_j$ in $C$. To see that this is the case,
consider an integer (row) vector $\yy=(y_1,\ldots, y_n)$ and $\overline{b}:=\sum_{j=1}^n y_j \overline{b}_j$. Then $\overline{b}=0$ in $C$ iff $b:=\sum_{j=1}^n y_j b_j \in \im(f)$, i.e., iff there exists an element $a=\sum_{i=1}^m x_i a_i \in A$ such that $b-f(a)=0$ in $B$. By definition of $Z$ and by assumption on $V$, this is the case iff there are integer vectors $\xx$ and $\xx'$ such that $\yy=\xx Z+ \xx' V$, an integer combination of rows of $U$.

It remains to prove the second part of Lemma~\ref{l:coker}, i.e., to provide an algorithm that, given $\beta\in \mathcal{B}$, tests whether $\overline{[\beta]}=0$ in $C$, or equivalently, whether $[\beta]\in \im(f)$, and if so, computes a preimage. For this, we express $[\beta]=\sum_{j=1}^n y_j \overline{b}_j$ as an integer linear combination of generators of $B$ and then solve the system $\yy=\xx Z+\xx' V$ of integer linear equations as above (where we rely again on Smith normal form computations).
\end{proof}

The last operation is conveniently described using a
\emph{short exact sequence} of Abelian groups:
\begin{equation}
\label{e:ses}
\xymatrix{0 \ar[r]& A \ar[r]^{f} & B \ar[r]^{g} & C \ar[r]&0}
\end{equation}
(in other words, we assume that $f\:A\to B$ is an injective
homomorphism, $g\:B\to C$ is a surjective homomorphism,
and $\im f=\ker g$). It is well known that the middle group $B$ is
determined, up to isomorphism, by $A,C,f$, and $g$. For computational
purposes, though, we also need to assume that the injectivity
of $f$ is ``effective'', i.e., witnessed by
a locally effective inverse mapping $r$,
and similarly for the surjectivity of~$g$. This is formalized
in the next lemma.

\begin{lemma}[Short exact sequence]\label{l:exact-s}
Let (\ref{e:ses}) be a short exact sequence of Abelian groups,
where $A$ and $C$ are fully effective, $B$ is
semi-effective, $f\:A\to B$ and $g\:B\to C$ are locally effective
homomorphisms, and suppose that, moreover, the following
locally effective maps (typically not homomorphisms) are
given:
\begin{enumerate}
\item[\rm(i)] $r\:\im f=\ker g\to A$ such that
$f(r(b))=b$ for every $b\in B$ with $g(b)=0$.\footnote{The
equality $f(r(b))=b$ is required on the level of group elements,
and not necessarily on the level of representatives; that is,
it may happen that $\varphi(\rho(\beta))\ne \beta$,
although necessarily $[\varphi(\rho(\beta))]=[\beta]$,
where $\varphi$ represents $f$ and $\rho$ represents~$r$.
}
\item[\rm (ii)] A map of representatives\footnote{For technical reasons, in the setting where we apply this lemma later, we do not get a well-defined map $s\colon C\to B$ on the level of group elements, that is, we cannot guarantee that $[\gamma_1]=[\gamma_2]$ implies $[\rez(\gamma_1)]=[\rez(\gamma_2)]$. Because of the injectivity of $f$, this problem does not occur for the map $r$.} $\rez\colon \CC\to \BB$ (where $\BB,\CC$ are the sets of representatives for $B,C$,
respectively) that behaves as
a \emph{section} for $g$, i.e.,
such that $g([\rez(\gamma)])=[\gamma]$ for all $\gamma \in \CC$.
\end{enumerate}
Then we can obtain a fully effective representation of~$B$.
\end{lemma}


\begin{proof} Let $a_1,\ldots,a_m$ be generators
of $A$ and $c_1,\ldots,c_n$ be generators of $C$, with \emph{fixed} representative $\gamma_j\in \CC$ for each $c_j$. We define $b_j:=[\rez(\gamma_j)]$ for $1\leq j \leq n$.

Given an arbitrary element $b\in B$, we set $c:=g(b)$,
express $c=\sum_{j=1}^n z_j c_j$, and let
$b^*:= b-\sum_{j=1}^n z_j bj$. Since $g(b^*)=
g(b)-\sum_{j=1}^n z_j g(b_j)=0$, we have $b^*\in \ker g$,
and so $a:= r(b^*)$ is well defined. Then we can express
$a=\sum_{i=1}^m y_i a_i$, and we finally get $b=\sum_{i=1}^m y_i f(a_i)+
\sum_{j=1}^n z_j b_j$.

Therefore, $(f(a_1),\ldots,f(a_m),b_1,\ldots,b_n)$ is a list of
generators of $B$, computable in terms of representatives,
and the above way of expressing $b$ in terms of generators
is algorithmic. Moreover, we have $b=0$ iff $g(b)=0$ and
$r(b)=0$, which yields equality test in~$B$.

It remains to determine a complete set of relations for the described
generators (and then apply Lemma~\ref{l:rels}). Let $U$ be a matrix
specifying a complete set of relations among the generators
$a_1,\ldots,a_m$ in $A$, and $V$ is an appropriate matrix
for $c_1,\ldots,c_n$.

Let $(v_{k1},\ldots,v_{kn})$ be the $k$th row of~$V$.
Since $\sum_{j=1}^n v_{kj} c_j=0$, we have
$b^*_k := \sum_{j=1}^n v_{kj} b_j\in\ker g$,
and so, as above, we can express $b^*_k=\sum_{i=1}^m y_{ik} f(a_i)$.
Thus, we have the relation $-\sum_{i=1}^m y_{ik} f(a_i)+\sum_{j=1}^n v_{kj} b_j=0$ for our generators of~$B$.

Let $Y=(y_{ik})$ be the matrix of the coefficients $y_{ik}$
constructed above.
We claim that the matrix
$$
\left(\begin{array}{cc} -Y & V \\ U & 0\end{array}\right)
$$
specifies a complete set of relations among the generators
$f(a_1),\ldots$, $f(a_m)$, $b_1,\ldots$, $b_n$ of~$B$.
Indeed, we have just seen that the rows in the upper part
of this matrix correspond to valid relations, and
the relations given by the rows in the bottom part are valid
because $U$ specifies relations among the $a_i$ in~$A$
and $f$ is a homomorphism.

Finally, let
\begin{equation}\label{e:therel}
x_1 f(a_1)+\cdots+x_m f(a_m)+z_1b_1+\cdots+z_n b_n=0
\end{equation}
be an arbitrary valid relation among the generators.
Applying $g$ and using $g\circ f=0$, we get that
$\sum_{j=1}^n z_j c_j=0$ is a relation in $C$,
and so $(z_1,\ldots,z_n)$ is a linear combination of the
rows of~$V$.

Let $(w_1,\ldots,w_m)$ be the corresponding
linear combination of the rows of $-Y$. Then we have
$\sum_{i=1}^m w_i f(a_i)+\sum_{j=1}^n z_j b_j=0$,
and subtracting this from (\ref{e:therel}),
we arrive at $\sum_{i=1}^m (x_i-w_i)f(a_i)=0$.
Since $f$ is an injective homomorphism, we have
$\sum_{i=1}^m (x_i-w_i)a_i=0$ in $A$, and so
$(x_1-w_1,\ldots,x_m-w_m)$ is a linear combination
of the rows of~$U$. This concludes the proof.
\end{proof}

\section{Topological preliminaries}
\label{sec:prelim}


In this part we summarize notions and results from the literature.
They are mostly standard in homotopy theory and can be found
in textbooks---see, e.g., Hatcher \cite{Hatcher}
for topological notions and May~\cite{PetMay67} for
simplicial notions
(we also refer to Steenrod \cite{Steenrod:CohomologyOperationsObstructionsExtendingContinuousFunctions-1972} as an excellent background text,
although its terminology differs somewhat from the more modern usage).
However, they are perhaps
not widely known to non-topologists, and they are somewhat scattered
in the literature. We also aim at conveying some simple intuition
behind the various notions and concepts, which is not always
easy to get from the literature.

On the other hand,
in order to follow the arguments in this paper, for some
of the notions it is sufficient to know some properties,
and the actual definition is never used directly.
Such definitions are usually omitted; instead, we illustrate
the notions with simple examples or with an informal explanation.

Even readers with a strong topological background may want to
skim this part because of the notation. Moreover, in Section~\ref{s:postni}
we discuss an algorithmic result on the construction of Postnikov
systems, which may not be well known.

\heading{CW-complexes. } Below we will state various topological
results. Usually they hold for fairly general topological
spaces, but not for all topological spaces. The appropriate
level of generality for such results is the class of \indef{CW-complexes}
(or sometimes spaces homotopy equivalent to CW-complexes).

A reader not familiar with CW-complexes may either look up the definition
(e.g., in \cite{Hatcher}), or take this just to mean
``topological spaces of a fairly general kind, including all
simplicial complexes and simplicial sets''. It is also good to
know that, similar to simplicial complexes, CW-complexes
are made of pieces (\indef{cells}) of various dimensions,
where the $0$-dimensional cells are also called \indef{vertices}.
There is only one place, in Section~\ref{s:hgroup},
where a difference between CW-complexes and simplicial
sets becomes somewhat important, and there we will stress this.

%
%

\subsection{Simplicial sets}\label{s:simplsets}

Simplicial sets are our basic device for representing topological
spaces and their maps in our algorithm. Here we introduce them
briefly, with emphasis on the ideas and intuition,
referring to Friedman \cite{Friedm08} for a very friendly thorough
introduction, to \cite{Curtis:SimplicialHomotopyTheory-1971,PetMay67}
for older compact sources, and to \cite{GoerssJardine}
for a more modern and comprehensive treatment.

A \emph{simplicial set} can be thought of as a generalization of simplicial
complexes. Similar to a simplicial complex,
a simplicial set is a space built of vertices, edges, triangles,
and higher-dimensional simplices, but simplices are allowed to be glued
to each other and to themselves in more general ways. For example,
one may have several 1-dimensional simplices connecting the same
pair of vertices, a 1-simplex forming a loop,
two edges of a  2-simplex identified
to create a cone, or the boundary of a 2-simplex all contracted
to a single vertex, forming an $S^2$.

\immfig{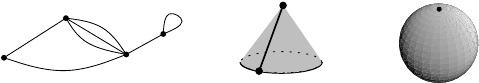}

However, unlike for the still more general CW-complexes,
a simplicial set can be described purely combinatorially.

Another new feature of a simplicial set, in comparison with a simplicial
complex, is the presence of \emph{degenerate simplices}. For example,
the edges of the triangle with a contracted boundary (in the last
example above) do not disappear---formally, each of them
keeps a phantom-like existence of a degenerate 1-simplex.

\heading{Simplices, face and degeneracy operators. }
A simplicial set $X$ is represented as a sequence $(X_0,X_1,X_2,\ldots)$
of mutually disjoint sets, where the elements of $X_m$
are called the \emph{$m$-simplices of $X$}. 
For every $m\ge 1$,
there are $m+1$ mappings $\partial_0,\ldots,\partial_m\:X_m\to X_{m-1}$
called \indef{face operators}; the meaning is that
for a simplex $\sigma\in X_m$, $\partial_i\sigma$ is the face of $\sigma$
obtained by deleting the $i$th vertex. Moreover, there are
$m+1$ mappings $s_0,\ldots,s_m\:X_m\to X_{m+1}$ (opposite direction)
called the \indef{degeneracy operators}; the meaning of $s_i\sigma$
is the degenerate simplex obtained from $\sigma$ by duplicating
the $i$th vertex. A simplex
is called \indef{degenerate} if it lies in the image of
some $s_i$; otherwise, it is \indef{nondegenerate}.
There are natural axioms that the $\partial_i$
and the $s_i$ have to satisfy, but we will not list them here,
since we won't really use them
(and the usual definition of a simplicial
set is formally different anyway, expressed
in the language of category theory). 

We call $X$ \emph{finite} if it has finitely many
nondegenerate simplices (every nonempty simplicial
set has infinitely many degenerate simplices).

\iffull 
\heading{Examples. } Here we sketch some basic examples of simplicial
sets; again, we won't provide all details, referring to \cite{Friedm08}.
Let $\Delta^n$ denote the standard
$n$-dimensional simplex regarded as a simplicial set.
For $n=0$, $(\Delta^0)_m$ consists of a single simplex,
denoted by $0^m$,  for every $m=0,1,\ldots$;  $0^0$ is the
only nondegenerate simplex. The face and degeneracy operators
are defined in the only possible way.

For $n=1$, $\Delta^1$ has two $0$-simplices (vertices), say
$0$ and $1$, and in general there are $m+2$ simplices in
$(\Delta^1)_m$; we can think of the $i$th one as containing
$i$ copies of the vertex $0$ and $m+1-i$ copies of the vertex
$1$, $i=0,1,\ldots,m+1$. For $n$ arbitrary, the $m$-simplices
of $\Delta^n$ can be thought of as all nondecreasing
$(m+1)$-term sequences with entries in $\{0,1,\ldots,n\}$;
the ones with all terms distinct are nondegenerate.

In a similar fashion,
every simplicial complex $K$ can be converted into a simplicial set $X$
in a canonical way; however, first we need to fix a linear
ordering of the vertices. The nondegenerate $m$-simplices
of $X$ are in one-to-one correspondence with the $m$-simplices of $K$,
but many degenerate simplices show up as well.

Finally we mention a ``very infinite'' but extremely
instructive example, the  \emph{singular set},
which contributed significantly to the invention
of simplicial sets---as Steenrod \cite{Steenrod:CohomologyOperationsObstructionsExtendingContinuousFunctions-1972} puts it, the definition
of a simplicial set is obtained by writing down fairly obvious properties
of the singular set. For a topological space $Y$,
the singular set $S(Y)$ is the simplicial set
whose $m$-simplices are all continuous
maps of the standard $m$-simplex into $Y$.
The $i$th face operator $\partial_i\:S(Y)_m\to S(Y)_{m-1}$
is given by the composition with
a canonical mapping that sends the standard $(m-1)$-simplex
to the $i$th face of the standard $m$-simplex. Similarly,
the $i$th degeneracy operator is induced by the canonical mapping
that collapses the standard $(m+1)$-simplex to its $i$th $m$-dimensional
face and then identifies this face with the standard $m$-simplex,
preserving the order of the vertices.

\heading{Geometric realization. } Similar to a simplicial
complex, each simplicial set $X$ defines a topological space
$|X|$ (the \indef{geometric realization of $X$}), uniquely up
to homeomorphism. Intuitively, one takes disjoint geometric
simplices corresponding to the nondegenerate simplices of
$X$, and glues them together according to the identifications
implied by the face and degeneracy operators
(we again refer to the
literature, especially to \cite{Friedm08}, for a formal
definition).

\heading{$k$-reduced simplicial sets. } A simplicial
set $X$ is called \indef{$k$-reduced} if it has a single
vertex and no nondegenerate simplices in dimensions $1$ through $k$.
Such an $X$ is necessarily $k$-connected.

A similar terminology can also be used for CW-complexes;
$k$-reduced means a single vertex (0-cell) and no cells
in dimensions $1$ through~$k$.

\heading{Products. } The \emph{product} $X\times Y$ of two
simplicial sets is formally defined in an incredibly simple way:
we have $(X\times Y)_m:=X_m\times Y_m$ for every $m$, and
the face and degeneracy operators work componentwise;
e.g., $\partial_i(\sigma,\tau):=(\partial_i\sigma,\partial_i\tau)$.
As expected, the product of simplicial sets corresponds
to the Cartesian product of the geometric realizations,
i.e., $|X\times Y|\cong|X|\times |Y|$.\footnote{To be precise,
the product of topological spaces on the right-hand side
should be taken in the category of $k$-spaces; but for the
spaces we encounter, it is the same as the usual product
of topological spaces.}
The simple definition hides some intricacies, though, as one can guess
after observing that, for example, the product of two 1-simplices
is not a simplex---so the above definition has to imply some
canonical way of triangulating the product. It indeed does,
and here the degenerate simplices deserve their bread.

\heading{Cone and suspension. } Given a simplicial set $X$,
the \indef{cone $CX$} is a simplicial set
obtained by adding a new vertex $*$ to $X$, taking
all simplices of $X$, and, for every $m$-simplex $\sigma\in X_m$
and every $i\ge 1$,
adding to $CX$ the $(m+i)$-simplex obtained from $\sigma$
by adding $i$ copies of $*$. In particular, the nondegenerate
simplices of $CX$ are the nondegenerate simplices of $X$
plus the cones over these (obtained by adding a single copy of $*$).
We skip the definition of face and degeneracy operators for $CX$ as usual.
The definitions
are discussed, e.g., in \cite[Chapter III.5]{GoerssJardine},
although there they are given in a more abstract language, and
later (in Section~\ref{s:factor} below)
we will state the concrete properties of $CX$ that we will
need.

We will also need the \indef{suspension $\susp X$}; this is the simplicial
set $CX/X$ obtained from $CX$ by contracting all simplices of $X$ into a single
vertex. The following picture illustrates both of the constructions
for a 1-dimensional $X$:
\immfig{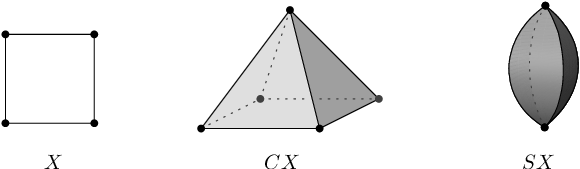}
Topologically, $SX$ is the usual (unreduced) suspension of $X$,
which is often presented as erecting a double cone
over $X$ (or a join with an $S^0$).
This would also be the ``natural'' way of defining the suspension
for a simplicial complex, but the above definition for simplicial sets
is combinatorially different, although topologically equivalent.
Even if $X$ is a simplicial complex, $\susp X$ is not. For us, the main
advantage is that the simplicial structure of $\susp X$
is particularly simple; namely, for $m>0$, the $m$-simplices
of $\susp X$ are in one-to-one correspondence with the
$(m-1)$-simplices of~$X$.\footnote{Let us also remark that in homotopy-theoretic
literature, one often works with \emph{reduced} cone
and suspension, which are appropriate for the category of
pointed spaces and maps. For example, the \emph{reduced suspension}
$\Sigma X$ is obtained from $\susp X$ by collapsing the segment
that connects the apex of $CX$ to the basepoint of~$X$.
For CW-complexes, $\Sigma X$ and $SX$ are homotopy equivalent,
so the difference is insignificant for our purposes.}

\heading{Simplicial maps and homotopies. }
Simplicial sets serve as a combinatorial way of describing a topological
space; in a similar way, simplicial maps provide a combinatorial
description of continuous maps.

A \indef{simplicial map} $f\:X\to Y$ of simplicial sets $X,Y$ consists of maps
$f_m\:X_m\to Y_m$, $m=0,1,\ldots$, that commute with the face and degeneracy
operators. We denote the set of all simplicial maps $X\to Y$
by $\SM(X,Y)$.\footnote{There is a technical issue to be clarified
here, concerning \emph{pointed maps}. We recall that a \emph{pointed
space} $(X,x_0)$ is a topological space $X$ with a choice
of a distinguished point $x_0\in X$ (the \emph{basepoint}).
In a CW-complex or simplicial set, we will always assume the
basepoint to be a vertex.
A \emph{pointed map} $(X,x_0)\to(Y,y_0)$ of pointed spaces is
a continuous map sending $x_0$ to $y_0$. Homotopies of pointed
maps are also meant to be pointed; i.e., they must keep the image of the basepoint fixed. The reader may recall that,
for example, the homotopy groups $\pi_k(Y)$ are really defined
as homotopy classes of pointed maps.

If $X,Y$ are simplicial sets, $X$ is arbitrary,
and $Y$ is a $1$-reduced
(thus, it has a single vertex, which is the basepoint),
as will be the case for the targets of simplicial maps
in our algorithm,
then every simplicial map is automatically pointed.
Thus, in this case, we need not worry about pointedness.

A topological counterpart of this is that, if $Y$ is a $1$-connected
CW-complex,
then every map  $X\to Y$ is (canonically) homotopic to a map sending
$x_0$ to $y_0$, and thus $[X,Y]$ is canonically isomorphic
to the set of all homotopy classes of pointed maps $X\to Y$.
}

It is useful to observe that it suffices to specify
a simplicial map $f\: X\to Y$ on the \emph{nondegenerate} simplices
of $X$; the values on the degenerate simplices are then determined
uniquely. In particular, if $X$ is finite,
then such an $f$ can be specified as a finite object.

A simplicial map $f\:X\to Y$ induces a continuous map $|f|\:|X|\to|Y|$
of the geometric realizations
in a natural way (we again omit the precise definition).
Often we will take the usual liberty of omitting $|\cdot|$
and not distinguishing between simplicial sets and maps
and their geometric realizations.

Of course,
not all continuous maps are induced by simplicial maps. But the usefulness
of simplicial sets for our algorithm (and many other applications)
stems mainly from the fact that, if the target $Y$
has the \emph{Kan extension property},
then \emph{every} continuous map $\varphi\:|X|\to|Y|$ is homotopic
to a simplicial map $f\:X\to Y$.\footnote{The reader may be familiar
with the \emph{simplicial approximation theorem}, which states that
for every continuous map $\varphi\:|K|\to|L|$ between the polyhedra
of simplicial complexes, there is a simplicial map of a \emph{sufficiently
fine subdivision} of $K$ into $L$ that is homotopic to $\varphi$.
The crucial difference is that in the case of simplicial sets,
if $Y$ has the Kan extension property, we need not subdivide $X$ at all!}

The Kan extension property is a certain property
of a simplicial set (and the simplicial sets having it are
called \indef{Kan simplicial sets}),
which need not be spelled out here---it will
suffice to refer to standard results to check the property where needed.
In particular,  every \indef{simplicial group}
is a Kan simplicial set, where a simplicial group
$G$ is a simplicial set for which every $G_m$ is endowed with
a group structure, and the face and degeneracy operators
are group homomorphisms
(we will see examples in Section~\ref{s:eilenb} below).

Homotopies of simplicial maps into a Kan simplicial set can
also be represented simplicially. Concretely,  a
\indef{simplicial homotopy} between two simplicial maps
$f,g\:X\to Y$ is a simplicial map $F\:X\times\Delta^1\to Y$
such that $F|_{X\times \{0\}}=f$ and $F|_{X\times \{1\}}=g$;
here, as we recall, $\Delta^1$ represents the geometric
1-simplex (segment) as a simplicial set, and, with some abuse
of notation, $\{0\}$ and $\{1\}$ are the simplicial subsets
of $\Delta^1$ representing the two vertices. Again, if $Y$ is
a Kan simplicial set, then two simplicial maps $f,g$ into $Y$
are simplicially homotopic iff they are homotopic in
the usual sense as continuous maps.

\heading{Locally effective simplicial sets and simplicial maps. } Unsurprisingly, there 
is a price to pay for the convenience of representing all continuous maps and homotopies 
simplicially: a Kan simplicial set necessarily has infinitely many simplices
in every dimension (except for some trivial cases); thus we need nontrivial techniques 
for representing it in a computer. Fortunately, the Kan simplicial sets relevant in our case 
have a sufficiently regular structure and can be handled; suitable techniques were developed and presented in
\cite{Sergeraert:ComputabilityProblemAlgebraicTopology-1994,RubioSergeraert:ConstructiveAlgebraicTopology-2002,RomeroRubioSergeraert,SergRub-homtypes,SergerGenova}.
\fi 

For algorithmic purposes, a simplicial set $X$ is represented in a \emph{black box} or \emph{oracle}  manner,
by a collection of various algorithms that allow us to access certain information about $X$. 
Specifically, let $X$ be a simplicial set, and suppose that some encoding for the simplices of $X$ by strings (finite
sequences over some fixed alphabet, say $\{0,1\}$) has been fixed. 

We say that $X$ is \emph{locally effective} if we have algorithms for evaluating the face and degeneracy maps, i.e.,
i.e., given (the encoding of) a $\thedim $-simplex $\sigma$ of $X$ and $i \in\{0,1,\ldots,\thedim\}$, we can compute 
the simplex $\partial_i\sigma$, and similarly for the degeneracy operators~$s_i$.

A simplicial map $f\colon X\to Y$  is called \emph{locally effective} if there is an algorithm that, given (an encoding of) a simplex $\sigma$ of 
$X$, computes (the encoding of) the simplex $f(\sigma)$.



\subsection{Eilenberg--MacLane spaces and cohomology}\label{s:eilenb}

%

\iffull

\heading{Cohomology. }  We will need some terminology from (simplicial)
cohomology, such as cochains, cocycles, and cohomology groups.
However, these will be mostly a convenient bookkeeping device for us,
and we won't need almost any properties of cohomology.

For a simplicial complex $X$, an integer $n\ge 0$,
 and an Abelian group $\pi$, an \indef{$n$-dimensional cochain}
with values in $\pi$ is an arbitrary mapping $c^n\:X_n\to \pi$,
i.e., a labeling of the $n$-dimensional simplices of $X$
with elements of $\pi$. The set of all $n$-dimensional
cochains is (traditionally) denoted by $C^n(X;\pi)$;
with componentwise addition, it forms an Abelian group.

For a simplicial \emph{set} $X$, we define $C^n(X;\pi)$
to consist only of cochains in which all \emph{degenerate} simplices receive value $0$ (these are sometimes called \emph{normalized cochains}).

Given an $n$-cochain $c^n$, the \emph{coboundary} of $c^n$
is the $(n+1)$-cochain $d^{n+1}=\delta  c^n$
whose value on a $\tau\in X_{n+1}$
is the sum of the values of $c^n$ over the $n$-faces of $\tau$
(taking orientations into account); formally,
$$
d^{n+1}(\tau)= \sum_{i=0}^{n+1} (-1)^i c^n(\partial_i\tau).
$$

A cochain $c^n$ is a \indef{cocycle} if $\delta c^n=0$;
$Z^n(X;\pi)\subseteq C^n(X;\pi)$ is the subgroup of all
cocycles (Z for \emph{koZyklus}), i.e., the kernel of $\delta$.
The subgroup $B^n(X;\pi)\subseteq C^n(X;\pi)$ of all \emph{coboundaries}
is the image of $\delta$; that is, $c^n$ is a coboundary if
$c^n=\delta b^{n-1}$ for some $(n-1)$-cochain $b^{n-1}$.

The $n$th (simplicial) \indef{cohomology group} of $X$
is the factor group
$$H^n(X;\pi):= Z^n(X;\pi)/B^n(X;\pi)$$
(for this to make sense,
of course, one needs the basic fact $\delta\circ\delta=0$).

\heading{Eilenberg--MacLane spaces as ``simple ranges''. }
The homotopy groups $\pi_k(Y)$ are among the most important
invariants of a topological space $Y$. The group $\pi_k(Y)$
collects information about the ``$k$-dimensional structure''
of $Y$ by probing $Y$ with all possible maps \emph{from}
$S^k$. Here the sphere $S^k$ plays a role of the ``simplest
nontrivial'' $k$-dimensional space; indeed, in some respects,
for example concerning homology groups, it is as simple as
one can possibly get.

However, as was first revealed by the famous \emph{Hopf map}
$S^3\to S^2$, the spheres are not at all simple concerning
maps going \emph{into} them. In particular, the groups
$\pi_k(S^n)$ are complicated and far from understood, in
spite of a huge body of research devoted to them. So if one
wants to probe a space $X$ with maps going \emph{into} some
``simple nontrivial'' space, then spaces other than spheres
are needed---and the Eilenberg--MacLane spaces can play this
role successfully.

Given an Abelian group $\pi$ and an integer $n\ge 1$, an
\indef{Eilenberg--MacLane space} $K(\pi,n)$ is defined as
any topological space $T$ with $\pi_n(T)\cong \pi$ and
$\pi_k(T)=0$ for all $k\ne n$. It is not difficult to show
that a $K(\pi,n)$ exists (by taking a wedge of $n$-spheres
and inductively attaching balls of dimensions
$n+1,n+2,\ldots$ to kill elements of the various homotopy
groups), and it also turns out that $K(\pi,n)$ is unique up
to homotopy equivalence.\footnote{Provided that we restrict
to spaces that are homotopy equivalent to CW-complexes.}

The circle $S^1$ is (one of the incarnations of) a $K(\Z,1)$,
and $K(\Z_2,1)$ can be represented as the
infinite-dimensional real projective space, but generally
speaking, the spaces $K(\pi,n)$  do not look exactly like
very simple objects.

\heading{Maps into $K(\pi,n)$. } Yet the following elegant
fact shows that the $K(\pi,n)$ indeed constitute ``simple''
targets of maps.

\begin{lemma}\label{l:Kpi-cohomo} For every $n\ge1$
and every Abelian group $\pi$, we have
$$
[X,K(\pi,n)]\cong H^n(X;\pi),
$$
where $X$ is a simplicial complex (or a CW-complex).
\end{lemma}

This is a basic and standard result (e.g.,
\cite[Lemma~24.4]{PetMay67} in a simplicial setting), but
nevertheless we will sketch an intuitive geometric proof,
since it explains why maps into $K(\pi,n)$ can be represented
discretely, by cocycles, and this is a key step towards
representing maps in our algorithm.

\begin{proof}[Sketch of proof. ]
For simplicity, let $X$ be a finite simplicial complex (the argument works
for a CW-complex in more or less the same way), and let us
consider an arbitrary continuous map $f\:|X|\to K(\pi,n)$,
$n\ge 2$.

First, let us consider the restriction of $f$ to the $(n-1)$-skeleton
$X^{(n-1)}$ of $X$.
Since by definition, $K(\pi,n)$
is $(n-1)$-connected,  $f|_{X^{(n-1)}}$ is homotopic
to the constant map sending ${X^{(n-1)}}$ to
a single point $y_0$ (we can imagine pulling the images
of the simplices to $y_0$ one by one, starting with vertices,
continuing with 1-simplices, etc., up to $(n-1)$-simplices).
Next, the homotopy of $f|_{X^{(n-1)}}$ with this constant map
can be extended to a homotopy of $f$ with a map $\tilde f$
defined on all of $X$ (this is a standard fact
known as the \emph{homotopy extension property} of $X$, valid
for all CW-complexes, among others). Thus,
$\tilde f\sim f$ sends ${X^{(n-1)}}$ to~$y_0$.

Next, we consider an $n$-simplex $\sigma$ of $X$. All of its boundary
now goes to $y_0$,  and so the restriction of $\tilde f$ to
$\sigma$ can be regarded as a map $S^n\to K(\pi,n)$ (since
collapsing the boundary of an $n$-simplex to a point
yields an $S^n$). Thus, up to homotopy, $\tilde f|_{\sigma}$
is described by an element of $\pi_n(K(\pi,n))=\pi$.
In this way, $\tilde f$ defines a cochain $c^n=c^n_{\tilde f}
\in C^n(X;\pi)$. The following picture captures this schematically:
\immfig{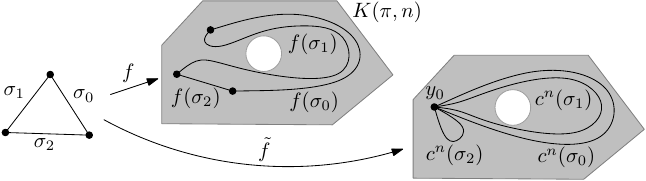}
The target space $K(\pi,n)$ is illustrated as having a hole
``responsible'' for the nontriviality of~$\pi_n$.

We note that $\tilde f$ is not determined uniquely by $f$,
and $c^n_{\tilde f}$ may also depend on the choice of~$\tilde f$.

Next, we observe that every cochain of the form $c^n_{\tilde f}$
is actually a \emph{cocycle}. To this end, we consider
an $(n+1)$-simplex $\tau\in X_{n+1}$. Since $\tilde f$ is defined
on all of $\tau$, the restriction $\tilde f|_{\partial\tau}$
to the boundary is nullhomotopic. At the same time,
$\tilde f|_{\partial\tau}$ can be regarded as the sum
of the elements of $\pi_n(K(\pi,n))$ represented by the restrictions
of $\tilde f$ to the $n$-dimensional faces of $\tau$.

Indeed, for any space $Y$ the sum $[f]$ of two elements
$[f_1],[f_2]\in \pi_n(Y)$ can be represented by contracting
an $(n-1)$-dimensional ``equator'' of $S^n$ to the basepoint,
thus obtaining a wedge of two $S^n$'s, and then defining $f$
to be $f_1$ on one of these and $f_2$ on the other, as
indicated in the picture below on the left (this time for
$n=2$). Similarly, in our case, the sum of the maps on the
facets of $\tau$ can be represented by contracting the
$(n-1)$-skeleton of $\tau$ to a point, and thus obtaining a
wedge of $n+2$ $n$-spheres. \immfig{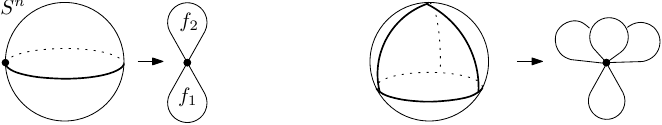}

Therefore, we have $(\delta c^n)(\tau)=0$, and $c^n=c^n_{\tilde f}\in Z^n(X;\pi)$ as claimed.

Conversely, given any $z^n\in Z^n(X;\pi)$, one can exhibit a
map $\tilde f\:X\to K(\pi,n)$ with $c^n_{\tilde f}=z^n$. Such
an $\tilde f$ is build one simplex of $X$ at a time. First,
all simplices of dimension at most $n-1$ are sent to $y_0$.
For every $\sigma\in X_n$, we choose a representative of the
element $z^n(\sigma)\in \pi_n(K(\pi,n))$, which is a
(pointed) map $S^n\to K(\pi,n)$, and use it to map $\sigma$.
Then for $\tau\in X_{n+1}$, $\tilde f$ can be extended to
$\tau$, since $\tilde f|_{\partial\tau}$ is nullhomotopic by
the cocycle condition for~$z^n$. Finally, for a simplex
$\omega$ of dimension larger than $n+1$, the $\tilde f$
constructed so far is necessarily nullhomotopic on
$\partial\omega$ because $\pi_k(K(\pi,n))=0$ for all $k>n$,
and thus an extension to $\omega$ is always possible.

We hope that this may convey some idea where the cocycle representation
of maps into $K(\pi,n)$
comes from. By similar, but a little more complicated
 considerations, which we omit here, one can
convince oneself that two maps $f,g\:X\to K(\pi,n)$ are
homotopic exactly when the corresponding cocycles
$c^n_{\tilde f}$ and $c^n_{\tilde g}$ differ by a coboundary.
In particular, for a given $f$, the cocycle $c^n_{\tilde f}$
may depend on the choice of $\tilde f$, but the cohomology
class $c^n_{\tilde f}+B^n(X;\pi)$ does not. This finishes the
proof sketch.
\end{proof}
\fi 

\iffull

\heading{A Kan simplicial model of $K(\pi,n)$. } The
Eilenberg--MacLane spaces $K(\pi,n)$ can be represented as
Kan simplicial sets, and actually as simplicial groups, in an
essentially unique way; we will keep the notation $K(\pi,n)$
for this simplicial set as well.

Namely, the set of $m$-simplices of $K(\pi,n)$ is given by the amazing formula
$$
K(\pi,n)_m := Z^n(\Delta^m;\pi).
$$
More explicitly, an $m$-simplex $\sigma$ can be regarded as a
labeling of the $n$-dimensional faces of the standard
$m$-simplex by elements of the group $\pi$; moreover, the
labels must add up to $0$ on every $(n+1)$-face. There are
${m+1\choose n+1}$ nondegenerate $n$-faces of $\Delta^m$, and
so an $m$-simplex $\sigma\in K(\pi,n)_m$ is determined by an
ordered ${m+1\choose n+1}$-tuple of elements of~$\pi$.

It is not hard to define the face and degeneracy operators
for $K(\pi,n)$, but we omit this since we won't use them
explicitly (see, e.g., \cite{PetMay67,SergerGenova}). It
suffices to say that the \emph{degenerate} $\sigma$ are
precisely those labelings with two facets of $\Delta^m$
labeled identically and zero everywhere else.

In particular, for every $m\ge 0$, we have an $m$-simplex in
$K(\pi,n)$ formed by the zero $n$-cochain, which is
nondegenerate for $m=0$ and degenerate for $m>0$, and which
we write simply as $0$ (with the dimension understood from
context). It is remarkable that the zero $n$-cochain on
$\Delta^0$ is the only vertex of the simplicial set
$K(\pi,n)$ for $n>0$.

We won't prove that this is indeed a simplicial model of $K(\pi,n)$.
Let us just note that $K(\pi,n)$ is $(n-1)$-reduced, and its $n$-simplices
correspond to elements of $\pi$ (since an $n$-cocycle on $\Delta^n$
is a labeling of the single nondegenerate $n$-simplex of $\Delta^n$
by an element of $\pi$). Thus, each $n$-simplex of $K(\pi,n)$ ``embodies''
one of the possible ways of mapping the interior of $\Delta^n$
into $K(\pi,n)$, given that the boundary goes to the basepoint.
The $(n+1)$-simplices then ``serve'' to get the appropriate
addition relations among the just mentioned maps, so that this addition
works as that in $\pi$, and the higher-dimensional simplices
kill all the higher homotopy groups.

The (elementwise) addition of cochains makes $K(\pi,n)$ into a simplicial
group, and consequently, $K(\pi,n)$ is a Kan simplicial set.

\heading{The simplicial sets $E(\pi,n)$. } The $m$-simplices
in the simplicial Eilenberg--MacLane spaces as above are all
\emph{$n$-cocycles} on $\Delta^m$. If we take all
\emph{$n$-cochains}, we obtain another simplicial set called
$E(\pi,n)$. Thus, explicitly,
$$
E(\pi,n)_m := C^n(\Delta^m;\pi).
$$
As a topological space, $E(\pi,n)$ is contractible,
and thus not particularly interesting
topologically in itself, but it makes a useful companion to
$K(\pi,n)$. Obviously, $K(\pi,n)\subseteq E(\pi,n)$, but there are
also other, less obvious relationships.

Since an $m$-simplex $\sigma\in E(\pi,n)$ is formally an $n$-cochain, we
can take its coboundary $\delta\sigma$. This is an $(n+1)$-coboundary
(and thus also cocycle), which we can interpret as an $m$-simplex
of $K(\pi,n+1)$. It turns out that this induces a \emph{simplicial}
map $E(\pi,n)\to K(\pi,n+1)$, which is (with the usual abuse of notation)
also denoted by~$\delta$. This map is actually surjective, since
the relevant cohomology groups of $\Delta^m$ are all zero and thus all cocycles
are also coboundaries.


\fi  

\iffull

\heading{Simplicial maps into $K(\pi,n)$ and $E(\pi,n)$.}
We have the following ``simplicial'' counterpart of Lemma~\ref{l:Kpi-cohomo}:

\begin{lemma}\label{l:sKpi-cohomo} For every simplicial complex (or simplicial
set) $X$,
we have
$$\SM(X,K(\pi,n))\cong Z^n(X;\pi)\mbox{ and }
\SM(X,E(\pi,n))\cong C^n(X;\pi).$$
\end{lemma}

We refer to \cite[Lemma~24.3]{PetMay67}
for a proof; here we just describe how the isomorphism\footnote{Both sets carry an Abelian group structure, and the bijection between them preserves these.
For the set $Z^n(X;\pi)$ of cocycles, the group structure is given by the usual addition of cocycles. For the set $\SM(X,K(\pi,n))$ of simplicial maps, the group structure is given by the fact that $K(\pi,n)$ is a \emph{simplicial Abelian group}, so simplicial maps into it can be added componentwise (simplexwise).} works,
i.e., how one passes between cochains and simplicial maps.
This is not hard to guess from the formal definition---there
is just one way to make things match formally.

Namely, given a $c^n\in C^n(X;\pi)$, we want to
construct the corresponding simplicial map $s=s(c^n)\:X\to E(\pi,n)$.
We consider an
$m$-simplex $\sigma\in X_m$. There is exactly one way
of inserting the standard $m$-simplex $\Delta^m$
to the ``place of $\sigma$'' into $X$; more formally,
there is a unique simplicial map $i_\sigma\:\Delta^m\to X$
that sends the $m$-simplex of $\Delta^m$ to $\sigma$
(indeed, a simplicial map has to respect the ordering of
vertices, implicit in the face and degeneracy operators).
Thus, for every such $\sigma$, the cochain $c^n$ defines
a cochain $i_\sigma^*(c^n)$ on $\Delta^m$ (the labels
of the $n$-faces of $\sigma$ are pulled back to $\Delta^m$),
and that cochain is taken as the image $s(\sigma)$.

For the reverse direction, i.e., from a simplicial map $s$
to a cochain, it suffices to look at the images of
the $n$-simplices under $s$: these are $n$-simplices
of $E(\pi,n)$ which, as we have seen, can be regarded
as elements of $\pi$---thus, they define the values
of the desired $n$-cochain.

\heading{Simplicial homotopy in $\SM(X,K(\pi,n))$. } Now that
we have a description of simplicial maps $X\to K(\pi,n)$, we
will also describe homotopies (or equivalently, simplicial
homotopies) among them. It turns out that the additive
structure (cocycle addition) on $\SM(X,K(\pi,n))\cong
Z^n(X;\pi)$ reduces the question of whether two maps
represented by cocycles $c_1$ and $c_2$ are homotopic to the
question whether their difference $c_1-c_2$ is
\emph{nullhomotopic} (homotopic to a constant map).

\begin{lemma}\label{l:shomotopy0}
Let $c_1^n,c_2^n\in Z^n(X;\pi)$ be two cocycles.
Then the simplicial maps $s_1,s_2\in\SM(X,K(\pi;n))$
represented by $c_1^n,c_2^n$, respectively, are simplicially
homotopic iff $c_1$ and $c_2$ are \emph{cohomologous},
i.e., $c_1-c_2\in B^n(X;\pi)$.
\end{lemma}

We refer to \cite[Theorem~24.4]{PetMay67} for a proof.
We also remark that a simplicial version of Lemma~\ref{l:Kpi-cohomo}
is actually proved using Lemmas~\ref{l:sKpi-cohomo}
and~\ref{l:shomotopy0}.
\fi 

\iffull

\subsection{Postnikov systems}\label{s:postni}

\fi

\iffull

Now that we have a combinatorial representation of maps from
$X$ into an Eilenberg--MacLane space, and of their
homotopies, it would be nice to have similar descriptions for
other target spaces $Y$. Expressing $Y$ through its
simplicial \emph{Postnikov system} comes as close to
fulfilling this plan as seems reasonably possible.

Postnikov systems are somewhat complicated objects, and so we will
not discuss them in detail, referring to standard textbooks
(\cite{Hatcher} in general and \cite{PetMay67} for the simplicial case)
instead. First we will explain some features of a Postnikov
system in the setting of topological spaces and continuous maps;
this part, strictly speaking, is not necessary for the algorithm.
Then we introduce a simplicial version of a Postnikov system,
and summarize the properties we will actually use.
Finally, we will present the subroutine used to
compute Postnikov systems.

\heading{Postnikov systems on the level of spaces and continuous maps.}
Let $Y$ be a CW-complex.
A \indef{Postnikov system} (also
called a \emph{Postnikov tower}) for $Y$
is a sequence of spaces $P_0,P_1,P_2,\ldots$,
where $P_0$ is a single point, together with maps
$\varphi_i\:Y\to P_i$ and $p_i\:P_i\to P_{i-1}$ such that
$p_i\circ \varphi_i=\varphi_{i-1}$, i.e., the following
diagram commutes:
%
%
%
%

\begin{center}\begin{tikzpicture}[inner sep=1mm]
 \pgftransformyscale{1.2}
 \pgftransformxscale{2.0}

%

 \node (VDOTS0) at (0,2.0) {\raisebox{.4em}{$\vdots$}};

  \node (P2) at (0,1.5) {$P_{2}$};

 \node (P1) at (0,0.5) {$P_{1}$};

 \node (P0) at (0,-0.5) {$P_0$};
\node (Y) at (-1,-0.5) {$Y$};

 \path[->]
             (P2) edge node[auto]{$p_2$} (P1)
             (P1) edge node[auto]{$p_1$} (P0)
             (Y) edge node[fill=white,inner sep=1pt]{$\varphi_2$} (P2)
             (Y) edge node[fill=white,inner sep=0pt]{$\varphi_1$} (P1)
             (Y) edge node[fill=white,inner sep=0pt]{$\varphi_0$} (P0);
\end{tikzpicture}\end{center}

Informally, the $P_i$, called the \indef{stages} of the
Postnikov system, can be thought of as successive stages in a
process of building $Y$ (or rather, a space homotopy
equivalent to $Y$) ``layer by layer'' from the Eilenberg--Mac
Lane spaces $K(\pi_i(Y),i)$.

More formally, it is required that for each $i$,
the mapping $\varphi_i$ induces an isomorphism
$\pi_j(Y)\cong \pi_j(P_i)$ of homotopy groups for every $j\le i$,
while $\pi_j(P_i)=0$ for all $j>i$.
These properties suffice
to define a Postnikov system uniquely up to homotopy equivalence,
provided that $Y$ is $0$-connected and the $P_i$ are
assumed to be CW-complexes; see, e.g.,
Hatcher \cite[Section~4.3]{Hatcher}.

For the rest of this paper, we will abbreviate $\pi_i(Y)$
to $\pi_i$.

One usually works with Postnikov systems with additional favorable properties,
sometimes called \emph{standard Postnikov systems}, and for these
to exist, more assumptions on $Y$ are needed---in particular,
they do exist if $Y$ is 1-connected. In this case, the first
two stages, $P_0$ and $P_1$, are trivial, i.e., just one-point spaces.

Standard Postnikov systems on the level of topological spaces
are defined using the notion of \emph{principal fibration},
which we do not need/want to define here. Let us just sketch
informally how $P_i$ is built from $P_{i-1}$ and $K(\pi_i,i)$.
\emph{Locally}, $P_i$ ``looks like'' the product
$P_{i-1}\times K(\pi_i,i)$, in the sense that the \emph{fiber}
$p^{-1}_i(x)$ of every point $x\in P_{i-1}$ is (homotopy equivalent to)
$K(\pi_i,i)$. However, \emph{globally} $P_i$ is usually
\emph{not} the product as above; rather, it is ``twisted''
(technically, it is the total space of  the fibration
$K(\pi_i,i)\to P_i\stackrel{p_i}{\to} P_{i-1}$).
A somewhat simple-minded analogue is the way the M\"obius band
is made by putting a segment ``over'' every point of $S^1$,
looking locally like the product $S^1\times [-1,1]$ but
globally, of course, very different from that product.

The way of ``twisting'' the $K(\pi_i,i)$ over $P_{i-1}$
to form $P_i$ is specified, for reasons that would need a somewhat
lengthy explanation, by a mapping $k_{i-1}\: P_{i-1}\to K(\pi_{i},i+1)$.
As we know, each such map $k_{i-1}$ can be represented by
a cocycle in $Z^{i+1}(P_{i-1}; \pi_{i})$, and since it really
suffices to know $k_{i-1}$ only up to homotopy, it is enough to specify it
by an element of the cohomology group $H^{i+1}(P_{i-1}; \pi_{i})$.
This element is also commonly denoted by $k_{i-1}$ and called
the  $(i-1)$st \indef{Postnikov class}\footnote{In the literature,
\emph{Postnikov factor} or \emph{Postnikov invariant}
are also used with the same meaning.} of~$Y$.

The beauty of the thing is that $P_i$, which conveys, in a
sense, complete information about the homotopy of $Y$ up to
dimension $i$, can be reconstructed from the \emph{discrete}
data given by $\pi_2,k_2,\pi_3,k_3,\ldots,k_{i-1},\pi_{i}$.

For our purposes, a key fact, already mentioned
in the outline section, is the following:

\begin{proposition}
\label{prop:XYvsXPi}
If $X$ is a CW-complex of dimension at most $i$, and $Y$
is a $1$-connected CW-complex,
then there is a bijection between $[X,Y]$
and $[X,P_i]$ (which is induced by composition with the map $\varphi_i$).
\end{proposition}


\heading{Simplicial Postnikov systems. } To use Postnikov systems algorithmically,
we represent the objects by simplicial sets and maps (this was actually
the setting in which Postnikov originally defined them). Concretely,
we will use the so-called \emph{pullback representation} (as opposed to some
other sources, where a \emph{twisted product} representation can
be found---but these representations can be converted into one another
without much difficulty).

We let $K(\pi,n)$ and $E(\pi,n)$ stand for the particular
simplicial sets as in Section~\ref{s:eilenb}. The $i$-th
stage $P_i$ of the Postnikov system for $Y$ is represented as
a simplicial subset of the product $P_{i-1}\times
E_i\subseteq E_0\times E_{1}\times\cdots\times E_i$, where
$E_j:=E(\pi_j,j)$. An $m$-simplex of $P_i$ can thus be
written as $(\sigma^0,\ldots,\sigma^{i-1},\sigma^i)$, where
$\sigma^j\in C^j(\Delta^m,\pi_j)$ is a simplex of $E_j$. It
will also be convenient to write
$(\sigma^0,\ldots,\sigma^{i-1})\in P_{i-1}$ as $\bsigma$ and
thus write a simplex of $P_i$ in the form
$(\bsigma,\sigma^i)$.
\fi 


We will introduce the following convenient abbreviations for the
Eilenberg--MacLane spaces appearing in the Postnikov system
(the first of them is quite standard):
\begin{eqnarray*}
K_{i+1}&:=&K(\pi_{i},i+1),\\
L_i&:=& K(\pi_i,i).
\end{eqnarray*}

The simplicial version of (a representative of)
the Postnikov class $k_{i-1}$ is a simplicial map
$$
\kk_{i-1}\in\SM(P_{i-1},K_{i+1}).
$$
Since $K_{i+1}$ is an Eilenberg--MacLane space, we can, and
will, also represent $\kk_{i-1}$ as a cocycle in
$Z^{i+1}(P_{i-1},\pi_i)$.

In this version, instead of ``twisting'', $\kk_{i-1}$ is used
to ``cut out'' $P_i$ from the product $P_{i-1}\times E_i$,
as follows:
\begin{equation}
\label{eq:Pi}
P_i:=\{(\bsigma,\sigma^i)\in P_{i-1}\times E_i: \kk_{i-1}(\bsigma)=
\delta\sigma^i\},
\end{equation}
where $\delta\:E_i\to K_{i+1}$ is given by the coboundary operator,
as was described above after the definition of $E(\pi,n)$. The map
$p_i\:P_i\to P_{i-1}$ in this setting is simply the projection
forgetting the last coordinate, and so it need not be specified
explicitly.

We remark that this describes what the simplicial Postnikov system looks like,
but it does not say when it really is a Postnikov system for~$Y$.
We won't discuss the appropriate conditions here;
we will just accept a guarantee of the algorithm in Theorem~\ref{t:eff-postni}
below, that it computes a valid Postnikov system for $Y$, and in particular,
such that it fulfills Proposition~\ref{prop:XYvsXPi}.

We also state another important property of the stages $P_i$ of the simplicial
Postnikov system of a simply connected $Y$: they are Kan simplicial sets
(see, e.g. \cite{Brown}).
Thus, for any simplicial set $X$, there is a bijection between the set of simplicial maps
$X\rightarrow P_i$ modulo simplicial homotopy and the set of homotopy classes of
continuous maps between the geometric realizations. Slightly abusing notation,
we will denote both sets by $[X,P_i]$.

\heading{Computing Postnikov systems. } Let $Y$ be a $1$-connected locally effective simplicial set.
For our purposes, we shall say that $Y$ has a \indef{locally effective (truncated) Postnikov 
system with $n$ stages} if the following are available:

\begin{itemize}
\item The homotopy groups $\pi_i=\pi_i(Y)$, $2\leq i\leq n$
(provided with a fully effective representation).\footnote{For our algorithm,
it suffices to have the $\pi_i$ represented as abstract Abelian groups,
with no meaning attached to the elements.
However, if we ever wanted to translate the elements of $[X,P_i]$
to actual maps $X\to Y$, we would need the generators of each
$\pi_i$ represented as actual mappings, say simplicial, $S^i\to Y$.}
\item The stages $P_i$ and the Eilenberg--MacLane spaces $K_{i+1}$ and $L_i$, $i\leq n$, 
as locally effective simplicial sets.
\item The maps $\varphi_i\colon Y\to P_i$, $p_i\: P_i\to P_{i-1}$, and $\kk_{i-1}\: P_{i-1}\to K_{i+1}$, $i\le n$, 
as locally effective simplicial maps.\footnote{As explained above, the map $\kk_{i-1}$ is represented by an $(i+1)$-dimensional 
cocycle on $P_{i-1}$; thus, we assume that we have an algorithm that, given an $(i+1)$-simplex $\bsigma\in P_{i-1}$, returns the value
$\kk_{i-1}(\bsigma)\in \pi_{i}$. Let us also remark that, by unwrapping the definition, we get that
the input $\bsigma\in P_{i-1}$ for $\kk_{i-1}$ means a labeling
of the faces of $\Delta^{i+1}$ of all dimensions up to $i-1$,
where $j$-faces are labeled by elements of $\pi_j$. Readers
familiar with obstruction theory may see some formal similarity here:
the $(i-1)$st obstruction determines extendability of a map defined
on the $i$-skeleton to the $(i+1)$-skeleton, after possibly modifying
the map on the interiors of the $i$-simplices.}

\end{itemize}
As a preprocessing step for our main algorithm, we need the following result:

\begin{theorem}[{\cite[Theorem~1.2]{polypost}}]\label{t:eff-postni} 
There is an algorithm that, given a $1$-connected simplicial set $Y$ with finitely many nondegenerate simplices (e.g., as obtained
from a finite simplicial complex) and an integer $n$, computes a locally effective Postnikov system with $n$ stages for $Y$.
\end{theorem}

\begin{Remarks}\label{rem:postni-algo}
\begin{enumerate}
\item In the case with  $\pi_2$ through $\pi_n$ all finite,
each $P_i$, $i\le n$, has finitely many simplices in the relevant dimensions,
and so a locally effective Postnikov system can be represented
simply by a lookup table. Brown
 \cite{Brown}
gave an algorithm for computing a simplicial Postnikov system in this
restricted setting.
\item The algorithm for proving Theorem~\ref{t:eff-postni} combines the basic construction
of Brown with the framework of \emph{objects with effective homology} (as explained, e.g., in \cite{SergerGenova}).
We remark that the algorithm works under the weaker assumption that $Y$ is a simplicial set with 
effective homology,  possibly with infinitely many nondegenerate simplices.\footnote{We also note that, for $Y$ with only finitely 
many  nondegenerate simplices, the maps $\varphi_i\colon Y\to P_i$ can be represented by finite lookup tables, so we do not need to require specifically that they be locally effective. 
}
\item In \cite{polypost}, it is shown that for fixed $n$, the construction of the first $n$ stages of a Postnikov system for $Y$ can actually carried out in time polynomial in the size (number of nondegenerate simplices) of $Y$. The (lengthy) analysis, and even the precise formulation of this result, involve some technical subtleties and depend on the notions \emph{locally polynomial-time simplicial sets} and \emph{objects with polynomial-time homology}, which refine the framework of locally effective simplicial sets and of objects with effective homology and were developed in \cite{pKZ1,polypost}. We refer to \cite{polypost} for a detailed treatment.
\end{enumerate}
\end{Remarks}



\heading{An example: the Steenrod square $\text{Sq}^2$. } The
Postnikov classes $\kk_i$ are not at all simple to describe
explicitly, even for very simple spaces. As an illustration,
we present an example, essentially following \cite{Steenr47},
where an explicit description is
available: this is for $Y=S^d$, $d\ge 3$, and it concerns the
first $\kk_i$ of interest, namely, $\kk_d$. It corresponds to the
\emph{Steenrod square} $\text{Sq}^2$ in cohomology, which
Steenrod \cite{Steenr47} invented for the purpose of
classifying all maps from a $(d+1)$-dimensional complex $K$
into $S^d$---a special case of the problem treated in our
paper.

For concreteness, let us take $d=3$. Then $\kk_3$ receives as the input
a labeling of the $3$-faces of $\Delta^5$ by elements of $\pi_3(S^3)$,
i.e., integers (the lower-dimensional faces are labeled with $0$s
since $\pi_j(S^3)=0$ for $j\le 2$),
and it should return an element of $\pi_4(S^3)\cong\Z_2$.
Combinatorially, we can thus think of the input as a function
$c\:{\{0,1,\ldots,5\}\choose 4}\to \Z$, and the value of $\kk_3$
turns out to be
$$
\sum_{\sigma,\tau} c(\sigma)c(\tau)~({\rm mod}~2),
$$
where the sum is over three pairs of 4-tuples $\sigma,\tau$
as indicated in the following picture ($\sigma$ consists of the
circled points and $\tau$ of the points marked by squares---there
is always a two-point overlap):
\immfig{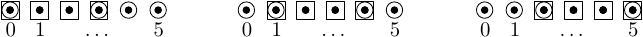}
This illustrates the nonlinearity of the Postnikov classes.

%
%
%

\section{Defining and implementing the group operation on ${[X,P_i]}$}
\label{s:implementing-grp-op}

We recall that the device that allows us to handle the
generally infinite set $[X,Y]$ of homotopy classes of maps,
under the dimension/connectedness assumption of Theorem~\ref{t:main},
is an Abelian group structure. We will actually use the group structure
on the sets $[X,P_i]$, $d\le i\le 2d-2$.
These will be computed inductively, starting
with $i=d$ (this is the first nontrivial one).

Such a group structure
with good properties
exists, and is determined uniquely,
 because $P_i$ may have nonzero homotopy groups
only in dimensions $d$ through $2d-2$; these are standard topological
considerations, which we will review in Section~\ref{s:hgroup} below.

However, we will need to work with the underlying binary operation
$\+_{i*}$ on the level of representatives, i.e., simplicial
maps in $\SM(X,P_i)$. This operation
lacks some of the pleasant properties of a group---e.g.,
it may fail to be associative.
Here considerable care and attention
to detail seem to be needed, and for an algorithmic implementation,
we also need to use the \emph{Eilenberg--Zilber reduction},
a tool related to the methods of effective homology.

\iffull

\subsection{An $H$-group structure on a space}\label{s:hgroup}

\heading{$H$-groups.}
Let $P$ be a CW-complex.
We will consider a \emph{binary
operation} on $P$ as a \emph{continuous map} $\mu\:P\times P\to P$.
For now, we will stick to writing $\mu(p,q)$ for the result of
applying $\mu$ to $p$ and $q$; later on, we will call
the operation $\+$ (with a subscript, actually)
and write it in the more usual
way as $p\+ q$.

The idea of $H$-groups is that the binary operation $\mu$
satisfies the usual group axioms but only \emph{up to homotopy}.
To formulate the existence of an inverse in this setting, we will
also need an explicit mapping $\nu\:P\to P$, continuous of course,
representing \emph{inverse up to homotopy}.

We thus say that
\begin{enumerate}
\item[(HA)] $\mu$ is \indef{homotopy associative} if
the two maps $P\times P\times P\to P$ given
by $(p,q,r)\mapsto \mu(p,\mu(q,r))$ and
by $(p,q,r)\mapsto \mu(\mu(p,q),r)$ are homotopic;
\item[(HN)] a distinguished element $\oo\in P$ (basepoint, assumed to
be a vertex in the simplicial set representation)
 is a \emph{homotopy neutral element}
if 
the maps $P\to P$ given by
$p\mapsto \mu(\oo,p)$ and $p\mapsto \mu(p,\oo)$ are
both homotopic to the identity $\id_P$;
\item[(HI)] $\nu$ is a \indef{homotopy inverse}
if the maps $p\mapsto \mu(\nu(p),p)$ and $p\mapsto \mu(p,\nu(p))$
are both  homotopic to the constant map
$p\mapsto\oo$;
\item[(HC)] $\mu$ is \indef{homotopy commutative} if $\mu$
is homotopic to $\mu'$ given by $\mu'(p,q):=\mu(q,p)$.
\end{enumerate}

An \indef{Abelian $H$-group} thus consists of $P$, $\oo$,
$\mu$, $\nu$ as above satisfying (HA), (HN), (HI), and (HC).

Of course, every Abelian topological group is also an Abelian
$H$-group. A basic example of an $H$-group that is typically
not a group is the \emph{loop space} $\Omega Y$ of a topological
space $Y$ (see, e.g. \cite[Section 4.3]{Hatcher}).
For readers familiar with the definition of the
fundamental group $\pi_1(Y)$, it suffices to say that $\Omega Y$
is like the fundamental group but \emph{without} factoring
the loops according to homotopy.

\begin{sloppypar}
We also define an \indef{$H$-homomorphism} of an $H$-group
$(P_1,\oo_1,\mu_1,\nu_1)$ into an $H$-group $(P_2,\oo_2,\mu_2,\nu_2)$
in a natural way, as a continuous map $h\:P_1\to P_2$
with $h(\oo_1)=\oo_2$ and such that the two maps
$(x,y)\mapsto h(\mu_1(x,y))$ and $(x,y)\mapsto \mu_2(h(x),h(y))$
are homotopic.
\end{sloppypar}

\heading{A group structure on homotopy classes of maps. }
For us, an $H$-group structure on $P$ is a device
for obtaining a group structure
on the set $[X,P]$ of homotopy classes of maps.
In a similar vein, an $H$-homomorphism $P_1\to P_2$
yields a group homomorphism $[X,P_1]\to [X,P_2]$.
Here is a more explicit statement:

\begin{fact}\label{f:PtoXP}
Let $(P,\oo,\mu,\nu)$ be an Abelian $H$-group,
 and let $X$ be a space. Let $\mu_*$, $\nu_*$ be the
operations defined on continuous maps $X\to P$ by pointwise
composition with $\mu$, $\nu$, respectively (i.e.,
$\mu_*(f,g)(x):=\mu(f(x),g(x))$, $\nu_*(f)(x):=\nu(f(x))$).
Then $\mu_*$, $\nu_*$ define an Abelian group structure on
the set of homotopy classes $[X,P]$ by $[f]+[g]:=[\mu_*(f,g)]$
and $-[f]:=[\nu_*(f)]$ (with the zero element given by the
homotopy class of the map sending all of $X$ to $\oo$).

If $h\:P_1\to P_2$ is an $H$-homomorphism of Abelian $H$-groups
$(P_1,\oo_1,\mu_1,\nu_1)$ and $(P_2,\oo_2,\mu_2,\nu_2)$,
then the corresponding map $h_*$, sending a continuous map
$f\:X\to P_1$ to $h_*(f)\:X\to P_2$ given by $h_*(f)(x):=
h(f(x))$, induces a homomorphism $[h_*]\:[X,P_1]\to [X,P_2]$
of Abelian groups.
\end{fact}

This fact is standard, and also entirely routine to prove.
We will actually work mostly with a simplicial counterpart
(which is proved in exactly the same way, replacing topological
notions with simplicial ones everywhere). Namely,
if $X$ is a simplicial set, $P$ is a Kan simplicial set, and
$\mu,\nu$ are simplicial maps, then by a composition as above,
we obtain maps $\mu_*\:\SM(X,P)\times \SM(X,P)\to \SM(X,P)$
and $\nu_*\:\SM(X,P)\to \SM(X,P)$, which induce an Abelian
group structure on the set $[X,P]$ of simplicial homotopy classes.
Similarly, if $h\:P_1\to P_2$ is a simplicial $H$-homomorphism
(with everything else in sight simplicial), then $h_*\:\SM(X,P_1)\to
\SM(X,P_2)$ defines a homomorphism $[h_*]\:[X,P_1]\to [X,P_2]$.

Moreover, if $\mu,\nu$ are locally effective (i.e.,
given $\sigma,\tau\in P$, we can evaluate $\mu(\sigma,\tau)$ and
$\nu(\sigma)$) and $X$ has finitely many nondegenerate simplices,
then $\mu_*,\nu_*$ are locally effective as well.
Indeed, as we have remarked, simplicial maps $X\to P$ are finitely
representable objects, and we will have them represented by vectors
of cochains.

Thus, under the above conditions, we have the Abelian group
$[X,P_i]$ semi-effectively represented, where the set of representatives
is $\SM(X,P)$.
Similarly, if $h\:P_1\to P_2$ is locally effective and $X$
is has finitely many nondegenerate simplices,
then $h_*\:\SM(X,P_1)\to \SM(X,P_2)$ is locally effective,
too.


\heading{A canonical $H$-group structure from connectivity. }
In our algorithm, the existence of a suitable $H$-group structure on $P_i$
 follows from the fact that $P_i$ has nonzero homotopy
groups only in the range from $d$ to $i$, $i\le 2d-2$. \iffull \else It is most convenient to
formulate this for CW complexes that are \emph{$(d-1)$-reduced}, i.e., that have no cells
in dimensions $1,\ldots, d-1$ (e.g., the stages $P_i$ of the
simplicial Postnikov system for a $(d-1)$-connected $Y$ have this property).
\fi

\begin{lemma} \label{lem:CW-h-group}
Let $d\geq 2$ and let $P$ be a $(d-1)$-reduced
CW complex
with distinguished vertex (basepoint) $\oo$, and with nonzero
$\pi_i(P)$ possibly occurring only for $i=d,d+1,\ldots,2d-2$.
Then there are $\mu$ and $\nu$ such that $(P,\oo,\mu,\nu)$
is an Abelian $H$-group, and moreover, $\oo$ is a
\indef{strictly neutral} element, in the sense
that $\mu(\oo,p)=\mu(p,\oo)=p$ (equalities, not only homotopy).

Moreover, if $\mu'$ is any continuous binary operation on $P$
with $\oo$ as a strictly neutral element, then $\mu'\sim\mu$
by a homotopy stationary on
the subspace $P\vee P:= (P\times\{\oo\})\cup (\{\oo\}\times P)$
(and, in particular, every such $\mu'$ automatically satisfies (HA), (HC),
and (HI) with a suitable $\nu'$).
\end{lemma}

This lemma is essentially well-known, and the necessary arguments
appear, e.g., in Whitehead \cite{Whitehead:HomotopyTheory-1978}. We nonetheless sketch
a proof, because we  are not aware of a specific reference for the lemma as
stated, and also because it sheds some light on how the
assumption of $(d-1)$-connectedness of $Y$ in Theorem~\ref{t:main}
is used.

The proof is based on the repeated application of the following basic fact
(which is a baby version
of obstruction theory and can be proved by induction of the dimension
of the cells on which the maps or homotopies have to be extended).

\begin{fact}
\label{lem:baby-obstruction}
Suppose that $X$ and $Y$ are CW complexes,
 $A\subseteq X$ is a subcomplex,
and assume that there is some integer $k$ such that all
cells in
$X\setminus A$ have dimension at least $k$ and that $\pi_i(Y)=0$
for all $i\geq k-1$. Then the following hold:
\begin{enumerate}
\item[\rm(i)] If $f\colon A\to Y$ is a continuous map,
 then there exists an extension $f'\colon X\to Y$
of $f$ (i.e., $f'|_A=f$).
\item[\rm(ii)] If $f \sim g\colon A\to Y$ are homotopic maps, and
if $f',g'\colon X\to Y$
are arbitrary extensions of $f$ and of $g$, respectively,
then $f'\sim g'$ (by a homotopy extending the given one on~$A$).
\end{enumerate}
\end{fact}

\begin{proof}[Proof of Lemma~\ref{lem:CW-h-group}.]
This proof is the \emph{only} place where it is important
that we work with CW-complexes, as opposed to simplicial
sets. This is because the \emph{product}  of CW-complexes
is defined differently from the product of simplicial sets.
In the product of CW-complexes, an $i$-cell times a $j$-cell
yields an $(i+j)$-cell (and nothing else), while
in products of simplicial sets, simplices of
problematic intermediate dimensions appear.

Let $\varphi\colon P\vee P\to P$ be the \indef{folding map}
given by $\varphi(\oo,p):=p$, $\varphi(p,\oo):=p$, $p\in P$.
Thus, the strict neutrality of $\oo$ just means that $\mu$
extends $\varphi$, and we can employ Fact~\ref{lem:baby-obstruction}.

Namely, all cells in $(P\times P)
\setminus (P\vee P)$ have dimension at least $2d$, and
$\pi_i(P)=0$ for $i\geq 2d-1$. Thus, $\varphi$ can be
extended to some $\mu\:P\times P\to P$, uniquely up to
homotopy stationary on $P\vee P$.

From the homotopy uniqueness we get the homotopy commutativity (HC)
immediately (for free). Indeed, if we define $\mu'(p,q):=
\mu(q,p)$, then the homotopy uniqueness applies and yields
$\mu'\sim\mu$. The homotopy associativity (HA) is also simple.
Let $\psi_1,\psi_2\colon P^3\to P$ be given by $\psi_1(p,q,r)
:=\mu(\mu(p,q),r)$ and $\psi_2(p,q,r):=\mu(p,\mu(q,r))$.
Then $\psi_1=\psi_2$ on the subspace
$P\vee P\vee P:=(P\times \{\oo\} \times \{\oo\})\cup (\{\oo\} \times P\times \{\oo\} )\cup (\{\oo\} \times \{\oo\} \times P )$. Since all cells in $(P\times P\times P)\setminus (P\vee P\vee P)$ are of dimension at least $2d$,
Fact~\ref{lem:baby-obstruction} gives $\psi_1\sim \psi_2$.

The existence of a homotopy inverse is not that simple,
and actually, we won't need it (since we will construct
an inverse explicitly). For a proof,
we thus refer to the literature: every 0-connected CW-complex
with an operation satisfying (HA) and (HN) also satisfies (HI);
see, e.g., \cite[Theorem~X.2.2, p.~461]{Whitehead:HomotopyTheory-1978}.
\end{proof}

\subsection{A locally effective $H$-group structure on the Postnikov stages}
\label{s:eff-boxplus}

Now we are in the setting of Theorem~\ref{t:main}; in particular,
$Y$ is a $(d-1)$-connected  simplicial set.
Let $P_i$, $i\geq 0$, denote the $i$th stage
of a locally effective simplicial Postnikov system for $Y$,
as in Section~\ref{sec:prelim};
we will consider only the first $2d-2$ stages. Since $Y$ is $(d-1)$-connected,
 $P_0$ through $P_{d-1}$ are trivial (one-point), and
each $P_i$ is $(d-1)$-reduced.
We will occasionally
refer to the $P_d,P_{d+1},\ldots,P_{2d-2}$ as the \emph{stable stages} of
the Postnikov system.

By Lemma~\ref{lem:CW-h-group}, we know that the stable stages possess
a (canonical) $H$-group structure. But we need to define the underlying
operations on $P_i$ concretely as simplicial maps and, mainly, make them effective.
Since $P_i$ is  typically an infinite object, we will have
just \emph{local effectivity}, i.e., the operations can be evaluated
algorithmically on any given pair of simplices.

From now on, we will denote the ``addition'' operation on $P_i$
by $\+_i$, and use the infix notation $\bsigma\+_i\btau$.
Similarly we write $\bmi_i\bsigma$ for the ``inverse'' of
$\bsigma$.  For a more convenient notation,
we also introduce a \emph{binary} version of $\bmi_i$
by setting $\bsigma \bmi_i\btau := \bsigma \+_i (\bmi_i\btau)$.

\heading{Preliminary considerations. }
We recall that an $m$-simplex of $P_i$ is written as
$(\sigma^0,\sigma^1,\ldots,\sigma^i)$, with $\sigma^i\in C^i(\Delta^m;\pi_i(Y))$.
Thus, its components are cochains. One potential source of confusion
is that we already \emph{have} a natural addition of such cochains
defined; they can simply be added componentwise, as effectively as
one might ever wish.

However, this  \emph{cannot} be used as the desired addition $\+_i$.
The reason is that the Postnikov classes $\kk_{i-1}$ are generally
nonlinear, and thus $\kk_{i-1}$ is typically not a homomorphism
with respect to cochain addition. In particular, we recall that
$P_i$ was defined as the subset of $P_{i-1}\times E_i$
``cut out'' by $\kk_{i-1}$, i.e., via $\kk_{i-1}(\bsigma)=\delta\sigma^i$,
where $\bsigma=(\sigma^0,\ldots,\sigma^{i-1})$. Therefore, $P_i$
is usually not even closed under the cochain addition.

Our approach to define a suitable operation $\+_i$ is inductive.
Suppose that we have already defined $\+_{i-1}$ on $P_{i-1}$.
Then we will first define $\+_i$ on special elements
of $P_i$ of the form $(\bsigma,0)$, by just adding the
$\bsigma$'s according to $\+_{i-1}$ and leaving $0$ in the last component.

Another important special case of $\+_i$ is on elements of
the form $(\bsigma,\sigma^i)\+_i (\bzero,\tau^i)$.
In this case, in spite of the general warning above against
the  cochain addition, the last components \emph{are} added
as cochains: $(\bsigma,\sigma^i)\+_i (\bzero,\tau^i)=
(\bsigma,\sigma^i+\tau^i)$. The main result of this section
constructs a locally effective $\+_i$ that extends the
two special cases just discussed.

Let us remark that by definition, $\+_i$ and $\bmi_i$, as simplicial maps,
operate on simplices of every dimension $m$. However,
in the algorithm, we will be using them only up to $m\le 2d-2$, and
so in the sequel we always implicitly assume that the considered simplices
satisfy this dimensional restriction.

\heading{The main result on $\+_i,\bmi_i$. }
The following proposition summarizes everything about
$\+_i,\bmi_i$ we will need.

\fi

\begin{proposition}
\label{l:effective-h-group} Let $Y$ be a $(d-1)$-connected
simplicial set, $d\ge 2$, and let
$P_d,P_{d+1},\ldots,P_{2d-2}$ be \iffull the stable \fi
stages of a locally effective Postnikov system with $2d-2$
stages for $Y$. Then each $P_i$ has an Abelian $H$-group
structure,\iffull\ \fi
 given by
locally effective simplicial maps $\+_i\colon P_i\times P_i\to P_i$
and $\bmi_i\colon P_i\to P_i$\iffull with
the following additional properties:

\begin{enumerate}
\item[\rm(a)] $(\bsigma,\sigma^i)\+_i(\bzero,\tau^i)=(\bsigma, \sigma^i+\tau^i)
$ for all
$(\bsigma,\sigma^i)\in P_i$ and
$\tau^i\in L_i$ (we recall that $L_i=K(\pi_i,i)$\iffull\else,
and $+$ stands for the addition of cocycles representing
the simplices of~$L_i$\fi.)
\item[\rm(b)]
$\bmi_i(\bzero,\sigma^i)=(\bzero,-\sigma^i)$
 for all $\sigma^i\in L_i$.
\item[\rm(c)] The projection $p_i\colon P_i\rightarrow P_{i-1}$
is a strict homomorphism, i.e.,
$p_i(\bsigma\+_i \btau )=p_i(\bsigma)\+_{i-1}p_i(\btau)$
and $p_i(\bmi_i \bsigma)=\bmi_{i-1} p_i(\bsigma)$ for all $\bsigma,\btau \in P_i$.
\item[\rm(d)] If, moreover, $i<2d-2$, then
the Postnikov class $\kk_i\: P_i\to K_{i+2}$
\iffull
is an $H$-homomorphism (with respect to $\+_i$ on $P_i$ and
the simplicial group operation $+$, addition of cocycles, on $K_{i+2}$).
\else
is an homomorphism up to homotopy, i.e., the simplicial maps
$(\bsigma,\btau)\mapsto \kk_{i}(\bsigma\+_i\btau)$
and $(\bsigma,\btau)\mapsto \kk_i(\bsigma)+k_i(\btau)$ are homotopic.
\fi
\end{enumerate}
\fi
\end{proposition}

\iffull

As was announced above, the proof of this proposition proceeds by induction
on $i$. The heart is an explicit and effective version of (d), which we
state and prove as a separate lemma.

\begin{lemma}\label{l:key}
Let $P_i$ be a $(d-1)$-connected simplicial set,
and let $\bzero,\+_i,\bmi_i$ be an Abelian $H$-group structure on $P_i$,
with $\+_i,\bmi_i$ locally effective. Let $\kk_i\:P_i\to K_{i+2}$
be a simplicial map, where $i<2d-2$. Then there is
a locally effective simplicial map $\Nonadd_i\: P_i\to E_{i+1}$ such
that, for all simplices
$\bsigma,\btau$ of equal dimension,
 $\Nonadd_i(\bsigma,\bzero)=\Nonadd_i(\bzero,\btau)=0$, and
$$
\kk_i(\bsigma\+_i\btau)=k_i(\bsigma)+k_i(\btau)+\delta \Nonadd_i(\bsigma,\btau).
$$
\end{lemma}

We recall that $\delta\:E_{i+1}\to K_{i+2}$ is the simplicial map
induced by the coboundary operator, and that a simplicial map
$f\: P_i\to K_{i+2}$ is nullhomotopic iff it is of the form
$\delta\circ F$ for some $F\:P_i\to E_{i+1}$
(see Lemma~\ref{l:shomotopy0}).
Therefore, the map $\Nonadd_i$ is an ``effective witness'' for the nullhomotopy
of the map $(\bsigma,\btau)\mapsto k_i(\bsigma\+_i\btau)-k_i(\bsigma)-k_i(\btau)$,
and so it shows that $k_i$ is an $H$-homomorphism.

We postpone the proof of the lemma, and prove the proposition first.

\begin{proof}[Proof of Proposition~\ref{l:effective-h-group}. ]
As was announced above, we proceed by induction
on $i$. As an inductive hypothesis, we assume that, for some
$i<2d-2$, locally effective simplicial maps $\+_i,\bmi_i$ providing an
$H$-group structure on $P_i$ have been defined satisfying (a)--(c)
in the proposition.

This inductive hypothesis is satisfied in the base case $i=d$:
in this case we have $P_d=L_d$, and $\+_d$ and $\bmi_d$ are
the addition and additive inverse of cocycles (under which $L_d$
is even a simplicial Abelian group). Then (a),(b)
obviously hold and (c) is void.

In order to carry out the inductive step from $i$ to $i+1$, we
first apply Lemma~\ref{l:key}  for $P_i$, $\+_i$, and $\kk_i$,
which yields a locally effective simplicial
map $\Nonadd_i\: P_i\times P_i\to E_{i+1}$
with $\Nonadd_i(\bsigma,\bzero)=\Nonadd_i(\bzero,\btau)=0$ and
$\kk_i(\bsigma\+_i\btau)=\kk_i(\bsigma)+\kk_i(\btau)+\delta A_i(\bsigma,\btau)$,
for all $\bsigma,\btau$. As was remarked after the lemma,
this implies that $\kk_i$ is an $H$-homomorphism with respect to~$\+_i$.

Next, using $\Nonadd_i$, we define the operations $\+_{i+1},\bmi_{i+1}$
on $P_{i+1}$. We set
\begin{equation}\label{e:defplusi+1}
(\bsigma,\sigma^{i+1})\+_{i+1} (\btau,\tau^{i+1}):=
(\bsigma\+_{i}\btau,\omega^{i+1}),
\ \ \ \mbox{where }\omega^{i+1}:=\sigma^{i+1}+\tau^{i+1}+\Nonadd_i(\bsigma,\btau).
\end{equation}
Why is $\+_{i+1}$ simplicial? Since $\+_i$ is simplicial, it suffices
to consider the last component, and this is a composition of simplicial
maps, namely, of projections, $\Nonadd_i$, and the operation $+$
in the simplicial group~$E_{i+1}$. Clearly, $\+_{i+1}$ is also locally
effective.

We also need to check that $P_{i+1}$ is closed under this
$\+_{i+1}$. We recall that, for $\bsigma\in P_i$, the condition
for $(\bsigma,\sigma^{i+1})\in P_{i+1}$ is
$\kk_i(\bsigma)=\delta \sigma^{i+1}$.
Using this condition for $(\bsigma,\sigma^{i+1}),(\btau,\tau^{i+1})\in P_{i+1}$,
together with $\bsigma\+_i\btau\in P_i$ (inductive assumption),
and the property of $\kk_i$ above, we calculate
$\kk_{i}(\bsigma\+_{i} \btau)=\kk_i(\bsigma)+\kk_i(\btau)+\delta A_i(\bsigma,\btau)=\delta\sigma^{i+1}+\delta\btau^{i+1}+\delta A_i(\bsigma,\btau)=
\delta\omega^{i+1}$,
and thus $(\bsigma,\sigma^{i+1})\+_{i+1}(\btau,\tau^{i+1})\in P_{i+1}$
as needed.

Part (a) of the proposition for $\+_{i+1}$ follows
from (\ref{e:defplusi+1}) and the property
$\Nonadd_i(\bzero,\btau)=0=\Nonadd_i(\bsigma,\bzero)$. In particular,
$(\bzero,0)$ is a strictly neutral element for $\+_{i+1}$.

Moreover, as a continuous map, $\+_{i+1}$ fulfills the assumptions
on $\mu'$ in Lemma~\ref{lem:CW-h-group}, and thus it satisfies
the axioms of an Abelian $H$-group operation.

Next, we define the inverse operation $\bmi_{i+1}$ by
\[
\bmi_{i+1}(\bsigma,\sigma^{i+1}):=
(\bmi_{i}\bsigma, -\sigma^{i+1}-\Nonadd_{i}(\bsigma,\bmi_{i}\bsigma)).
\]
It is simplicial for the same reason as that for $\+_{i+1}$, and
 by a computation similar to the one for $\+_{i+1}$ above,
we verify that $P_{i+1}$ is closed under $\bmi_{i+1}$.

To verify that this $\bmi_{i+1}$ indeed defines a homotopy
inverse to $\+_{i+1}$, we check that it actually is a \emph{strict}
inverse. Inductively, we assume $\bsigma\bmi_i\bsigma=\bzero$
for all $\bsigma\in P_i$, and from the formulas defining
$\+_{i+1}$ and $\bmi_{i+1}$, we check that
$(\bsigma,\sigma^{i+1})\bmi_{i+1}(\bsigma,\sigma^{i+1})=(\bzero,0)$.
Another simple calculation yields (b) for $\bmi_{i+1}$.

Part (c) for $\+_{i+1}$ and $\bmi_{i+1}$
follows from the definitions and from $\Nonadd_i(\bzero,\bzero)=0$.
This finishes the induction step and proves the proposition.
\end{proof}

\begin{proof}[Proof of Lemma~\ref{l:key}. ]
Here we will use (``locally'') some terminology concerning
chain complexes (e.g., chain homotopy,
homomorphism of chain complexes), for which we refer to the literature
(standard textbooks, say \cite{Hatcher}).

First we define the \indef{nonadditivity map}
$\nonadd_i\colon P_i\times P_i \rightarrow K_{i+2}$ by
$$\nonadd_i(\bsigma,\btau):=
\kk_i (\bsigma\+_i \btau) - \kk_i(\bsigma) -\kk_i(\btau).
$$
(Thus, the map $\nonadd_i$ measures the failure of $\kk_i$
to be strictly additive with respect to $\+_i$.)
We want to show that $\nonadd_i=\delta \Nonadd_i$ for a
locally effective $\Nonadd_i$.

Let us remark that the \emph{existence} of $\Nonadd_i$ can be
proved by an argument similar to the one in Lemma~\ref{lem:CW-h-group}.
That argument works for CW-complexes, and as was remarked
in the proof of that lemma, it is essential
that the product of an $i$-cell and a $j$-cell is an $(i+j)$-cell
and \emph{nothing else}. For simplicial sets the product
is defined differently, and if we consider $P_i\times P_i$
as a simplicial set, we do get simplices of ``unpleasant''
intermediate dimensions there.

We will get around this using the
\emph{Eilenberg--Zilber reduction} (which is also one of the basic tools
in effective homology---but we won't need effective homology directly);
here, we follow the exposition in~\cite{GonRea05} (see also \cite[Sections~7.8 and 8.2]{SergerGenova}).
Loosely speaking, it will allow us to convert the setting of the simplicial set
$P_i\times P_i$ to a setting of a tensor product of chain complexes,
where only terms of the ``right'' dimensions appear.

We note that $\Nonadd_i$ is defined on an infinite
object, so we cannot compute it globally---we need a local
algorithm for evaluating it, yet its answers have to be
globally consistent over the whole computation.

First we present the Eilenberg--Zilber reduction for an
arbitrary simplicial set $P$ with basepoint (and single
vertex) $\oo$. The reduction consists of three locally
effective maps\footnote{The acronyms stand for the
mathematicians \emph{Alexander} and \emph{Whitney},
\emph{Eilenberg} and \emph{Mac Lane}, and \emph{Shih},
respectively.} $\AW$, $\EML$ and $\SHI$ that fit into the
following diagram:
\begin{center}
\begin{tikzpicture}[inner sep=1mm]
\matrix (m) [row sep=4em, column sep=4em, text height=1.5ex,
text depth=0.25ex]{
 \node(Tens) {$C_*(P)\otimes C_*(P)$};&\node(Prod) {$C_*(P\times P)$};&
    \node(Bend)[label=0:{}]{$\SHI$};\\
};
 \path[->,bend left=10]
  (Tens.north east) edge node[auto] {$\EML$} (Prod.north west)
  (Prod.south west) edge node[auto] {$\AW$} (Tens.south east);
 \path[-, bend left=85] (Bend.west) edge (Prod);
 \path[->,bend right=85] (Bend.west) edge (Prod);
\end{tikzpicture}
\end{center}

Here $C_*(\cdot)$ denotes the (normalized) chain complex of a simplicial set,
with integer coefficients (so we omit the coefficient group
in the notation).
For brevity, chains of all dimensions are collected into a single structure
(whence the star subscript),
and $\otimes$ is the tensor product. Thus,
$(C_\ast(P)\otimes C_\ast(P))_n=\bigoplus_{i+j=n}C_i(P)\otimes C_j(P)$.
The operators $\AW$ and $\EML$ are homomorphisms
of chain complexes,
while $\SHI$ is a \emph{chain homotopy} operator raising
the degree by $+1$.
Thus, for each $n$, we have
$\AW_n\: C_n(P\times P)\to (C_\ast(P)\otimes C_\ast(P))_n$,
 $\EML_n\:(C_\ast(P)\otimes C_\ast(P))_n\to
C_n(P\times P)$, and
$\SHI_n\colon C_n(P\times P)\rightarrow C_{n+1}(P\times P)$.

We refer to \cite[pp.~1212--1213]{GonRea05} for explicit formulas for $\AW$ and $\EML$
in terms of the face and degeneracy operators.
We give only the formula for $\SHI$,
since $\Nonadd_i$ will be defined using $\SHI_{i+1}$,
and we summarize the properties of $\AW,\EML,\SHI$ relevant for our purposes.

The operator $\SHI_n$ operates on $n$-chains
on $P\times P$.
The formula given below specifies its values on the ``basic'' chains
of the form $(\sigma^n,\tau^n)$; here $\sigma^n,\tau^n$ are $n$-simplices
of $P$, but $(\sigma^n,\tau^n)$ is interpreted as the chain with coefficient $1$
on $(\sigma^n,\tau^n)$ and $0$ elsewhere. The definition then
extends to arbitrary chains by linearity.

Let $p$ and $q$ be non-negative integers. A \indef{$(p,q)$-shuffle}
 $(\alpha,\beta)$ is a partition
$$\{\alpha_1<\cdots<\alpha_p\}\cup\{\beta_1<\cdots<\beta_q\}$$
of the set $\{0,1,\ldots,p+q-1\}$.  Put
$$\sig(\alpha,\beta)=\sum_{i=1}^p (\alpha_i-i+1).$$
Let $\gamma=\{\gamma_i,\ldots,\gamma_r\}$ be a set of
integers. Then $s_\gamma$ denotes the compositions of the
degeneracy operators $s_{\gamma_1}\ldots s_{\gamma_r}$ (the $s_m$ are the
degeneracy operators of $P$, and $\partial_m$ are its face operators).
The operator $\SHI$ is defined by
$$\SHI(\sigma^0,\tau^0)=0,$$
$$\SHI(\sigma^m,\tau^m)=\sum_{T(m)} (-1)^{\epsilon(\alpha,\beta)}
(s_{\bar\beta+\bar m} \partial_{m-q+1}\cdots \partial_m \sigma^m,
s_{\alpha+\bar m} \partial_{\bar m}\cdots \partial_{m-q-1}\tau^m),$$
where $T(m)$ is the set of all $(p+1,q)$-shuffles such that
$0\le p+q\le m-1$,
\begin{multline*}
\bar m=m-p-q, \quad \epsilon(\alpha,\beta)=\bar m -1 + \sig(\alpha,\beta),\\
\alpha+\bar m=\{\alpha_1+\bar m,\ldots,\alpha_{p+1}+\bar m\}, \quad
\bar\beta+\bar m=\{\bar m-1,\beta_1 + \bar
m,\ldots,\beta_q+\bar m\}.
\end{multline*}
The above formula shows that $\SHI_n$ is locally effective, in the sense
that, if a chain $c_n\in C_n(P\times P)$ is given in a locally effective way,
by an algorithm that can evaluate the coefficient for each given $n$-simplex of $P\times P$,
then a similar algorithm is available for the $(n+1)$-chain $\SHI_n(c_n)$ as well.

The first fact we will need is that for every $n$, the maps satisfy the following identity (where $\partial$ denotes the boundary operator
in $C_\ast(P\times P)$):
\begin{eqnarray}
 \id_{C_n(P\times P)} - \EML_n\circ \AW_n&=&\SHI_{n-1}\circ \partial + \partial \circ \SHI_n.\label{ChainHom}
\end{eqnarray}
This identity says that $\SHI_n$ is a chain homotopy
between $\EML_n\circ \AW_n$ and the identity on $C_n(P\times P)$.

The second fact, which follows directly from the formulas in \cite{GonRea05}, is
that the operators $\EML$ and $\SHI$ behave well with respect to the basepoint $\oo$
and its degeneracies, in the following sense: For every $n$ and for every (nondegenerate) $n$-dimensional simplex $\tau^n$ of $P$ (regarded as a chain),
\begin{equation}
\label{EMLFact}
\EML_n(\oo\otimes \tau^n)=\pm (\oo^n,\tau^n), \qquad \EML_n(\tau^n\otimes \oo)=\pm(\tau^n,\oo^n),
\end{equation}
where $\oo^n$ is the (unique) $n$-dimensional degenerate simplex obtained from $\oo$.
The images in (\ref{EMLFact}) lie in the subgroup
$C_n(P\vee P) \subseteq C_n(P\times P)$.
Moreover, the operator $\SHI_n$ maps $C_n(P\vee P)$ into $C_{n+1}(P\vee P)$,
i.e., the chains $\SHI(\oo^n,\tau^n)$ and $\SHI(\tau^n,\oo^n)$ are linear
combinations of simplices of the form $(\oo^{n+1},\sigma^{n+1})$
and $(\sigma^{n+1},\oo^{n+1})$, respectively,
where $\sigma^{n+1}$ ranges over certain $(n+1)$-dimensional simplices of $P$.

We now apply this to $P=P_i$
(with basepoint $\bzero$).
We consider the nonadditivity map $\nonadd_i$ as
an $(i+2)$-cocycle on $P_i\times P_i$, which can be
regarded as a homomorphism
$\nonadd_i \colon C_{i+2}(P_i\times P_i)\to \pi_{i+1}$.
If we compose this homomorphism $\nonadd_i$ on the left
with both sides of the identity (\ref{ChainHom}),
for $n=i+2$, we get
\begin{equation}\label{e:shishi}
 \nonadd_i \circ \id_{C_{i+2}(P\times P)}-\nonadd_i\circ \EML_{i+2}\circ \AW_{i+2} =
   \nonadd_i\circ \SHI_{i+1}\circ\partial + \nonadd_i\circ \partial \circ \SHI_{i+2}.
\end{equation}
Now $\nonadd_i\circ \partial=0$ since
 $\nonadd_i$ is a cocycle. Moreover, every basis element of $C_*(P_{i}) \otimes
C_*(P_{i})$ in degree $i+2<2d$ is of the form
$\bzero \otimes \tau^{i+2}$ or $\tau^{i+2}\otimes \bzero$
(since $P_i$ has no nondegenerate simplices in dimensions $1,\ldots,d-1$).
Such elements are taken by $\EML$ into $C_{i+1}(P\vee P)$, on which
$\nonadd_i$ vanishes because $\bzero$ is a strictly neutral element for $\+_i$.
Thus, $\nonadd_i\circ \EML_{i+2}=0$  for $i+2<2d$.

Therefore, (\ref{e:shishi}) simplifies to
$\nonadd_i=\nonadd_i\circ \SHI_{i+1}\circ \partial$.
Thus, if we set $\Nonadd_i:=\nonadd_i\circ \SHI_{i+1}$,
then $\nonadd_i=\delta \Nonadd_i$, as desired (since applying $\delta$ to a cochain $\alpha$
corresponds to the composition $\alpha \circ \partial$ on the level of homomorphisms from
chains into $\pi_{i+1}$). Finally, the property $\Nonadd_i(\bzero,\cdot)=\Nonadd_i(\cdot,\bzero)=0$ follows because
the corresponding property holds for $\nonadd_i$ and $\SHI_{i+1}$ maps $C_{i+1}(P_i\vee P_i)$ to
$C_{i+2}(P_i \vee P_i)$.
\end{proof}

\fi

\iffull

\subsection{A semi-effective representation of $[X,P_i]$} \label{s:smops}

Now let $X$ be a finite simplicial complex or, more generally,
a simplicial set with finitely many nondegenerate simplices
(as we will see, the greater flexibility offered by simplicial
sets will be useful in our algorithm, even if we want to
prove Theorem~\ref{t:main} only for simplicial \emph{complexes}~$X$).

Having the locally effective $H$-group structure on the
stable Postnikov stages $P_i$, we obtain the desired
locally effective Abelian group structure on $[X,P_i]$ immediately.

\fi
\iffull
Indeed, according to the remarks following Fact~\ref{f:PtoXP},
 a simplicial map
 $s\:P\to Q$ of  arbitrary simplicial sets induces a map
$s_{*}\:\SM(X,P)\to \SM(X,Q)$ by composition,
i.e., by $s_{*}(f)(\sigma)=(s\circ f)(\sigma)$ for each simplex $\sigma \in P$.
If $P$ and $Q$ are Kan, we also get a well-defined map $[s_{*}]\: [X,P]\to
 [X,Q]$. Moreover, if  $s$ is locally effective,
then so is $s_{*}$ (since $X$ has only finitely
many nondegenerate simplices). In particular,
\fi 
 the group operations on $[X,P_i]$ are represented
by  locally effective maps
$\+_{i\ast}\:\SM(X,P_i)\times \SM(X,P_i)\to \SM(X,P_i)$
and $\bmi_{i*}\: \SM(X,P_i)\to \SM(X,P_i)$.

\iffull
\heading{The cochain representation. }
However, we can make the algorithm considerably more efficient if we use the
special structure of $P_i$ and work with cochain representatives of the simplicial maps in $\SM(X,P_i)$.

We recall from Section~\ref{sec:prelim} that simplicial maps into $K(\pi,n)$
and $E(\pi,n)$ are canonically represented by cocycles and cochains,
respectively. Simplicial maps $X\to P_i$ are, in particular,
maps into the product $E_0\times\cdots\times E_i$, and so
they can be represented by $(i+1)$-tuples of cochains
$\cc=(c^0,\ldots,c^i)$, with $c^j\in C^j:=C^j(X;\pi_j)$.

The ``simplicial'' definition of $\+_{i*},\bmi_{i*}$ can easily be translated
to a ``cochain'' definition, using the correspondence explained after Lemma~\ref{l:sKpi-cohomo}.
For simplicity, we describe the result concretely for the unary
operation $\bmi_{i*}$;
the case of $\+_{i*}$ is entirely analogous, it just
would require more notation.

Thus, to evaluate $(d^0,\ldots,d^i):=\bmi_{i*}\cc$, we need to compute the value
of $d^j$ on each $j$-simplex $\omega$ of $X$, $j=0,1,\ldots,i$.
To this end, we first identify $\omega$  with
the standard $j$-simplex $\Delta^j$ via the unique order-preserving
map of vertices. Then the restriction of $(c^0,\ldots,c^i)$
to $\omega$ (i.e., a labeling of the faces of $\omega$ by the elements
of the appropriate Abelian groups) can be regarded
as a $j$-simplex $\bsigma$ of $P_i$. We compute $\btau:=\bmi_j\bsigma$,
again a $j$-simplex of $P_i$. The component $\tau^j$ of $\btau$
is a $j$-cochain on $\Delta^j$, i.e., a single element of $\pi_j$,
and this value, finally, is the desired value of $d^j(\omega)$.
For $\+_{i*}$ everything works similarly.

We also get that $\bzero\in \SM(X,P_i)$, the simplicial map represented by
the zero cochains, is a strictly neutral element under $\+_{i*}$.

We have made $[X,P_i]$ into a \emph{semi-effectively
represented Abelian group} in the sense of Section~\ref{s:abelops}.
The representatives are the $(i+1)$-tuples $(c^0,\ldots,c^i)$
of cochains as above. However, our state of knowledge of $[X,P_i]$
is rather poor at this point; for example, we have as yet no equality test.

A substantial amount of work still lies ahead
to make $[X,P_i]$ fully effective.
\fi

\section{The main algorithm}
\label{sec:main-algo}

\iffull
In order to prove
our main result, Theorem~\ref{t:main}, on computing $[X,Y]$,
we will prove the following statement by induction on $i$.
\fi

\begin{theorem}\label{t:induct}
 Let $X$ be a simplicial set
with finitely many nondegenerate simplices,
and let $Y$ be a $(d-1)$-connected simplicial set, $d\ge 2$, for which
a locally effective Postnikov system with $2d-2$ stages
$P_0,\ldots,P_{2d-2}$ is available. Then, for every
$i=d,d+1,\ldots,2d-2$, a fully effective representation of
$[X,P_{i}]$ can be computed, with the cochain representations
of simplicial maps $X\to P_i$ as representatives.
\end{theorem}

\iffull
Two comments on this theorem are in order. First,
unlike
\fi
in Theorem~\ref{t:main}, there is no restriction on $\dim X$
(the assumption $\dim X\le 2d-2$ in Theorem~\ref{t:main} is needed
only for the isomorphism $[X,Y]\cong [X,P_{2d-2}]$).
\iffull
Second,
as was already mentioned in Section~\ref{s:smops},
\fi 
even if we want
Theorem~\ref{t:main} only for a simplicial \emph{complex} $X$,
we need Theorem~\ref{t:induct} with simplicial \emph{sets} $X$,
because of recursion.

\iffull

First we will (easily) derive Theorem~\ref{t:main}
from Theorem~\ref{t:induct}.

\begin{proof}[Proof of  Theorem~\ref{t:main}]
Given a $Y$ as in Theorem~\ref{t:main}, we first obtain
a fully effective Postnikov system for it with $2d-2$ stages
using Theorem~\ref{t:eff-postni}. Then we compute a fully effective
representation of $[X,P_{2d-2}]$ by Theorem~\ref{t:induct}.
Since $Y$ is $(d-1)$-connected
and $\dim X\le 2d-2$, there is a bijection between
$[X,Y]$ and $[X,P_{2d-2}]$\iffull\ by Proposition~\ref{prop:XYvsXPi}\fi.

It remains to implement the homotopy testing.
Given two simplicial maps $f,g\:X\to Y$, we
use the locally effective simplicial map $\varphi_{2d-2}\:Y\to P_{2d-2}$
(which is a part of a locally effective simplicial Postnikov system),
and we compute the cochain representations $\cc,\dd$ of the corresponding
simplicial  maps $\varphi_{2d-2}\circ f,
\varphi_{2d-2}\circ g\:X\to P_{2d-2}$. Then we can check, using
the fully effective representation of $[X,P_{2d-2}]$,
whether $[\cc]-[\dd]=0$ in $[X,P_{2d-2}]$.
This yields the promised homotopy testing algorithm for $[X,Y]$ and
concludes the proof of Theorem~\ref{t:main}.
\end{proof}

\subsection{The inductive step: An exact sequence for $[X,P_i]$}\label{s:ind-step}

\fi 

Theorem~\ref{t:induct} is proved by induction on $i$. The
base case is $i=d$ (since $P_0,\ldots,P_{d-1}$ are trivial
for a $(d-1)$-connected $Y$), which presents no problem: we
have $P_d=L_d=K(\pi_d,d)$, and so \iffull
$$
\fi
[X,P_d]\cong H^d(X;\pi_d)
\iffull .$$\fi
This group is fully effective, since it is the cohomology group of a
simplicial set with finitely many nondegenerate simplices, with coefficients in a fully effective
group.
\iffull (Alternatively, we could start the algorithm at $i=0$;
then it would obtain $[X,P_d]$ at stage $d$ as well.)
\fi

So now we consider $i>d$, and we assume that a fully effective representation
of $[X,P_{i-1}]$ is available, where the representatives
of the homotopy classes $[f]\in [X,P_{i-1}]$ are
(cochain representations of) simplicial maps $f\:X\to P_{i-1}$.
We want to obtain a similar representation for $[X,P_i]$.
\iffull

Let us first describe on an intuitive level what this task means
and how we are going to approach it.
\fi 

\iffull As we know, every \fi
map $g\in\SM(X,P_i)$ yields a map
$f=p_{i*}(g)=p_i\circ g\in\SM(X,P_{i-1})$
by projection (forgetting the last coordinate in $P_i$).
We first ask the question of \emph{which} maps $f\in\SM(X,P_{i-1})$
are obtained as such projections; this is traditionally
called the \indef{lifting problem} (and $g$ is called a \indef{lift}
of $f$). Here the answer follows easily from the
properties of the Postnikov system: \iffull
liftability of a map
$f$ depends only on its homotopy class $[f]\in [X,P_{i-1}]$,
and
\fi
the liftable maps in $[X,P_{i-1}]$ are obtained as the kernel
of the homomorphism $[k_{(i-1)*}]$ induced by the Postnikov class.
This is very similar to the one-step extension in the
setting of obstruction theory, as was mentioned in the introduction.
\iffull
This step will be discussed in Section~\ref{s:lift}.
\fi

Next, a single map $f\in \SM(X,P_{i-1})$ may in general have many
lifts $g$, and we need to describe their structure. This is
reasonably straightforward to do on the level of
\emph{simplicial maps}. Namely, if $\cc=(c^0,\ldots,c^{i-1})$
is the cochain representation of $f$ and $g_0$ is a fixed
lift of $f$, with cochain representation $(\cc,c_0^i)$,
then it turns out that all possible lifts $g$ of $f$
are of the form (again in the cochain representation)
$(\cc,c_0^i+z^i)$, $z^i\in Z^i(X,\pi_i)\cong\SM(X,L_i)$.
Thus, all of these lifts have a simple ``coset structure''.

This allows us to compute a list of generators of $[X,P_i]$.
We also need to find all
\emph{relations} of these generators, and for this, we need to
be able to test whether two maps $g_1,g_2\in\SM(X,P_i)$ are homotopic.
\iffull
This is somewhat more complicated, and we will develop
a recursive algorithm for homotopy
testing in Section~\ref{s:homotopytest}.

\fi
Using the group structure, it suffices
to test whether a given $g\in \SM(X,P_i)$ is nullhomotopic.
An obvious necessary
condition for this is nullhomotopy of the projection
 $f=p_i\circ g$, which  we test recursively.
Then, if $f\sim 0$, we $\+_{i*}$-add a suitable nullhomotopic map to $g$,
and this reduces the nullhomotopy test to the case
where $g$ has a cochain representation of the form $(\bzero,z^i)$,
$z_i\in Z^i(X,\pi_i)\cong \SM(X,L_i)$.

Now $(\bzero,z^i)$ can be nullhomotopic, as a map $X\to P_i$,
 by an ``obvious'' nullhomotopy, namely, one ``moving'' only the
last coordinate, or in other words, induced by a nullhomotopy
in $\SM(X,L_i)$.
But there may also be ``less obvious'' nullhomotopies,
and it turns out that these correspond to maps
$\susp X\to P_{i-1}$\iffull, where $\susp X$
is the suspension of $X$ defined in Section~\ref{s:simplsets}\fi.
Thus, in order to be able to test homotopy of
maps $X\to P_i$, we also need to compute $[\susp X,P_{i-1}]$
recursively\iffull, using the inductive assumption, i.e.,
Theorem~\ref{t:induct} for~$i-1$\fi.

\heading{The exact sequence. }
We will organize the computation of $[X,P_i]$ using
an \emph{exact sequence}, a basic tool in algebraic topology
and many other branches of mathematics.
\iffull First we write the sequence down,
including some as yet undefined symbols, and then we provide
some explanations. \fi 
It goes as follows:
\begin{equation}\label{theexact}
\hskip-1mm\xymatrix{
 [\susp X,P_{i-1}]\ar[r]^{[\mu_{i}]}&
   [X,L_i]\ar[r]^{[\lambda_{i*}]}&
     [X,P_i] \ar[d]^{[p_{i*}]}& \\
 &&[X,P_{i-1}]\ar[r]^{[k_{(i-1)*}]}&[X,K_{i+1}]}
\end{equation}
This is a sequence of Abelian groups and homomorphisms of these groups,
and  exactness means that the image of each of the homomorphisms
equals the kernel of the successive one.

\iffull

We have already met most of the objects in this exact sequence,
but for convenience, let us summarize them all.
\begin{itemize}
\item $[\susp X,P_{i-1}]$ is the group of homotopy classes
of maps from the suspension into the one lower stage $P_{i-1}$;
inductively, we may assume it to be fully effective.
\item $[\mu_{i}]$ is a homomorphism appearing here
for the first time, which will be discussed later.
\item $[X,L_i]\cong H^i(X;\pi_i)$ consists of the homotopy classes of
maps into the Eilenberg--MacLane space $L_i=K(\pi_i,i)$, and
it is fully effective.
\item $[\lambda_{i*}]$ is the homomorphism induced by the
mapping  $\lambda_i\:L_i\to P_i$, the ``insertion
to the last component''; i.e., $\lambda_i(\sigma^i)=(\bzero,\sigma^i)$.
 In terms of cochain representatives,
$\lambda_{i*}$ sends $z^i$ to $(\bzero,z^i)$.
\item $[X,P_i]$ is what we want to compute, $[p_{i*}]$ is the projection
  (on the level of homotopy), and $[X,P_{i-1}]$
has already been computed, as a fully effective Abelian group.
\item $[k_{(i-1)*}]$ is the homomorphism induced by the composition with
the Postnikov class $k_{i-1}\:P_{i-1}\to K_{i+1}=K(\pi_i,i+1)$.
\item $[X,K_{i+1}]\cong H^{i+1}(X,\pi_i)$ are again maps into
an Eilenberg--MacLane space.
\end{itemize}

Let us remark that the exact sequence (\ref{theexact}),
with some $[\mu_{i}]$,  can be obtained by standard topological
considerations from the so-called \emph{fibration sequence}
for the fibration $L_i\to P_i\to P_{i-1}$;
see, e.g., \cite[Chap.~14]{MosherTangora:CohomologyOperations-1968}.\footnote{Let us consider topological
 spaces $E$ and $B$ with basepoints and a pointed map $p\:E\to B$.
If $p$ has the so-called \emph{homotopy lifting property}
(which is the case for our $p_i$)
it is called a \indef{fibration} and the preimage $F$
of the base point in $B$ is called the \indef{fibre} of~$p$.
The sequence of maps $F\overset{i}{\hookrightarrow} E\xrightarrow{p} B$ can be prolonged into
the \indef{fibration sequence}
$$\dots \to \Omega F\xrightarrow{\Omega i}\Omega E\xrightarrow{\Omega p} \Omega B\xrightarrow{\mu} F
\xrightarrow{i} E\xrightarrow{p} B$$
of pointed maps, where, for a pointed space $Y$,
$\Omega Y$ is the space of loops starting at the base point.
For spaces $X$ and $Y$ with base points, let
$\M(X,Y)_*$ denote the  set of all continuous pointed
maps, and let $[X,Y]_*$ be the set of (pointed) homotopy classes of these maps.
Then the fibration sequence yields the sequence
$$
\dots\to\M(X,\Omega F)_*\to\M(X,\Omega E)_*\to\M(X,\Omega B)_*
\to\M(X,F)_*\to \M(X,E)_*\to\M(X,B)_*.
$$
As it turns out, on the level of homotopy classes we get even the long \emph{exact} sequence
$$\dots\to [X,\Omega F]_*\to [X,\Omega E]_*\to [X,\Omega B]_*\to [X,F]_*\to
[X,E]_*\to [X,B]_*.$$
There is a natural bijection between
$[\Sigma X,E]_*$ and $[X,\Omega E]_*$, where $\Sigma X$
is the \emph{reduced} suspension of $X$, and so we get the long exact sequence
$$ \dots \to[\Sigma X,F]_*\to [\Sigma X,E]_*\to [\Sigma X, B]_*\to [X,F]_*\to
[X,E]_*\to [X,B]_*.$$
For CW-complexes, the difference between $\susp X$
and $\Sigma X$ does not matter, and for the sequence $P_i\to P_{i-1}\to K_{i+1}$, which can be considered as a fibration, we arrive at (\ref{theexact}).
}
However, in order to have all the homomorphisms locally
effective and also to provide the locally effective
``inverses'' (as required in Lemma~\ref{l:exact-s}),
we will need to analyze the sequence in some detail;
then we will obtain  a complete ``pedestrian'' proof of the exactness
with only a small extra effort.
Thus, the fibration sequence serves just as a background.

\fi 

\heading{The algorithm } for computing $[X,P_i]$ goes as follows.

\begin{enumerate}
\itemsep=0mm
\item Compute $[X,P_{i-1}]$ fully effective, recursively.
\item\label{step:ker} Compute
$N_{i-1}:=\ker\,[k_{(i-1)*}]\subseteq [X,P_{i-1}]$
(so $N_{i-1}$ consists of all homotopy classes of liftable maps),
fully effective,
using Lemma~\ref{l:ker} and Theorem~\ref{t:eff-postni}.
\item Compute $[\susp X,P_{i-1}]$ fully effective, recursively.
\item\label{step:fact} Compute the factor group $M_{i}:=\coker\,[\mu_{i}]=
 [X,L_i]/\im\,[\mu_{i}]$ using Lemma~\ref{l:coker}, fully effective and
including the possibility of computing ``witnesses for $0$'' as in the lemma.
\item\label{step:sho} The exact sequence (\ref{theexact}) can now be
transformed to
\iffull
the short exact sequence
$$
0\to M_i \xrightarrow{\ell_i} [X,P_i]
 \xrightarrow{[p_{i*}]} N_{i-1}\to 0
$$
\fi
(where $\ell_i$ is induced by exactly the same mapping $\lambda_{i*}$
of representatives as $[\lambda_{i*}]$
in the original exact sequence (\ref{theexact})).
Let $\mathcal{N}_{i-1}:=\{f\in \SM(X,P_{i-1}): [\kk_{(i-1)*}(f)]=0\}$ be the
set of representatives of elements in $N_{i-1}$.
Implement a locally effective
``section'' $\rez_i\:\mathcal{N}_{i-1} \to\SM(X,P_i)$
with $[p_{i*}\circ \rez_i]=\id$ and a locally effective ``inverse'' $r_i\:\im\,[\lambda_{i*}]\to
M_i$ with $\ell_i\circ r_i=\id$, as in Lemma~\ref{l:exact-s},
and compute $[X,P_i]$ fully effective using that lemma.
\end{enumerate}

\iffull 
We will now examine steps~\ref{step:ker},\ref{step:fact},\ref{step:sho}
 in detail, and simultaneously establish
the exactness of~(\ref{theexact}).

\heading{Convention. } It will be notationally convenient
to let maps such as $p_{i*}$, $k_{(i-1)*}$, $\lambda_{i*}$,
which send simplicial maps to simplicial maps,
operate directly on the cochain representations (and in such
case, the result is also assumed to be a cochain representation).
Thus, for example, we can write $p_{i*}(\cc,c)=\cc$,
$\lambda_{i*}(z^i)=(\bzero,z^i)$, etc. We also write
$[\cc]$ for the homotopy class of the map represented by~$\cc$.

\subsection{Computing the liftable maps}\label{s:lift}

Here we will deal with the last part of the exact sequence
(\ref{theexact}), namely,
$$
[X,P_i] \xrightarrow{[p_{i*}]} [X,P_{i-1}]\xrightarrow{[k_{(i-1)*}]} [X,K_{i+1}].
$$

First we note that, since the projection map $p_i$ is an $H$-homomorphism by
Proposition~\ref{l:effective-h-group}(c),
the (locally effective) map $p_{i*}\:\SM(X,P_i)\to\SM(X,P_{i-1})$
indeed induces a well-defined group homomorphism $[X,P_i]\to[X,P_{i-1}]$
(Fact~\ref{f:PtoXP}). Similarly, the $H$-homo\-morphism $k_{i-1}$
(Proposition~\ref{l:effective-h-group}(d)) induces a
group homomorphism $[k_{(i-1)*}]\:[X,P_{i-1}]\to [X,K_{i+1}]\cong H^{i+1}(X;\pi_i)$.

\begin{lemma}[Lifting lemma]\label{l:lifting}\
 We have $\im\,[p_{i*}]=\ker\,[k_{(i-1)*}]$. Moreover, if we set $\mathcal{N}_{i-1}:=\{f\in \SM(X,P_{i-1}): [\kk_{(i-1)*}(f)]=0\}$, then
 there is a locally effective mapping $\rez_i\colon
\mathcal{N}_{i-1}\to \SM(X,P_i)$
such that $p_{i\ast}\circ \rez_i$ is the identity map (on the level of simplicial maps).
\end{lemma}

\begin{proof}
Let us consider a map $f\in \SM(X,P_{i-1})$
with cochain representation $\cc$.
Every cochain $(\cc,c^i)$ with $c^i\in C^i(X;\pi_i)$
represents a simplicial map $X\to P_{i-1}\times E_i$,
and this map goes into $P_{i}$ iff the condition
\begin{equation}\label{eq:SMXPi}
\kk_{(i-1)*}(\cc)=\delta c^i
\end{equation}
holds. Thus, $f$ has a lift iff $k_{(i-1)*}(\cc)$ is a coboundary,
or in other words, iff $[k_{(i-1)*}(\cc)]=0$ in $[X,K_{i+1}]$.
Hence $\im\,[p_{i*}]=\ker\,[k_{(i-1)*}]$ indeed.

Moreover, if $k_{(i-1)*}(\cc)$ is a coboundary, we can
compute some $c^i$ satisfying (\ref{eq:SMXPi})
and set $\rez_i(f):=(\cc,c^i)$. This involves some arbitrary
choice, but if we fix some (arbitrary) rule for choosing $c^i$, we obtain a locally effective $\xi_i$ as needed.
The lemma is proved.
\end{proof}

\begin{Remark}
In the previous proof as well as in a few more situations below, we will need to make some 
arbitrary choice of a particular solution to a system of linear equations over the integers.
We refrain from specifying any particular such rule, but typically, such a rule will be built into 
any particular Smith normal form algorithm that we use as a subroutine to solve the system of integer linear equations (\ref{eq:SMXPi}).
\end{Remark}

We have thus proved exactness of the sequence (\ref{theexact})
at $[X,P_{i-1}]$. Step~\ref{step:ker} of the algorithm
can be implemented using Lemma~\ref{l:ker}. We have also prepared
the section $\xi_i$ for Step~\ref{step:sho}.

\subsection{Factoring by maps from $\susp X$}\label{s:factor}

We now focus on the initial part
$$
[\susp X,P_{i-1}]\xrightarrow{[\mu_{i}]}[X,L_i]\xrightarrow{[\lambda_{i*}]} [X,P_i]
$$
of the exact sequence (\ref{theexact}), and explain how
the suspension comes into the picture.
We remark that $[\lambda_{i*}]$ is a well-defined homomorphism,
for the same reason as $[p_{i*}]$ and $[k_{(i-1)*}]$; namely,
$\lambda_i$ is an $H$-homomorphism by
Proposition~\ref{l:effective-h-group}(a).

The kernel of $[\lambda_{i*}]$ describes all homotopy
classes of maps $X\to L_i$ that are nullhomotopic as maps $X\to P_i$.
To understand how they arise as images
of maps $\susp X\to P_{i-1}$, we first need to discuss a representation
of nullhomotopies as maps from the cone.

\heading{Maps from the cone. }
A  map $X\to Y$ between two topological spaces
is nullhomotopic iff it can be extended to a map $CX\to Y$
on the cone over $X$; this is more or less a reformulation
of the definition of nullhomotopy.
The same is true in the simplicial setting if the target is a
Kan simplicial set, such as~$P_i$.

We recall that the $n$-dimensional nondegenerate simplices
of $CX$ are of two kinds: the $n$-simplices of $X$ and
the cones over the $(n-1)$-simplices of~$X$. In the language
of cochains, this means that, for any coefficient group $\pi$,
we have
$$
C^n(CX;\pi)\cong C^{n-1}(X;\pi)\oplus C^{n}(X;\pi),
$$
and thus a  cochain $b\in C^n(CX;\pi)$ can be written as
$(e,c)$, with $e\in  C^{n-1}(X;\pi)$ and $c\in C^{n}(X;\pi)$.
We also write $c=b|_X$ for the restriction of $b$ to~$X$.
The coboundary operator $C^n(CX;\pi)\to C^{n+1}(CX,;\pi)$
then acts as follows:
$$
\delta(e,c)=(-\delta e+c,\delta c).
$$

Rephrasing Lemma~\ref{l:shomotopy0} in the language of extensions to $CX$, we get the following:

\begin{cor}A map $f\in\SM(X,L_i)$, represented
by a cocycle $c^i\in Z^i(X;\pi_i)$, is nullhomotopic
iff there is a cocycle $b\in Z^i(CX;\pi)\cong \SM(CX,L_i)$
such that $b|_X=c$.
\end{cor}

This describes the homotopies in $\SM(X,L_i)$, which
induce the ``obvious'' homotopies in $\im\lambda_{i*}$.
Let us now consider an element in the image of $\lambda_{i*}$,
i.e., a map $g\:X\to P_i$ with a cochain representation
$(\bzero,c^i)$. By the above, a nullhomotopy of $g$
can be regarded as a simplicial map $G\:CX\to P_i$
whose cochain representation $(\bb,b^i)$ satisfies
$(\bb|_X,b^i|_X)=(\bzero,c^i)$ (here $\bb|_X=(b^0|_X,\ldots,b^{i-1}|_X)$
is the componentwise restriction to $X$).  Thus, the projection
$F:=p_{i*}\circ G\in\SM(CX,P_{i-1})$ is represented by
$\bb$ with $\bb|_X=\bzero$, and hence it maps all of the ``base'' $X$
in $CX$ to $\bzero$.

Recalling that $\susp X$ is obtained from $CX$
by identifying $X$ to a single vertex, we can see that
such $F$ exactly correspond to simplicial maps $\susp X\to P_{i-1}$
(here we use that $P_{i-1}$ has a single vertex $\bzero$).
Thus, maps in $\SM(\susp X,P_{i-1})$ give rise
to nullhomotopies of maps in $\im\lambda_{i*}$.

After this introduction, we develop the definition of $\mu_i$
and prove the exactness of our sequence (\ref{theexact}) at~$[X,L_i]$.

\heading{The homomorphism $\mu_i$. }
Since the
nondegenerate $(i+1)$-simplices of $\susp X$
are in one-to-one correspondence with
the nondegenerate $i$-simplices of $X$,
we have the isomorphism of the cochain groups
$$
D_i\: C^{i+1}(\susp X;\pi_i)\to  C^i(X;\pi_i).
$$
Moreover, this
is compatible with the coboundary operator (up to sign):
$$
\delta D_i(c)=-D_i(\delta c).
$$
Alternatively, if we identify the $(i+1)$-cochains on $\susp X$
 with those $(i+1)$-cochains $b=(e,c)\in C^{i+1}(CX;\pi_i)$
 for which $b|_X=c=0$, then the isomorphism is given by $D_i(e,0)=e$.
The coboundary formula $\delta(e,c)=(-\delta e+c,\delta c)$ for $CX$ indeed
gives  $D_i(\delta(e,0))=D_i(-\delta e,0)=-\delta e=-\delta
D_i(e,0)$.

Because of the compatibility with $\delta$,
$D_i$ restricts to an isomorphism
$Z^{i+1}(\susp X;\pi_i)\to Z^i(X;\pi_i)$
(which we also denote by $D_i$). This induces
an isomorphism
$[D_i]\:  H^{i+1}(\susp X;\pi_i) \to  H^i(X;\pi_i)$.

Translating from cochains to simplicial maps, we
can also regard $D_i$ as an isomorphism
$\SM(\susp X,K_{i+1})\to \SM(X,L_i),$
(where, as we recall, $K_{i+1}=K(\pi_i,i+1)$ and
$L_i=K(\pi_i,i)$),
 and $[D_i]$ as an isomorphism $[\susp X,K_{i+1}]
\to [X,L_i]$.

Now we define $\mu_i\:\SM(\susp X,P_{i-1})\to \SM(X,L_i)$
by
$$
\mu_i:=D_i\circ k_{(i-1)*}.
$$
That is, given $F\in\SM(\susp X,P_{i-1})$,
we first compose it with $k_{i-1}$, which yields
a map in $\SM(\susp X,K_{i+1})$ represented by a cocycle
in $Z^{i+1}(\susp X;\pi_i)$. Applying $D_i$ means
re-interpreting this as a cocycle in $Z^i(X;\pi_i)$
representing a map in $\SM(X,L_i)$, which we declare
to be $\mu_i(F)$.
This, clearly, is locally effective, and $[\mu_i]$
is a well-defined homomorphism $[\susp X,P_{i-1}]\to [X,L_i]$
(since $[D_i]$ and $[k_{(i-1)*}]$ are well-defined homomorphisms).

The connection of this definition of $\mu_i$ to the previous considerations
on nullhomotopies may not be obvious at this point, but the lemma below
shows that $\mu_i$ works.

\begin{lemma}\label{l:exactatLi}
The sequence (\ref{theexact}) is exact at $[X,L_i]$, i.e.,
$\im\,[\mu_i]=\ker\,[\lambda_{i*}]$.
\end{lemma}

\begin{proof}
First we want to prove the inclusion
$\im\,[\mu_i]\subseteq \ker\,[\lambda_{i*}]$.
To this end, we consider $F\in\SM(\susp X,P_{i-1})$ arbitrary
and want to show that $[\lambda_{i*}(\mu_i(F))]=0$ in $[X,P_i]$.

As was discussed above, we can view $F$ as a map $\overline F\:CX\to P_{i-1}$
that is zero on $X$. Let $\bb$ be the cochain representation of
$\overline F$; thus, $\bb|_X=\bzero$.

Let $z^i\in Z^{i}(X;\pi_i)$ be the cocycle representing $\mu_i(F)$.
Then $(0,z^i)\in C^{i-1}(X;\pi_i)\oplus C^i(X;\pi_i)$ represents
a map $CX\to E_i$, and $(\bb,(0,z^i))$ represents a map
$G\:CX\to P_{i-1}\times E_i$.

We claim that $G$ actually goes into $P_{i}$, i.e.,
is a lift of $\overline F$. For this, we just need
to verify the lifting condition (\ref{eq:SMXPi}),
which reads $k_{(i-1)*}(\bb)=\delta(0,z^i)$.

By the coboundary formula for the cone, we have
$\delta(0,z^i)=(z^i,0)$, while $k_{(i-1)*}(\bb)=(z^i,0)$
by the definition of $\mu_i(F)$. So $G\in\SM(CX,P_i)$
is indeed a lift of $\overline F$. At the same time,
$(\bb,(0,z^i))|_X=(\bzero,z^i)$, and so $G$ is a nullhomotopy
for the map represented by $(\bzero,z^i)$, which is just
$\lambda_{i*}(\mu_i(F))$.
\medskip

To prove the reverse inclusion $\im\,[\mu_i]\supseteq
\ker\,[\lambda_{i*}]$, we proceed similarly. Suppose that
$z^i\in Z^i(X;\pi_i)$ represents a map $f\in\SM(X,L_i)$ with
 $[\lambda_{i*}(f)]=0$ in $[X,P_i]$. Then $\lambda_{i*}(f)$
has the cochain representation $(\bzero,z^i)$, and there is a nullhomotopy
$G\in\SM(CX,P_i)$ for it, with a cochain representation
$(\bb,(a^{i-1},z^i))$, where $\bb|_X=\bzero$.

Since $\bb|_X=\bzero$, $\bb$ represents a map $\overline
F\in\SM(CX,P_{i-1})$ zero on $X$, which can also be regarded
as $F\in \SM(\susp X,P_{i-1})$. Let $\tilde z^i$ represent
$\mu_i(F)$. Since $G$ is a lift of $\overline F$, the lifting
condition $k_{(i-1)*}(\bb)=\delta(a^{i-1},z^i)$ holds. We
have $k_{(i-1)*}(\bb)=(\tilde z^i,0)$, again by the
definition of $\mu_i$, and $\delta(a^{i-1},z^i)= (-\delta
a^{i-1}+z^i,\delta z^i)$ by the coboundary formula for the cone.
Hence $\tilde z^i-z^i=\delta a^{i-1}$, which means that
$[z^i]=[\tilde z^i]$.
Thus $[f]=[\mu_i(F)]\in \im\,[\mu_i]$, and the lemma is proved.
\end{proof}

Having $[\mu_i]$ defined as a locally effective homomorphism,
we can employ Lemma~\ref{l:coker} and implement Step~\ref{step:fact}
of the algorithm.

\subsection{Computing nullhomotopies}\label{s:homotopytest}

The next step is to prove the exactness of the sequence (\ref{theexact})
at $[X,P_i]$.

\begin{lemma}\label{l:atXP_i}
We have $\im\,[\lambda_{i*}]=\ker\,[p_{i*}]$.
\end{lemma}

\begin{proof}
The inclusion $\im\,[\lambda_{i*}]\subseteq \ker\,[p_{i*}]$
holds even on the level of simplicial maps, i.e.,
$ \im \lambda_{i*}\subseteq \ker p_{i*}$. Indeed,
$p_{i*}(\lambda_{i*}(z^i))=p_{i*}(\bzero,z^i)=\bzero$.

For the reverse inclusion,
consider $(\cc,c^i)\in \SM(X,P_i)$ and suppose
that $[p_{i*}(\cc,c^i)]=[\cc]=0\in [X,P_{i-1}]$.
We need to find some $z^i\in Z^i(X;\pi_i)$
with $[(\bzero,z^i)]=[(\cc,c^i)]$ in $[X,P_i]$.

A suitable $z^i$ can be constructed by taking a nullhomotopy  $CX\to P_{i-1}$
for $\cc$ and lifting it. Namely,
let $\bb$ represent a nullhomotopy for $\cc$, i.e., $\bb|_X=\cc$, and let
$(\bb,b^i)$ be a lift of $\bb$ (it exists because
$CX$ is contractible and thus \emph{every} map on it can be lifted).
We set
$$
z^i:= c^i-(b^i|_X).
$$

We need to verify that $z^i$ is a cocycle. This follows from
the lifting conditions $k_{(i-1)*}(\cc)=\delta c^i$
and $k_{(i-1)*}(\bb)=\delta b^i$, and from the fact that
$k_{(i-1)*}(\bb)|_X=k_{(i-1)*}(\bb|_X)=k_{(i-1)*}(\cc)$ (this is because
applying $k_{(i-1)*}$ really means a composition of maps,
and thus it commutes with restriction). Indeed, we have
 $\delta z^i=\delta c^i-\delta (b^i|_X)=
\kk_{(i-1)*}(\cc)-\kk_{(i-1)*}(\cc)=0$.

It remains to to check that $[(\cc,c^i)]=[(\bzero,z^i)]$.
We calculate $[(\cc,c^i)]-[(\bzero,z^i)]=
[(\cc,c^i)\bmi_{i*}(\bzero,z^i)]= [(\cc,c^i-z^i)]=[(\cc,b^i|_X)]=
[(\bb|_X,b^i|_X)]=0$, since $(\bb,b^i)$ is a nullhomotopy
for $(\bb|_X,b^i|_X)$.
\end{proof}

\heading{Defining the inverse for $\lambda_{i*}$. }
Now we consider the cokernel $M_i=[X,L_i]/\im\,[\mu_i]$
as in Step~\ref{step:fact} of the algorithm, and the (injective) homomorphism
$\ell_i\:M_i\to [X,P_i]$ induced by $[\lambda_{i*}]$.

The last thing we need for applying Lemma~\ref{l:exact-s} in
Step~\ref{step:sho} is a locally effective map
$r_i\:\im \ell_i\to M_i$ with $\ell_i\circ r_i=\id$.

Let $\RR_i$ be the set of representatives of the elements in
$\im\ell_i=\im\,[\lambda_{i*}]$;
by the above, we can write
$\RR_i=\{(\cc,c^i)\in \SM(X,P_i): [\cc]=0\}$.

For every
$(\cc,c^i)\in\RR_i$ we set
$$
\rho_i(\cc,c^i):=z^i,
$$
where $z^i$ is as in the above proof of Lemma~\ref{l:atXP_i}
(i.e., $z^i=c^i-(b^i|_X)$, where $(\bb,b^i)$ is a lifting
of some nullhomotopy $\bb$ for $\cc$). This definition involves a
choice of a particular $\bb$ and $b^i$, which we make arbitrarily (see above)
for each $(\cc,c^i)$.

\begin{lemma}\label{l:rho-i} The map $\rho_i$ induces a map
$r_i\:\im\,[\lambda_{i*}]\to [X,L_i]$ such that $\ell_i\circ r_i=\id$.
\end{lemma}

\begin{proof} In the proof of Lemma~\ref{l:atXP_i} we have verified that
$[\lambda_{i*}(\rho_i(\cc,c^i))]=[(\cc,c^i)]$, so $\lambda_{i*}\circ
\rho_i$ acts as the identity on the level of homotopy classes.
It follows that $r_i$ is well-defined, since $\ell_i$ is injective
and thus the condition $\ell_i\circ r_i=\id$ determines $r_i$
uniquely.
\end{proof}

We note that, since we assume $[X,P_{i-1}]$ fully effective,
we can algorithmically test whether $[\cc]=0$, i.e., whether
$\cc$ represents a nullhomotopic map---the problem is in computing
a concrete nullhomotopy $\bb$ for $\cc$.

We describe a recursive algorithm for doing that. For more convenient
notation, we will formulate it for computing nullhomotopies
for maps in $\SM(X,P_i)$, but we note that, when evaluating
$\rho_i$, we actually use this algorithm with $i-1$ instead of~$i$.
Some of the ideas in the algorithm are very similar to
those in the proof of the exactness at $[X,P_i]$
(Lemma~\ref{l:atXP_i} above), so we could have started with a
presentation of the algorithm instead of Lemma~\ref{l:atXP_i},
but we hope that a more gradual development may be easier to follow.

\heading{The nullhomotopy algorithm. }
So now we formulate a recursive algorithm $\NULLHOA(\cc,c^i)$, which takes as input a cochain
representation of a nullhomotopic map in $\SM(X,P_i)$
(i.e., such that $[(\cc,c^i)]=0$), and outputs
a nullhomotopy $(\bb,b^i)$ for $(\cc,c^i)$.

The required nullhomotopy $(\bb,b^i)$ will be $\+_{i*}$-added together
from several nullhomotopies; this decomposition is guided
by the left part of our exact sequence (\ref{theexact}).
Namely, we recursively find a nullhomotopy for $\cc$ and lift
it, which reduces the original problem to finding a nullhomotopy
for a map in $\im\lambda_{i*}$, of the form $(\bzero,z^i)$, as in the proof
of Lemma~\ref{l:atXP_i}.
Then, using the fact that $\ell_i$ is an isomorphism,
we find nullhomotopies witnessing that $[z^i]=0$ in $M_i$.
Here we need the assumption that the representation of $M_i$
allows for computing ``witnesses of zero'' as in Lemma~\ref{l:coker}.

For this to work, we need the fact that if $\bb_1$ is a nullhomotopy
for $\cc_1$ and $\bb_2$ is a nullhomotopy for $\cc_2$,
then $\bb_1\+_{i*}\bb_2$ is a nullhomotopy for $\cc_1\+_{i*}\cc_2$.
This is because $\+_{i*}$ operates on mappings by composition,
and thus it commutes with restrictions---we have already used the same
observation for~$k_{i*}$.

The base case of the algorithm is $i=d$. Here, as we recall,
$P_d=L_d=K(\pi_d,d)$, and a nullhomotopic $c^d$ means that
$c^d\in Z^{d}(X;\pi_d)$
is a coboundary. We thus compute $e\in Z^{d-1}(X;\pi_d)$
with $c^d=\delta e$, and the desired nullhomotopy is
$(e,\delta e)$ (indeed, $(e,\delta e)$ specifies a valid map
$CX\to L_d$ since, by the coboundary formula for the cone,
it is a cocycle).

Now we can state the algorithm formally.

\heading{Algorithm $\NULLHOA(\cc,c^i)$.}
\begin{enumerate}
\item[A.] (Base case)
If $i=d$, return $(\bb,b^d)=(\bzero,(e,\delta e))$ as above and stop.
\item[B.] (Recursion) Now $i>d$. Set $\bb_0:=\NULLHOA(\cc)$,
and let $(\bb_0,b_0^i)$ be an arbitrary lift of~$\bb_0$.
\item[C.] (Nullhomotopy coming from $\susp X$)
 Set $z^i:= c^i-(b_0^i|_X)$, and use the representation of $M_i$
to find a ``witness for $[z^i]=0$ in $M_i$''. That is,
compute  $F\in[\susp X,P_{i-1}]$ such that $[z^i]=[\tilde z^i]$ in
$[X,L_i]$, where $\tilde z^i$ is the cocycle representing $\mu_i(F)$.
Let $\aa$ be the cochain representation of the map $\overline F\in
\SM(CX,P_{i-1})$ corresponding to~$F$.
\item[D.] (Nullhomotopy in $[X,L_i]$) Compute
$e\in Z^{i-1}(X;\pi_i)$ with $\tilde z^i-z^i=\delta e$.
(Then, as in the base case $i=d$ above, $(e,\delta e)$ is
a nullhomotopy for $\tilde z^i-z^i$,
and thus $(\bzero,(e,\delta e))$ is a nullhomotopy for $(\bzero,\tilde z^i-z^i)$.)
\item[E.] Return $$
(\bb,b^i):=(\bb_0,b_0^i)\+_{i*}
\Bigl((\aa,(0,\tilde z^i))\+_{i*}(\bzero,(e,\delta e))\Bigr).
$$
\end{enumerate}

\begin{proof}[Proof of correctness.]
First we need to check that $z^i$ in Step~C indeed represents
$0$ in $M_i$. This is because, as in the proof of Lemma~\ref{l:atXP_i},
$[(\bzero,z^i)]=[\lambda_{i*}(z^i)]=0$, and since $\ell_i$
is injective, we have $[z^i]=0$ in $M_i$ as claimed.
So the algorithm succeeds in computing some $(\bb,b^i)$,
and we just need to check that it is a nullhomotopy for $(\cc,c^i)$.

All three terms in the formula in Step~E are valid
representatives of maps $CX\to P_i$ (for $(\bb_0,b_0^i)$
this follows from the inductive hypothesis, for $(\aa,(0,\tilde z^i))$
we have checked this in the first part of the proof of Lemma~\ref{l:exactatLi},
and for $(\bzero,(e,\delta e))$ we have already discussed this).
So $(\bb,b^i)$ also represents such a map, and all we need
to do is to check that $(\bb|_X,b^i|_X)=(\cc,c^i)$:
\begin{eqnarray*}
(\bb|_X,b^i|_X)&=& (\bb_0|_X,b_0^i|_X)\+_{i*}
\Bigl((\aa|_X,\tilde z^i)\+_{i*}(\bzero,\delta e)\Bigr)\\
&=& (\cc,b_0^i|_X)\+_{i*} \Bigl((\bzero,\tilde z^i)\+_{i*}(\bzero,z^i-\tilde z^i)\Bigr)\\
&=& (\cc,b_0^i|_X+\tilde z^i+z^i-\tilde z^i)  = (\cc,b_0^i|_X+(c^i-(b_0^i|_X)))=(\cc,c^i).
\end{eqnarray*}
Thus, the algorithm correctly computes the desired nullhomotopy.
\end{proof}

As we have already explained, the algorithm
makes $\rho_i$ locally effective, and so Step~\ref{step:sho}
of the main algorithm can be implemented.
This completes the proof of Theorem~\ref{t:induct}.
\fi

\subsection*{Acknowledgments}We would like to thank Martin
Tancer for useful conversations at early stages of this
research, and Peter Landweber for numerous useful comments
concerning an earlier version of this paper. We would also like to thank
two anonymous referees for helpful suggestions on how to improve the
exposition.

\bibliographystyle{plain}
\bibliography{Postnikov,simplicHomotopy}

\begin{thebibliography}{10}

\bibitem{Anick-homotopyhard}
D.~J. Anick.
\newblock {The computation of rational homotopy groups is {\#}$\wp$-hard}.
\newblock {Computers in geometry and topology, Proc. Conf., Chicago/Ill. 1986,
  Lect. Notes Pure Appl. Math. 114, 1--56}, 1989.

\bibitem{Boone:SimpleUnsolvableProblemsGroupTheory1-1954}
W.~W. Boone.
\newblock Certain simple, unsolvable problems of group theory. {I}.
\newblock {\em Nederl. Akad. Wetensch. Proc. Ser. A.}, 57:231--237 = Indag.
  Math. 16, 231--237 (1954), 1954.

\bibitem{Boone:SimpleUnsolvableProblemsGroupTheory2-1954}
W.~W. Boone.
\newblock Certain simple, unsolvable problems of group theory. {II}.
\newblock {\em Nederl. Akad. Wetensch. Proc. Ser. A.}, 57:492--497 = Indag.
  Math. 16, 492--497 (1954), 1954.

\bibitem{Boone:SimpleUnsolvableProblemsGroupTheory3-1955}
W.~W. Boone.
\newblock Certain simple, unsolvable problems of group theory. {III}.
\newblock {\em Nederl. Akad. Wetensch. Proc. Ser. A.}, 58:252--256 = Indag.
  Math. 17, 252--256 (1955), 1955.

\bibitem{Brown}
E.~H. B{rown (jun.)}.
\newblock {Finite computability of Postnikov complexes}.
\newblock {\em Ann. Math. (2)}, 65:1--20, 1957.

\bibitem{polypost}
M.~{\v{C}}adek, M.~Kr\v{c}\'al, J.~Matou\v{s}ek, L.~Vok\v{r}\'{\i}nek, and
  U.~Wagner.
\newblock Polynomial-time computation of homotopy groups and {P}ostnikov
  systems in fixed dimension.
\newblock Preprint, arXiv:1211.3093, 2012.

\bibitem{ext-hard}
M.~{\v{C}}adek, M.~Kr\v{c}\'al, J.~Matou\v{s}ek, L.~Vok\v{r}\'{\i}nek, and
  U.~Wagner.
\newblock Extendability of continuous maps is undecidable.
\newblock {\em Discr. Comput. Geom.}, 2013.
\newblock To appear. Preprint arXiv:1302.2370.

\bibitem{aslep}
M.~{\v{C}}adek, M.~Kr\v{c}\'al, and L.~Vok\v{r}\'{\i}nek.
\newblock Algorithmic solvability of the lifting-extension problem.
\newblock Preprint, arXiv:1307.6444, 2013.

\bibitem{Carlsson:TopologyData-2009}
G.~Carlsson.
\newblock Topology and data.
\newblock {\em Bull. Amer. Math. Soc. (N.S.)}, 46(2):255--308, 2009.

\bibitem{Curtis:SimplicialHomotopyTheory-1971}
E.~B. Curtis.
\newblock Simplicial homotopy theory.
\newblock {\em Advances in Math.}, 6:107--209, 1971.

\bibitem{EdelsbrunnerHarer:PersistentTopologySurvey-2008}
H.~Edelsbrunner and J.~Harer.
\newblock Persistent homology---a survey.
\newblock In {\em Surveys on discrete and computational geometry}, volume 453
  of {\em Contemp. Math.}, pages 257--282. Amer. Math. Soc., Providence, RI,
  2008.

\bibitem{EdelsbrunnerHarer:ComputationalTopology-2010}
H.~Edelsbrunner and J.~L. Harer.
\newblock {\em Computational topology}.
\newblock American Mathematical Society, Providence, RI, 2010.

\bibitem{Eilenberg:CohomologyContinuousMappings-1940}
S.~Eilenberg.
\newblock Cohomology and continuous mappings.
\newblock {\em Ann. of Math. (2)}, 41:231--251, 1940.

\bibitem{VokriFil-homotopic}
M.~Filakovsk\'y and L.~Vok\v{r}\'{\i}nek.
\newblock Are two given maps homotopic? {A}n algorithmic viewpoint, 2013.
\newblock Preprint, arXiv:1312.2337.

\bibitem{FranekKrcal:RobustSatisfiability}
P.~Franek and M.~Kr\v{c}\'al.
\newblock Robust satisfiability of systems of equations.
\newblock In {\em Proc.\ Ann. ACM-SIAM Symp. on Discrete Algorithms (SODA)},
  2014.

\bibitem{Franek-al}
P.~Franek, S.~Ratschan, and P.~Zgliczynski.
\newblock Satisfiability of systems of equations of real analytic functions is
  quasi-decidable.
\newblock In {\em \emph{Proc. 36th International Symposium on Mathematical
  Foundations of Computer Science (MFCS)}, LNCS 6907}, pages 315--326.
  Springer, Berlin, 2011.

\bibitem{free-kru}
M.~Freedman and V.~Krushkal.
\newblock Geometric complexity of embeddings in {$\mathbb{R}^d$}.
\newblock Preprint, arXiv:1311.2667, 2013.

\bibitem{Friedm08}
G.~{Friedman}.
\newblock An elementary illustrated introduction to simplicial sets.
\newblock {\em Rocky Mountain J. Math.}, 42(2):353--423, 2012.

\bibitem{polymake}
E.~Gawrilow and M.~Joswig.
\newblock polymake: a framework for analyzing convex polytopes.
\newblock In G.~Kalai and G{.\,M.} Ziegler, editors, {\em Polytopes --
  Combinatorics and Computation}, pages 43--74. Birkh\"auser, Basel, 2000.

\bibitem{GoerssJardine}
P.~G. Goerss and J.~F. Jardine.
\newblock {\em Simplicial homotopy theory}.
\newblock Birkh{\"a}user, Basel, 1999.

\bibitem{GonzalezDiazReal:ComputationCohomologyOperations-2003}
R.~Gonz{{\'a}}lez-D{\'{\i}}az and P.~Real.
\newblock Computation of cohomology operations of finite simplicial complexes.
\newblock {\em Homology Homotopy Appl.}, 5(2):83--93, 2003.

\bibitem{GonRea05}
R.~Gonzalez-Diaz and P.~Real.
\newblock Simplification techniques for maps in simplicial topology.
\newblock {\em J. Symb. Comput.}, 40:1208--1224, October 2005.

\bibitem{Haken:TheorieNormalflaechen-1961}
W.~Haken.
\newblock Theorie der {N}ormalfl\"achen.
\newblock {\em Acta Math.}, 105:245--375, 1961.

\bibitem{Hatcher}
A.~Hatcher.
\newblock {\em Algebraic {T}opology}.
\newblock Cambridge University Press, Cambridge, 2001.
\newblock Electronic version available at
  \url{http://math.cornell.edu/hatcher#AT1}.

\bibitem{HuBook}
S.~Hu.
\newblock {\em Homotopy theory}.
\newblock Academic Press, New York, 1959.

\bibitem{Kochman:Stable-homotopy-groups-of-spheres-1990}
S.~O. Kochman.
\newblock {\em Stable homotopy groups of spheres}, volume 1423 of {\em Lecture
  Notes in Mathematics}.
\newblock Springer-Verlag, Berlin, 1990.
\newblock A computer-assisted approach.

\bibitem{Kochman}
S.~O. Kochman.
\newblock {\em {Stable homotopy groups of spheres. A computer-assisted
  approach.}}
\newblock Lecture Notes in Mathematics 1423. Springer-Verlag, Berlin etc.,
  1990.

\bibitem{Krcal-thesis}
M.~Kr\v{c}\'al.
\newblock {\em Computational Homotopy Theory}.
\newblock PhD thesis, Department of Applied Mathematics, Faculty of Mathematics
  and Physics, Charles University, Prague, 2013.

\bibitem{pKZ1}
M.~Kr\v{c}\'al, J.~Matou\v{s}ek, and F.~Sergeraert.
\newblock Polynomial-time homology for simplicial {Eilenberg--MacLane} spaces.
\newblock {\em J. Foundat. of Comput. Mathematics}, 13:935--963, 2013.
\newblock Preprint, arXiv:1201.6222.

\bibitem{Matousek:BorsukUlam-2003}
J.~Matou{\v{s}}ek.
\newblock {\em Using the {B}orsuk-{U}lam theorem (revised 2nd printing)}.
\newblock Universitext. Springer-Verlag, Berlin, 2007.

\bibitem{MatousekTancerWagner:HardnessEmbeddings-2011}
J.~Matou{\v{s}}ek, M.~Tancer, and U.~Wagner.
\newblock {Hardness of embedding simplicial complexes in $\R^d$}.
\newblock {\em J. Eur. Math. Soc.}, 13(2):259--295, 2011.

\bibitem{Mat-homotopyW1}
J.~Matou\v{s}ek.
\newblock Computing higher homotopy groups is {$W[1]$}-hard.
\newblock {\em Fundamenta Informaticae}, 2014.
\newblock To appear.

\bibitem{Matveev:AlgorithmicTopology-2007}
S.~Matveev.
\newblock {\em Algorithmic topology and classification of 3-manifolds},
  volume~9 of {\em Algorithms and Computation in Mathematics}.
\newblock Springer, Berlin, second edition, 2007.

\bibitem{PetMay67}
J.~P. {May}.
\newblock {\em Simplicial Objects in Algebraic Topology}.
\newblock Chicago University Press, 1967.

\bibitem{MosherTangora:CohomologyOperations-1968}
R.~E. Mosher and M.~C. Tangora.
\newblock {\em Cohomology operations and applications in homotopy theory}.
\newblock Harper \& Row Publishers, New York, 1968.

\bibitem{NabutovskyWeinberger:AlgorithmicUnsolvabilityTrivialityProblemMultidimensionalKnots-1996}
A.~Nabutovsky and S.~Weinberger.
\newblock {Algorithmic unsolvability of the triviality problem for
  multidimensional knots}.
\newblock {\em Comment. Math. Helv.}, 71(3):426--434, 1996.

\bibitem{NabutovskyWeinberger:AlgorithmicAspectsHomeomorphismProblem-1999}
A.~Nabutovsky and S.~Weinberger.
\newblock Algorithmic aspects of homeomorphism problems.
\newblock In {\em Tel Aviv Topology Conference: Rothenberg Festschrift (1998)},
  volume 231 of {\em Contemp. Math.}, pages 245--250. Amer. Math. Soc.,
  Providence, RI, 1999.

\bibitem{Novikov:UndecidabilityWordProblem-1955}
P.~S. Novikov.
\newblock Ob algoritmi\v cesko\u\i\ nerazre\v simosti problemy to\v zdestva
  slov v teorii grupp ({O}n the algorithmic unsolvability of the word problem
  in group theory).
\newblock {\em Trudy Mat. inst. im. Steklova}, 44:1--143, 1955.

\bibitem{Ravenel}
D.~C. Ravenel.
\newblock {\em Complex Cobordism and Stable Homotopy Groups of Spheres (2nd
  ed.)}.
\newblock Amer. Math. Soc., 2004.

\bibitem{Real96}
P.~Real.
\newblock An algorithm computing homotopy groups.
\newblock {\em Mathematics and Computers in Simulation}, 42:461---465, 1996.

\bibitem{RomeroRubioSergeraert}
A.~Romero, J.~Rubio, and F.~Sergeraert.
\newblock {Computing spectral sequences}.
\newblock {\em J. Symb. Comput.}, 41(10):1059--1079, 2006.

\bibitem{RubioSergeraert:ConstructiveAlgebraicTopology-2002}
J.~Rubio and F.~Sergeraert.
\newblock Constructive algebraic topology.
\newblock {\em Bull. Sci. Math.}, 126(5):389--412, 2002.

\bibitem{SergRub-homtypes}
J.~Rubio and F.~Sergeraert.
\newblock Algebraic models for homotopy types.
\newblock {\em Homology, Homotopy and Applications}, 17:139--160, 2005.

\bibitem{SergerGenova}
J.~Rubio and F.~Sergeraert.
\newblock Constructive homological algebra and applications.
\newblock Preprint, arXiv:1208.3816, 2012.
\newblock Written in 2006 for a MAP Summer School at the University of Genova.

\bibitem{Schoen-effectivetop}
R.~Sch{\"o}n.
\newblock {Effective algebraic topology}.
\newblock {\em Mem. Am. Math. Soc.}, 451:63 p., 1991.

\bibitem{Schrij86}
A.~Schrijver.
\newblock {\em Theory of linear and integer programming}.
\newblock John Wiley \& Sons, Inc., New York, NY, USA, 1986.

\bibitem{SergerCouples}
F.~Sergeraert.
\newblock Effective exact couples.
\newblock Preprint, Univ. Grenoble, available at
  \url{http://www-fourier.ujf-grenoble.fr/~sergerar/Papers/}.

\bibitem{Sergeraert:ComputabilityProblemAlgebraicTopology-1994}
F.~Sergeraert.
\newblock The computability problem in algebraic topology.
\newblock {\em Adv. Math.}, 104(1):1--29, 1994.

\bibitem{skopenkov-survey}
A.~B. Skopenkov.
\newblock Embedding and knotting of manifolds in {E}uclidean spaces.
\newblock In {\em Surveys in contemporary mathematics}, volume 347 of {\em
  London Math. Soc. Lecture Note Ser.}, pages 248--342. Cambridge Univ. Press,
  Cambridge, 2008.

\bibitem{smith-mstructures}
J.~R. Smith.
\newblock {m}-{S}tructures determine integral homotopy type.
\newblock Preprint, arXiv:math/9809151v1, 1998.

\bibitem{Soare:ComputabilityDifferentialGeometry-2004}
R.~I. Soare.
\newblock Computability theory and differential geometry.
\newblock {\em Bull. Symbolic Logic}, 10(4):457--486, 2004.

\bibitem{Steenr47}
N.~E. Steenrod.
\newblock Products of cocycles and extensions of mappings.
\newblock {\em Annals of Mathematics}, 48(2):pp. 290--320, 1947.

\bibitem{Steenrod:CohomologyOperationsObstructionsExtendingContinuousFunctions-1972}
N.~E. Steenrod.
\newblock Cohomology operations, and obstructions to extending continuous
  functions.
\newblock {\em Advances in Math.}, 8:371--416, 1972.

\bibitem{Stillwell}
J.~Stillwell.
\newblock {\em {Classical Topology and Combinatorial Group Theory}}.
\newblock Graduate Texts in Mathematics 72. Springer, New York, 2nd edition,
  1993.

\bibitem{Storjohann:NearOptimalAlgorithmsSmithNormalForm-1996}
A.~Storjohann.
\newblock Near optimal algorithms for computing {Smith} normal forms of integer
  matrices.
\newblock In {\em International Symposium on Symbolic and Algebraic
  Computation}, pages 267--274, 1996.

\bibitem{Vokrinek}
L.~Vok\v{r}\'{\i}nek.
\newblock Constructing homotopy equivalences of chain complexes of free
  {${\mathbb{Z}}G$}-modules.
\newblock Preprint, arXiv:1304.6771, 2013.

\bibitem{Whitehead:HomotopyTheory-1978}
G.~W. Whitehead.
\newblock {\em Elements of homotopy theory}, volume~61 of {\em Graduate Texts
  in Mathematics}.
\newblock Springer-Verlag, New York, 1978.

\bibitem{Zomorodian:TopologyComputing-2005}
A.~J. Zomorodian.
\newblock {\em Topology for computing}, volume~16 of {\em Cambridge Monographs
  on Applied and Computational Mathematics}.
\newblock Cambridge University Press, Cambridge, 2005.

\end{thebibliography}

\end{document}